\documentclass[acmsmall,screen]{acmart}
\PassOptionsToPackage{authoryear}{natbib}
\PassOptionsToPackage{sort}{natbib}
\RequirePackage{natbib}
\setcopyright{rightsretained}
\acmDOI{10.1145/3656463}
\acmYear{2024}
\copyrightyear{2024}
\acmSubmissionID{pldi24main-p861-p}
\acmJournal{PACMPL}
\acmVolume{8}
\acmNumber{PLDI}
\acmArticle{233}
\acmMonth{6}
\received{2023-11-16}
\received[accepted]{2024-03-31}

\usepackage{natbib}
\usepackage[framemethod=tikz]{mdframed}
\usepackage{mathpartir}
\usepackage{lipsum}
\usepackage{wrapfig}
\usepackage{mathtools}
\usepackage{array}
\usepackage{makecell}
\usepackage{arydshln}
\usepackage{multirow}
\usepackage{enumitem}
\usepackage{tabularx}
\usepackage{subcaption}
\mdfdefinestyle{linesStyle}{%
  outerlinewidth=0.05mm,
  leftline=false,
  rightline=false,
  topline=true,
  bottomline=true,
  skipabove=5pt,
  skipbelow=5pt,
}
\makeatletter
\renewcommand{\paragraph}{%
  \@startsection{paragraph}{4}%
  {\z@}{1.25ex \@plus 0.5ex \@minus .2ex}{-1em}%
  {\normalfont\normalsize\bfseries}%
}
\makeatother

\usepackage{tikz}
\usetikzlibrary{calc}
\usepackage{moresize}
\usepackage[frozencache,cachedir=.]{minted}
\usepackage{esint}
\usepackage{textcomp}
\usepackage{pifont}% http://ctan.org/pkg/pifont
\newcommand{\RR}{\mathbb{R}}

\newcommand\doubleplus{+\kern-1.3ex+\kern0.8ex}
\newcommand\mdoubleplus{\ensuremath{\mathbin{+\mkern-5mu+}}}

\usepackage{tabularx}

\newcommand{\trace}[0]{\text{Trace}}

\newcommand{\ter}[0]{t}

\newcommand{\sem}[1]{\llbracket #1\rrbracket}
\newcommand{\type}[0]{\tau}
\newcommand{\var}[0]{x}

\newcommand{\return}[0]{\textbf{return}}

\newcommand{\marginal}[0]{\textbf{marginal}}
\newcommand{\normalize}[0]{\textbf{normalize}}

\newcommand{\op}[0]{op}

\newcommand{\pop}[1]{\ensuremath{\partial_i}\op}

\newcommand{\ad}[1]{\ensuremath{\mathcal{D}\{#1\}}}
\newcommand{\dens}[0]{{D}}

\newcommand{\eRR}[0]{\widetilde{\RR}}
\newcommand{\lin}[0]{\textbf{in}}
\newcommand{\llet}[0]{\textbf{let}}

\newcommand{\str}[0]{\text{Str}}

\newcommand{\sample}[0]{\textbf{sample}}

\newcommand{\haskdo}[0]{\textbf{do}}
\newcommand{\normal}[0]{\textbf{normal}}
\newcommand{\flip}[0]{\textbf{flip}}

\newcommand{\observe}[0]{\textbf{observe}}

\newcommand{\score}[0]{\textbf{score}}

\usepackage[export]{adjustbox}

\begin{document}

\title{Probabilistic Programming with Programmable\texorpdfstring{\newline}{} Variational Inference}

\author{McCoy R. Becker}
\orcid{0009-0000-1930-8150}
\affiliation{%
  \institution{MIT}
  \city{Cambridge}
  \country{USA}
}
\email{mccoyb@mit.edu}
\authornote{Equal contribution.}

\author{Alexander K. Lew}
\orcid{0000-0002-9262-4392}
\affiliation{%
  \institution{MIT}
  \city{Cambridge}
  \country{USA}
}
\email{alexlew@mit.edu}
\authornotemark[1]

\author{Xiaoyan Wang}
\orcid{0000-0001-7058-4679}
\affiliation{%
  \institution{MIT}
  \city{Cambridge}
  \country{USA}
}
\email{xyw@mit.edu}

\author{Matin Ghavami}
\orcid{0000-0003-3052-7412}
\affiliation{%
  \institution{MIT}
  \city{Cambridge}
  \country{USA}
}
\email{mghavami@mit.edu}

\author{Mathieu Huot}
\orcid{0000-0002-5294-9088}
\affiliation{%
  \institution{MIT}
  \city{Cambridge}
  \country{USA}
}
\email{mhuot@mit.edu}

\author{Martin C. Rinard}
\orcid{0000-0001-8095-8523}
\affiliation{%
  \institution{MIT}
  \city{Cambridge}
  \country{USA}
}
\email{rinard@mit.edu}

\author{Vikash K. Mansinghka}
\orcid{0000-0003-2507-0833}
\affiliation{%
  \institution{MIT}
  \city{Cambridge}
  \country{USA}
}
\email{vkm@mit.edu}

\renewcommand{\shortauthors}{M. R. Becker, A. K. Lew, X. Wang, M. Ghavami, M. Huot, M. C. Rinard, V. K. Mansingkha}
\renewcommand{\shorttitle}{Probabilistic Programming with Programmable Variational Inference}

\begin{abstract}
    
Compared to the wide array of advanced Monte Carlo methods supported by modern probabilistic programming languages (PPLs), PPL support for \textit{variational inference} (VI) is less developed: users are typically limited to a predefined selection of variational objectives and gradient estimators, which are implemented monolithically (and without
formal correctness arguments) in PPL backends.
In this paper, we propose a more modular approach to supporting variational inference in PPLs, based on compositional program transformation.
%We present (1) a probabilistic programming language for defining models, variational families, and compositional strategies for propagating gradients, (2) a differentiable programming language for defining variational objectives and (3) an automatic differentiation algorithm that differentiate these variational objectives, yielding provably unbiased gradient estimators for use during optimization. Models and variational families from the first language are automatically compiled into new differentiable functions that can be called from the second language, for estimating densities and expectations.
In our approach, variational objectives are expressed as programs, that may employ first-class constructs for computing \textit{densities of} and \textit{expected values under} user-defined models and variational families. We then transform these programs systematically into unbiased gradient estimators for optimizing the objectives they define. %We then extend the automatic differentiation algorithm of~\citet{lew_adev_2023} to compute provably unbiased gradient estimates for optimizing these objectives.
Our design enables modular reasoning about many interacting concerns, including automatic differentiation, density accumulation, tracing, and the application of unbiased gradient estimation strategies.
Additionally, relative to existing support for VI in PPLs, our design increases expressiveness along three axes: (1) it supports an open-ended set of user-defined variational objectives, rather than a fixed menu of options; (2) it supports a combinatorial space of gradient estimation strategies, many not automated by today's PPLs; and (3) it supports a broader class of models and variational families, because it supports constructs for approximate marginalization and normalization (previously introduced only for Monte Carlo inference). %This makes it possible to concisely express many models, variational families, objectives, and gradient estimators from the machine learning literature, including importance-weighted autoencoders (IWAE), hierarchical variational inference (HVI), and reweighted wake-sleep (RWS). 
% We establish the correctness of our gradient automation 
% using logical relations over the semantics of $\lambda \nabla_\text{SP}$, a new calculus synthesized by combining a calculus for probabilistic programming with density estimation ($\lambda_{\text{SP}}$) and a calculus for unbiased gradient estimation of expected values (ADEV). 
We implement our approach in an extension to the Gen probabilistic programming system (\href{https://gen.dev/genjax/vi}{\texttt{genjax.vi}}, implemented in JAX), and evaluate our automation on several deep generative modeling tasks, showing minimal performance overhead vs. hand-coded implementations and performance competitive with well-established open-source PPLs.

\end{abstract}




\maketitle

\section{Introduction}
\label{sec:intro}

\textit{Variational inference} (VI) is a popular approach to two fundamental probabilistic modeling tasks:
\begin{itemize}[leftmargin=*]
\item \textbf{Fitting probabilistic models to data.} Given a family of joint probability distributions $\mathcal{P} = \{P_\theta(x, y) \mid \theta \in \mathbb{R}^n\}$ defined over \textit{latent variables} $x$ and \textit{observed variables} $y$, find the one that best explains an observed dataset $\mathbf{y}$. For example, writing $p_\theta$ for the probability density function of $P_\theta$, we may be interested in finding $\theta \in \mathbb{R}^n$ that maximizes the \textit{marginal likelihood} \begin{equation}
    p_\theta(\mathbf{y}) = \int_X p_\theta(x, \mathbf{y}) dx.\label{eqn:marginal_likelihood}
\end{equation}
\item \textbf{Approximating intractable posterior distributions.} For a particular probabilistic model $P_\theta(x, y)$, find the best approximation to the (usually intractable) posterior distribution $P_\theta(x \mid \mathbf{y})$, from a class $\mathcal{Q} = \{Q_\phi(x) \mid \phi \in \mathbb{R}^m\}$ of tractable approximations (the \textit{variational family}). For example, again using lower-case letters for probability density functions, we may be interested in finding $\phi$ that minimizes the \textit{reverse KL divergence} \begin{equation}D_{KL}(Q_\phi(x) \,||\, P_\theta(x \mid \mathbf{y})) = -\mathbb{E}_{x \sim Q_\phi}\!\left[\log \frac{p_\theta(x \mid \mathbf{y})}{q_\phi(x)}\right].\label{eqn:kl_divergence}\end{equation}
\end{itemize}

Practitioners often aim to solve both these tasks at once, simultaneously fitting a probabilistic model and approximating its posterior distribution. To do so, one defines a \textit{variational objective} $\mathcal{F} : \mathcal{P} \times \mathcal{Q} \to \mathbb{R}$, mapping particular distributions $P_\theta$ and $Q_\phi$ to a scalar \textit{loss} (or \textit{reward}). For example, one common choice is the \textit{evidence lower bound}, or ELBO: \begin{equation}\mathbf{ELBO}(P, Q) := \mathbb{E}_{x\sim Q}[\log p(x, \mathbf{y}) - \log q(x)] = \log p(\mathbf{y}) - D_{KL}(Q(x) || P(x \mid \mathbf{y}))\label{eqn:elbo}\end{equation}
As the decomposition on the right-hand side suggests, maximizing the ELBO simultaneously maximizes the (log) marginal likelihood of the data $\mathbf{y}$ and minimizes the KL divergence of the posterior approximation $Q$ to the posterior. Besides the ELBO, researchers have also proposed many alternative objectives~\cite{dempster_maximum_1977,hinton_wake-sleep_1995,agakov_auxiliary_2004,bornschein_reweighted_2015,burda_importance_2016,ranganath_hierarchical_2016,rainforth_tighter_2018,sobolev_importance_2019,malkin2022gflownets}, which formalize the two goals of \textit{fitting models} and \textit{approximating posteriors} differently (e.g., by using divergences other than the KL).
Once a variational objective has been defined, practitioners aim to find parameters $(\theta, \phi)$ that maximize (or minimize) $\mathcal{F}(P_\theta, Q_\phi)$. There are many possible approaches to performing this optimization. The most popular methods rely on \textit{gradients} $\nabla_{(\theta, \phi)} \mathcal{F}(P_\theta, Q_\phi)$ of the objective\textemdash or more often, unbiased stochastic estimates of these gradients. Designing and implementing algorithms for estimating these gradients, with sufficiently low variance and computational expense, is the key roadblock on the path from  \textit{defining} a variational inference problem to solving it.

Indeed, although variational inference algorithms have found widespread adoption in Bayesian statistics~\citep{fox2012tutorial,kucukelbir2017automatic,blei2017variational,hoffman2013stochastic,blei2006variational} and in probabilistic deep learning~\cite{kingma_auto-encoding_2022,pu2016variational,kingma2021variational,vahdat2020nvae,malkin2022gflownets}, implementing variational inference algorithms by hand remains a tedious and error-prone endeavor. The key mathematical ingredients specifying a variational inference problem\textemdash $\mathcal{P}$, $\mathcal{Q}$, and $\mathcal{F}$\textemdash are typically not represented directly in code; rather, the practitioner must:

\begin{enumerate}[leftmargin=*]
% \item Define a \textit{variational family}: a parametric class $\{Q_\theta \mid \theta \in \mathbb{R}^n\}$ of distributions to search over.
% \item Define a real-valued objective function $Q \mapsto \mathcal{L}(Q)$ to optimize. For example, $\mathcal{L}$ might be the KL divergence of $Q$ to another distribution $P$ ($\mathcal{L}(Q) = \mathbb{E}_{x \sim Q}[\log \frac{p(x)}{q(x)}]$).
\item use algebra, probability theory, and calculus to derive a \textit{gradient estimator}: a way to rewrite the gradient $\nabla_{(\theta, \phi)} \mathcal{F}(P_\theta, Q_\phi)$ as an expectation $\mathbb{E}_{y \sim M_{(\theta,\phi)}}[f_{(\theta, \phi)}(y)]$, for some family of distributions $M$ and some family of functions $f$; and then
\item write code to sample $y \sim M_{(\theta,\phi)}$ and evaluate $f_{(\theta, \phi)}(y)$\textemdash an unbiased estimate of $\nabla_{(\theta,\phi)} \mathcal{F}(P_\theta, Q_\phi)$.
%\item Run stochastic gradient descent, using unbiased gradient estimates obtained by sampling $x \sim M_\theta$ and evaluating $f_\theta(x)$. 
\end{enumerate}
It is often non-trivial to ensure that $M$ and $f$ are faithfully implemented, and that the math used to derive them in the first place is error-free. Small changes to $\mathcal{F}$, $\mathcal{P}$, or $\mathcal{Q}$, or to the gradient estimation strategy employed in step (1), can require large, non-local changes to $M$ and $f$, and in implementing these changes, it is easy to introduce hard-to-detect bugs. When optimization fails, it is often unclear whether the problem is with the math, the code, or just the hyperparameters.
%\vspace{-2mm}

\paragraph{Automation via Probabilistic Programming.} Reflecting the importance of variational inference, many probabilistic programming languages (PPLs) (especially ``deep'' PPLs, such as Pyro~\cite{bingham_pyro_2018}, Edward~\cite{tran_2017_deep,tran2018simple}, and ProbTorch \cite{stites_learning_2021}), feature varying degrees of automation for VI workflows. In these languages, users can express both models $\mathcal{P}$ and variational families $\mathcal{Q}$  as probabilistic programs; the system then automates the estimation of gradients for a pre-defined set of supported variational objectives $\mathcal{F}$. This design significantly lowers the cost of implementing and iterating on variational inference algorithms, but several pain points remain:  %selected from a pre-defined menu of options. %Early implementations relied on black-box variational inference~\cite{ranganath_black_2014}, But recent years have seen many improvements: for example, Pyro~\cite{bingham_pyro_2018} supports an impressive roster of roster of variational objectives and gradient estimation strategies. But they do so via a {menu-based} approach: users can freely specify $\mathcal{P}$ and $\mathcal{Q}$ as probabilistic programs, but then must choose from a pre-defined set of options for the variational objective $\mathcal{F}$. Each objective in turn supports only limited customization of the strategy for estimating its gradient. This approach has several drawbacks:
\begin{itemize}[leftmargin=*]
\item \textbf{Incomplete coverage.} Existing PPLs offer limited or no support for many variational objectives, including forward KL objectives~\cite{naesseth_markovian_2020}; hierarchical, nested, or recursive variational objectives~\cite{ranganath_hierarchical_2016,zimmermann_nested_2021,lew_recursive_2022}; symmetric divergences~\cite{domke_easy_2021}; trajectory-balance objectives~\cite{malkin2022gflownets}; SMC-based objectives~\cite{maddison_filtering_2017,gu_neural_2015,li_neural_2023,naesseth2018variational}; and others. Today's PPLs also do not automate many powerful gradient estimation strategies, for example those based on measure-valued differentiation~\citep{mohamed2020monte}.
%\item \textbf{Incomplete coverage (gradient estimators).} Even when the user's desired variational objective is supported, it may not implement the gradient estimation strategy that the user wants to apply. For example, Pyro supports importance-weighted objectives~\cite{burda_importance_2016} (\texttt{RenyiELBO}), but uses high-variance gradient estimators~\cite{williams_simple_1992} to handle discrete variables in the variational family. Variance reduction strategies like data-dependent baselines or enumeration, although implemented as parts of other menu options, cannot be applied with importance-weighted objectives.
\item \textbf{Duplicative engineering effort.} For PPL maintainers, supporting new gradient estimation strategies or language features requires separately introducing the same logic into the implementations of multiple variational objectives. Because this engineering effort is non-trivial, many capabilities are not uniformly supported. For example, as of this writing, Pyro's \texttt{ReweightedWakeSleep} objective~\cite{le_revisiting_2019} does not support minibatching, even though other objectives do. As another example, variance reduction strategies such as data-dependent baselines and enumeration of discrete latents are implemented for the ELBO in Pyro, but not, e.g., for the importance-weighted ELBO.
\item \textbf{Difficulty of reasoning.} The monolithic implementations of each variational objective's gradient estimation logic intertwine various concerns, including log density accumulation, automatic differentiation, gradient propagation through stochastic choices, and variance reduction logic. This can make it difficult to reason about correctness. Indeed, while the community has made tremendous progress in understanding the compositional correctness arguments of an increasingly broad class of Monte Carlo inference methods %\footnote{Monte Carlo methods are another approach to approximating intractable distributions, such as the posteriors of probabilistic inference problems.} 
for probabilistic programs~\cite{scibior_denotational_2017,scibior_functional_2018,lew_probabilistic_2023,lunden2021correctness,borgstrom2016lambda}, pioneering work on correctness for variational inference~\cite{lee_towards_2019,lee_smoothness_2023,li2023type} has generally focused on specific properties (e.g., smoothness and absolute continuity) in somewhat restricted languages, and not to end-to-end correctness of gradient estimation for variational inference. 
\end{itemize}

\paragraph{This Work.} In this paper, we present a highly modular, programmable approach to supporting variational inference in PPLs. In our approach, all three ingredients of the variational inference problem\textemdash $\mathcal{P}$, $\mathcal{Q}$, and $\mathcal{F}$\textemdash are encoded as programs in expressive probabilistic languages, which support compositional annotation for specifying the desired mix of gradient estimation strategies. We then use a sequence of modular program transformations\textemdash each of which we independently prove correct\textemdash to construct unbiased gradient estimators for the user's variational objective.%\\\vspace{-1mm}%\footnote{A critical element to making this work is the recently introduced ADEV framework for modular gradient estimation~\cite{lew_adev_2023}: we specialize and extend ADEV to the variational inference setting. Our extensions include new constructs for differentiably tracing and estimating the densities of models and variational families (crucial for expressing virtually all variational objectives), a \textit{reverse-mode} ADEV implementation that efficiently computes gradient estimates in one pass (instead of one pass per parameter), and the first GPU-accelerated implementation of ADEV-style AD.}  %A particularly important extension  \textit{densities} of probabilistic programs, required by virtually all variational objectives. ADEV works in \textit{forward-mode}, and thus is hopelessly inefficient for optimizing neural networks with many parameters. I, and variational objectives $\mathcal{F}$ are typically expressed as functions of both expected values and density functions. The ADEV algorithm is a \textit{forward-mode} 

\paragraph{Contributions.} This paper contributes:
\begin{itemize}[leftmargin=*]
    \item \textbf{Languages for models, variational families, and variational objectives}: We present an expressive language for models and variational families (\S\ref{sec:gen-lang}), similar to Gen, ProbTorch, or Pyro, along with an expressive differentiable language for variational objective functions (\S\ref{sec:adev-lang}).%, but with support for compositional annotation of gradient estimation strategies.
    %\item \textbf{Objective language}: In contrast to the menu-based approach of existing PPLs, we present an expressive differentiable language for variational objectives (\S\ref{sec:adev-lang}), by extending ADEV~\cite{lew_adev_2023} with constructs for tracing and estimating densities of probabilistic programs (\S\ref{sec:gen-transformations}).
    %\item \textbf{Modular design}: 
    \item \textbf{Flexible, modular automation}: We automate a broad class of unbiased gradient estimators for variational objectives (\S\ref{sec:adev}). New primitive gradient estimation strategies can be added modularly with just a few lines of code, without deeply understanding system internals (Appx.~\ref{appdx:extensibility}).
    \item \textbf{Formalization}: We formalize our approach as a sequence of composable program transformations (\S\ref{sec:gen-transformations}-\ref{sec:vi}) of simply-typed $\lambda$-calculi for probabilistic programs (\S\ref{sec:syntax}), and prove the unbiasedness of gradient estimation (under mild technical conditions) by logical relations (\S\ref{sec:correctness}). Ours is the first formal account of variational inference for PPLs that accounts for the interactions between tracing, density computation, gradient estimation strategies, and automatic differentiation.
    \item \textbf{System}: We contribute \texttt{genjax.vi}, a performant, GPU-accelerated implementation of our approach in JAX~\cite{frostig2018compiling}, which also extends our formal modeling language with constructs for marginalization and normalization (\S\ref{sec:sp})~\cite{lew_probabilistic_2023}. We also contribute concise, pedagogical Haskell and Julia versions. Our implementations are the first to feature \textit{reverse-mode} variants of the ADEV algorithm for modularly differentiating higher-order probabilistic programs~\cite{lew_adev_2023}.\footnote{System and code available at  \url{https://gen.dev/genjax/vi}}
    \item \textbf{Empirical evaluation}: We evaluate \texttt{genjax.vi} on several benchmark tasks, including the challenging Attend-Infer-Repeat model~\cite{eslami_attend_2016}. We find that \texttt{genjax.vi} makes it possible to encode new gradient estimators that converge faster and to better solutions than Pyro's estimators. We also show, for the first time, that a version of ADEV can be scalably and performantly implemented, to deliver competitive performance on realistic probabilistic deep learning workloads.\footnote{\citet{lew_adev_2023} present only a toy Haskell implementation of  forward-mode ADEV, and report no experiments.}
\end{itemize}

Programmability is sometimes seen as being at odds with automation. We emphasize that this paper \textit{expands} the automation provided by the system, relative to existing VI support in PPLs, by automating a combinatorial space of gradient estimators for {arbitrary objectives} specified as programs (rather than the handful of objectives and estimators supported by existing PPLs).%, and supporting a combinatorial space of gradient estimation strategies, specified using compositional annotations on generative programs. %In \S\ref{sec:evaluation}, we show that the expanded class of gradient estimators that our system automates can out-perform standard estimators supported by existing systems.

% \noindent \textbf{Contributions.} This paper contributes:
% \begin{itemize}[leftmargin=*]
%     \item The first formalization for variational inference in a higher order probabilistic language, including \textit{compositional correctness guarantees} covering usage of a wide variety of variational objectives and gradient estimator strategies. 
%     \item Our language design exposes automation for \textit{expressive} variational objectives, including those defined by utilizing \textbf{marginal} to construct marginal distributions.
%     \item An \textbf{evaluation of our approach} on a set of case studies designed to test the performance and expressivity implications of our design. We establish several properties of our design: the overhead which our approach incurs compared to hand coded estimators is negligible, and our performance comparisons with well-supported variational inference PPLs are highly competitive. Our design also supports variational objectives which are not currently supported by any PPL - without losing in runtime cost.
% \end{itemize}

\section{Overview}
\label{sec:overview}
Fig.~\ref{fig:pipeline-exact} illustrates the workflow of a typical user of our system for modular variational inference: %our overall approach to modular variational inference is illustrated in Fig.~\ref{fig:pipeline-exact}. Before walking through a more concrete example, we first briefly summarize the typical user workflow:
%In this section, we give an overview of our approach to modular variational inference, illustrated in Fig.~\ref{fig:pipeline-exact}. %The key technical challenge in variational inference is the computation of \textit{unbiased gradient estimates} for the \textit{variational objective} that the user wishes to optimize. 
%We describe our approach at a high level, and then work through a toy example in detail (Figs.~\ref{fig:example_transcript} and~\ref{fig:objectives}). Along the way, we will also review several key ideas in variational inference. %First, we describe the basic workflow for users: % based around a three-step workflow:

\begin{figure}
    \setlength{\belowcaptionskip}{-10pt}
    \includegraphics[width=0.99\linewidth]{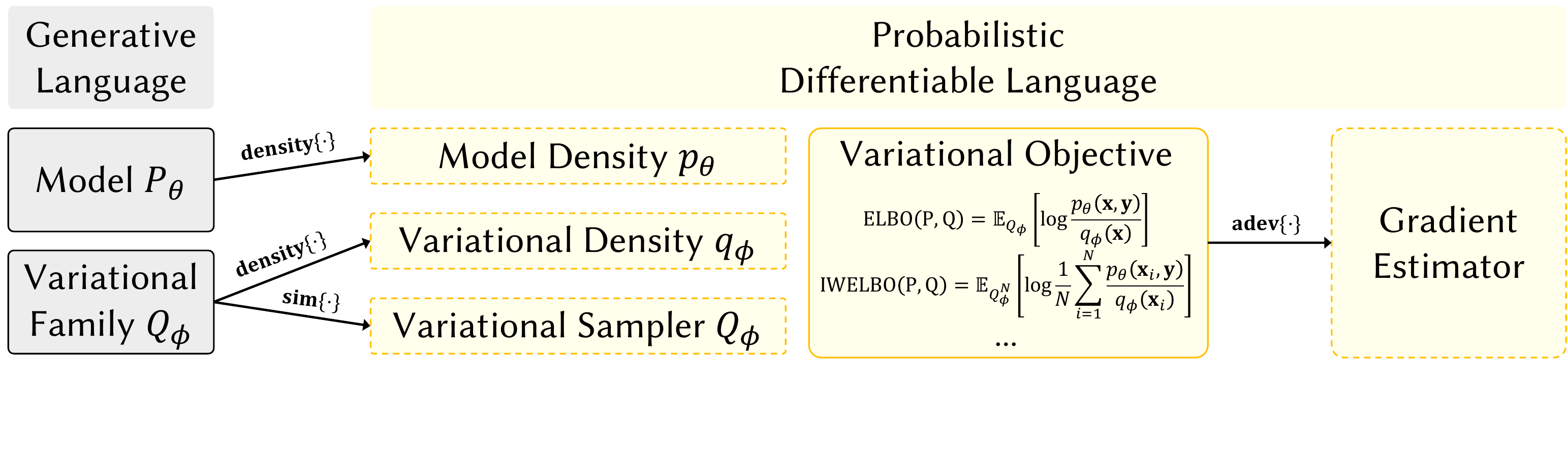}
    \vspace{-7mm}
    \caption{We compose multiple program transformations to automate the construction of unbiased gradient estimators for variational inference. The user begins by writing programs in the generative language %$\lambda_{\text{Gen}}$ 
    (gray) to encode a model and variational family. These programs are compiled into procedures for density evaluation and simulation in the differentiable language %$\lambda_{\text{ADEV}}$ 
    (yellow). These automated procedures can be used to concisely define a variational objective. %e.g., the ELBO $(\theta, \phi) \mapsto \mathbb{E}_{x \sim Q_\phi}\left[\log \frac{p_\theta(x)}{q_\phi(x)}\right]$, also in %$\lambda_{\text{ADEV}}$
    We can then apply the ADEV differentiation algorithm to automatically construct a \textit{gradient estimator}, which unbiasedly estimates gradients of the variational objective. Solid outlines indicate user-written programs, whereas dashed outlines indicate automatically constructed programs.}
    \Description{An illustration of our pipeline.}
    \label{fig:pipeline-exact}
\end{figure}

\begin{itemize}[leftmargin=*]
    \item \textbf{Model and variational programs.} The user begins by writing two probabilistic programs in our \textit{generative language}: a \textit{model program} and a \textit{variational program}. The generative language is a trace-based PPL that resembles Pyro~\cite{bingham_pyro_2018}, ProbTorch~\cite{stites_learning_2021}, and Gen~\cite{cusumano-towner_gen_2019}. The model program encodes a family of joint probability distributions $P_\theta({x}, {y})$, %The key tasks that variational inference can help users solve are: (1) approximating the posterior distribution $P_\theta(x \mid \mathbf{y})$, and (2) finding model parameters $\theta$ that maximize the marginal likelihood $P_\theta(\mathbf{y})$ of the observed data.
    and the variational program encodes a family of distributions $Q_\phi({x})$, possible approximations to the posterior $P_\theta({x} \mid \mathbf{y})$ for data $\mathbf{y}$.
 %   \item \textbf{Model and variational family.} The user implements a \textit{model program} and a \textit{variational family} (also known as a \textit{guide program}) using the \textit{generative language}. The generative language is a trace-based PPL that resembles Pyro~\cite{bingham_pyro_2018}, ProbTorch~\cite{stites_learning_2021}, and Gen~\cite{cusumano-towner_gen_2019}. %Programs in the language can be compositionally annotated with information about what \textit{gradient estimation strategy} should be applied. %but goes beyond them in two ways: it features new, first-class constructs for \textit{normalizing} and \textit{marginalizing} probabilistic programs, and allows programs to be compositionally annotated with information about what \textit{gradient estimation strategy} should be applied to different components. %Each primitive distribution comes in several varieties, each equipped with a different gradient estimation strategy for that primitive.

    \item \textbf{Objective function.} The user now seeks to find values of $(\theta, \phi)$ that simultaneously (1) fit $P_\theta$ to the data $\mathbf{y}$, and (2) make $Q_\phi$ close to the posterior $P_\theta({x} \mid \mathbf{y})$. To make these informal desiderata precise, the user defines an \textit{objective function}, using our \textit{differentiable  language}. %A standard choice is the \textit{evidence lower bound}, or ELBO: $$\mathbf{ELBO}(P, Q) = \mathbb{E}_{\mathbf{x} \sim Q}\left[\log P(\mathbf{x}, \mathbf{y}) - \log Q(\mathbf{x})\right] = \log P(\mathbf{y}) - KL(Q(\mathbf{x}) || P(\mathbf{x} \mid \mathbf{y})).$$ As the decomposition on the right-hand side suggests, maximizing this objective function encourages $P$ to assign high probability to the observed data $\mathbf{y}$, and $Q$ to be close to $P(\mathbf{x} \mid \mathbf{y})$ in KL divergence. The ELBO is not, however, the only choice of objective, and many alternatives have been proposed~\cite{dempster_maximum_1977,hinton_wake-sleep_1995,agakov_auxiliary_2004,bornschein_reweighted_2015,burda_importance_2016,ranganath_hierarchical_2016,rainforth_tighter_2018,sobolev_importance_2019,malkin2022gflownets}.
    %To help users explore a wide range of possible objective functions, we provide a \textit{differentiable probabilistic language} for expressing them. 
    This language features constructs for taking expectations with respect to, and evaluating densities of, generative language programs, making it easy to concisely express objectives like the ELBO (Eqn.~\ref{eqn:elbo}).
    
    \item \textbf{Gradient estimator.} The final step is to optimize the objective function via stochastic gradient ascent. %The key difficulty is that the objective function is usually an expected value with respect to some distribution, e.g. $Q_\phi$, that depends on the parameters being optimized. Practitioners manually employ some combination of \textit{gradient estimation strategies}~\cite{mohamed2020monte}, not all of which are applicable to every model and variational family, to compute \textit{unbiased estimates} of the gradient of the objective function. %\footnote{Systems like Pyro, ProbTorch, and Gen can automate particular gradient estimation strategies for particular variational objectives, but do not guarantee unbiasedness, and even when they are unbiased, the automated gradient estimators may have excessive variance, hindering training.} 
    We construct unbiased gradient estimators for the user's objective function with an extended version of the ADEV algorithm~\cite{lew_adev_2023}. Users can rapidly explore a combinatorial space of estimation strategies with compositional annotations on their generative programs, to navigate tradeoffs between the variance and the computational cost of the automated gradient estimator. 
\end{itemize}

%We begin by introducing the syntax and semantics of our two formal languages, $\lambda_{\text{Gen}}$ and $\lambda_{\text{ADEV}}$ (\S\ref{sec:syntax}). We then present our first program transformations, for converting $\lambda_{\text{Gen}}$ probabilistic programs into differentiable $\lambda_{\text{ADEV}}$ procedures for simulation and density evaluation (\S\ref{sec:gen-transformations}). We next show how those procedures can be concisely composed into larger $\lambda_{\text{ADEV}}$ programs encoding common objective functions for variational inference (\S\ref{sec:objectives}), and conclude by presenting the ADEV transformation that derives unbiased gradient estimators for those objectives (\S\ref{sec:adev}).

% \input{fig/example_transcript/transcript}
\begin{figure}
    \includegraphics[width=\textwidth]{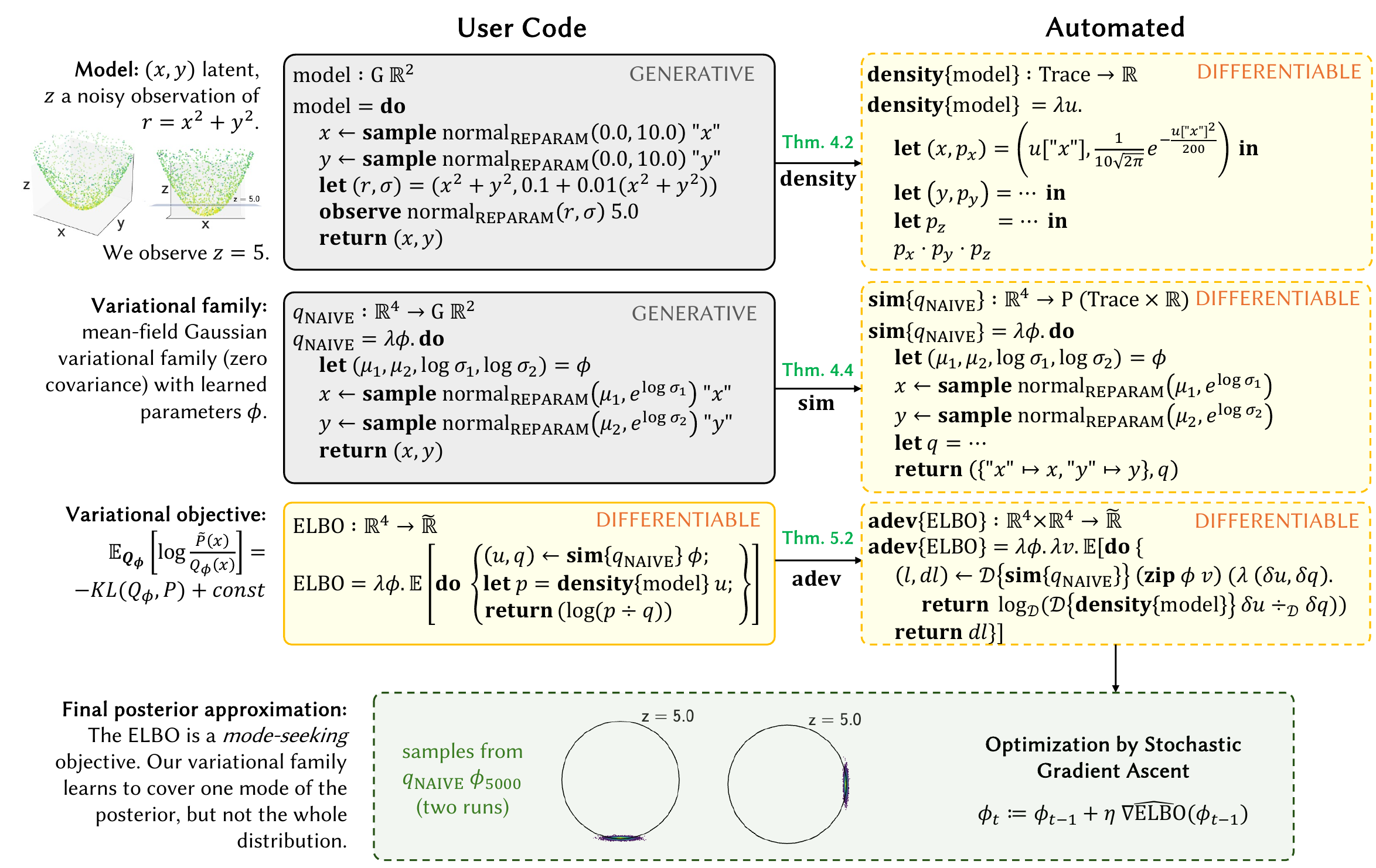}
    \caption{An illustration of our modular approach to automating variational inference, on a toy example. \textbf{(Top)} Users write \textit{generative} code to define a \textit{model} and a \textit{variational family}. Automated program transformations, formalized and proven correct in \S\ref{sec:gen-transformations}, compile \textit{differentiable} code for evaluating densities and simulating traces. \textbf{(Middle)} Users write a program in the differentiable language to define a variational objective, in this case the \textit{evidence lower bound} (ELBO). This code may invoke compiled simulators and density evaluators for generative programs. The \textbf{adev} transformation automates an unbiased gradient estimator for the objective. \textbf{(Bottom)} The gradients are used for optimization, training the variational family to approximate the posterior.}
    \label{fig:example_transcript}
    \vspace{-3mm}
\end{figure}

\paragraph{Example.} To make this concrete, consider the toy problem illustrated in Fig. ~\ref{fig:example_transcript}, which we seek to solve by training a variational approximation to the posterior. We go through the following steps:
\begin{itemize}[leftmargin=*]
\item \textbf{Define a model}. Our model encodes a generative process for points $(x, y, z)$ around a 3D cone. We use $\sample$ to sample latents $x$ and $y$ with string-valued \textit{names}, and $\observe$ to condition on the observation that $z=5.0$. Our goal is to infer $(x, y)$ consistent with this observation. 

\item \textbf{Define a variational family}. In the second panel of Fig.~\ref{fig:example_transcript}, we construct a \textit{variational family}, a parametric family of possible approximations to the posterior distribution. Our variational inference task will be to learn parameters that maximize the quality of the approximation. Our $\texttt{q}_\texttt{NAIVE}$ is a \textit{mean-field} variational approximation, i.e., it generates $x$ and $y$ independently. Note that primitive distributions (here, \texttt{normal}) are annotated with gradient estimation strategies (here, \texttt{REPARAM}) for propagating derivative information through the corresponding primitive.\footnote{For a primitive distribution $\mu_\theta$ over values of type $X$, parameterized by arguments $\theta \in \RR^n$, a gradient estimation strategy for $\mu_\theta$ is an approach to unbiasedly estimating $\nabla_\theta\mathbb{E}_{x \sim \mu_\theta}[f(\theta, x)]$ for functions $f : \RR^n \times X \to \RR$. ADEV composes these primitive estimation strategies into composite strategies for estimating gradients of variational objectives. Supported strategies vary by primitive; for the Normal distribution, for instance, they include \texttt{REPARAM}, \texttt{MEASURE-VALUED}, and \texttt{REINFORCE}, corresponding to different approaches to gradient estimation from the literature. See~\S\ref{sec:adev}.}

% \item \textbf{Compile the model and guide programs}
% Given model and guide programs in $\lambda_\text{Gen}$, our program transformation automation compiles procedures for \textit{density evaluation} ($\textbf{density}\{\text{model}\}$) and \textit{trace simulation} ($\textbf{sim}\{\text{guide}\}$) (with density evaluation) (Section \ref{sec:sp}). The compiled implementation of these procedures is in a new language $\lambda_\text{ADEV}$, a language which supports automatic differentiation.

\item \textbf{Define a variational objective}. We now use the differentiable language to define the objective function we wish to optimize. Three constructs are especially useful: (1) \textbf{density}, which computes or estimates densities of probabilistic programs; (2) \textbf{sim}, which generates a pair $(x, w)$ of a sample and its density from a probabilistic program; and (3) $\mathbb{E}$, which takes the expected value of a stochastic procedure. We use them together to implement the ELBO objective from Eqn.~\ref{eqn:elbo}.
% With a differentiable density evaluation procedure $\textbf{density}\{P\}$ for the model program (we shorten to $P$) and a trace simulation procedure $\textbf{sim}\{Q\}$ for the guide program (shortened to $Q$), we can write a variational objective. Using these procedures and $\lambda_\text{ADEV}$, we can readily construct the \textit{evidence lower bound} (ELBO) objective common to variational inference:

% \noindent By virtue of the fact that the ELBO program is expressed in $\lambda_\text{ADEV}$, and is a term of type $\mathbb{R}^{N + M} \rightarrow \widetilde{\mathbb{R}}$, we can apply ADEV's automatic differentiation transformation to produce an unbiased gradient estimator for the program.

\item \textbf{Perform stochastic optimization.} Our objective is compiled into an \textit{unbiased estimator} of its \textit{gradient} with respect to the input parameters $\phi$. We can then apply stochastic optimization algorithms, such as stochastic gradient ascent and ADAM. The bottom of Fig.~\ref{fig:example_transcript} illustrates samples from $\texttt{q}_\texttt{NAIVE}$ after training. Because the ELBO minimizes the \textit{reverse} (or \textit{mode-seeking}) KL divergence, our variational approximation learns to hug one edge of the circle-shaped posterior.
\end{itemize}

\begin{figure}
    \includegraphics[width=\textwidth]{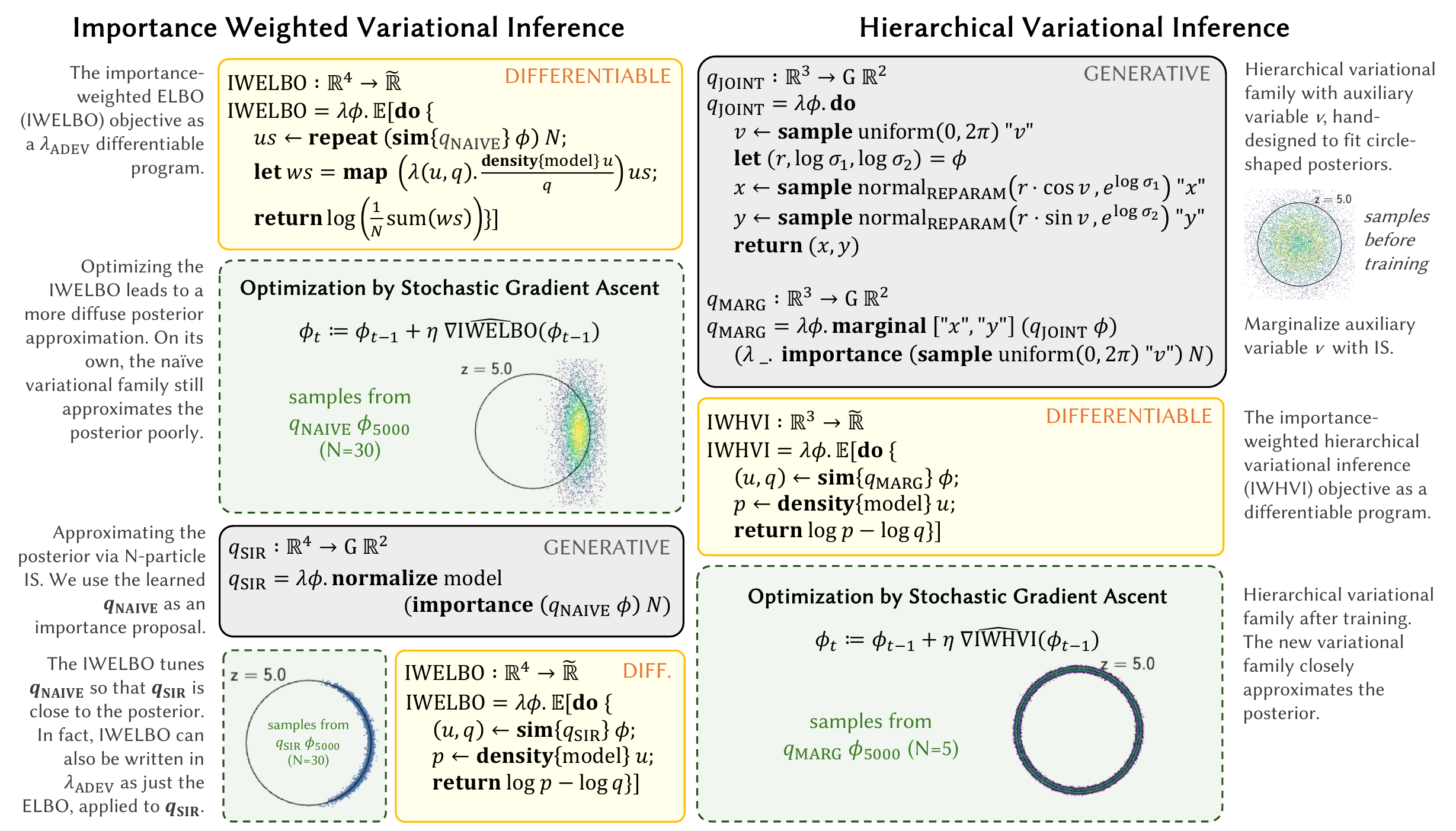}
    \caption{With programmable VI, users can define their own variational objectives, and use new modeling language features to program more expressive models and variational families. Here, we apply importance-weighted VI~\citep{burda_importance_2016} and hierarchical VI~\citep{ranganath_hierarchical_2016,sobolev_importance_2019} to the toy problem from Fig.~\ref{fig:example_transcript}.}
    \label{fig:objectives}
\end{figure}
\paragraph{Better Inference with Programmable Objectives.} To better fit the whole posterior, we will need either a new objective function or a more expressive variational family. Our system's programmability allows users to quickly iterate within a broad space of VI algorithms. In Fig.~\ref{fig:objectives}, we illustrate two possible strategies for improving the posterior approximation:

% For example, the importance-weighted ELBO (IWELBO) objective in Fig.~\ref{fig:objectives} no longer directly minimizes the KL divergence between $Q$ and the posterior, and can therefore avoid the ``mode-seeking'' behavior we just observed. Alternatively, we can stick with the ELBO and change the variational family. Although not part of our formal calculus, our full system (\S\ref{sec:sp}) extends the generative language with constructs not supported by other PPLs that automate VI, constructs that make it easy to add expressive power to our variational family:
\begin{itemize}[leftmargin=*]
    \item \textbf{Importance-weighted variational inference.} On the left side of Fig.~\ref{fig:objectives}, we define the \texttt{IWELBO} variational objective~\citep{burda_importance_2016}, as the expected value of the process that generates $N$ \textit{particles} from the variational family $\texttt{q}_\texttt{NAIVE}$, and computes a log mean importance weight.  
    %$\texttt{q}_\texttt{SIR}$ using the new $\normalize$ construct.\footnote{Both $\marginal$ and $\normalize$ lead to guides with intractable densities; in these cases, the weights computed by \textbf{density} and \textbf{sim} are stochastic estimates (as in \cite{lew_probabilistic_2023}). We discuss the implications of this for VI in \S\ref{sec:margnorm} and \S\ref{sec:diffsp}.\label{fn:estimated}} 
    This objective does not directly encourage $\texttt{q}_\texttt{NAIVE}$ to approximate the posterior well; rather, it encourages $\texttt{q}_\texttt{NAIVE}$ to be a good \textit{proposal distribution} for $N$-particle importance sampling, targeting the posterior. To illustrate this, we define a new probabilistic program $\texttt{q}_\texttt{SIR}$, which uses \textbf{normalize} to construct a \textit{sampling importance resampling} (SIR) approximation to the posterior: it generates $N$ samples from $\texttt{q}_\texttt{NAIVE}$, computes \textit{importance weights} for each sample, and randomly selects a sample to return according to the weights. 
    %In the first three plots on the right of $\texttt{q}_\texttt{SIR}$, we show samples from $\texttt{q}_\texttt{NAIVE}~\phi_N$, where $\phi_N$ are the parameters found by optimizing the ELBO objective on $\texttt{q}_\texttt{SIR}$ with $N$ particles. As we add more particles, it becomes safer for $\texttt{q}_\texttt{NAIVE}~\phi_N$ to spread its mass, because it will get $N$ attempts to propose a good posterior sample. 
    Drawing samples from $\texttt{q}_\texttt{SIR}$, with $N=30$ and $\phi$ set to the parameters that optimized the IWELBO objective, yields a close approximation to a larger portion of the posterior. (In fact, the IWELBO objective can be equivalently expressed in our framework as the ordinary ELBO, but applied directly to $\texttt{q}_\texttt{SIR}$ rather than to $\texttt{q}_\texttt{NAIVE}$.)

    \item \textbf{Hierarchical variational families.} On the right side of Fig.~\ref{fig:objectives}, we define a better variational family: $\texttt{q}_\texttt{MARG}$. First, we define a hand-designed posterior approximation $\texttt{q}_\texttt{JOINT}$, which extends $\texttt{q}_\texttt{NAIVE}$ with an \textit{auxiliary variable} $v$, to allow it to better fit circle-shaped posteriors. We cannot directly use $\texttt{q}_\texttt{JOINT}$ with the ELBO objective, however, because it is defined over the space of triples $(v, x, y)$, not the space of pairs $(x, y)$ as in our model. The $\marginal$ construct shrinks the sample space back down to just $(x, y)$, approximating the marginal density $\texttt{q}_\texttt{MARG}(x, y)$ using importance sampling. The resulting algorithm is an instance of importance weighted HVI (IWHVI) \cite{sobolev_importance_2019}.
\end{itemize}

\section{Syntax and Semantics}
\label{sec:syntax}
\begin{figure}
\centering
\footnotesize{
    \textcolor{blue}{\textit{\textbf{Shared Core}}}
    \begin{align*}
    \text{Ground types }\sigma \, ::=\,& 
        1 \mid 
        \mathbb{B} \mid
        %{\RR}_{\geq 0} \mid 
        \mathbb{I} \mid \RR_{> 0} \mid \RR_{\geq 0} \mid \RR \mid \RR^* \mid
        %\NN \mid
        \str \mid \sigma_1 \times \sigma_2 %\\ %\mid 
        %\trace \\
        %\sigma_1 + \sigma_2 
        %\mid \sigma_1 \times \sigma_2
        \quad\quad\quad\quad 
    \text{Types }\type \, ::= \,%&
        \sigma \mid D~\sigma \mid
        \tau_1 \to \tau_2 \mid 
       % \tau_1 + \tau_2 \mid 
        \tau_1 \times \tau_2\\
    \text{Terms }\ter \, ::= \,&
        () \mid 
        r \mid
        c \mid 
        \var \mid 
        \lambda x. t \mid
        t_1 \, t_2 \mid 
        (\ter_1, \ter_2) \mid 
        \pi_1~\ter \mid 
        \pi_2~\ter \mid 
        \mathbf{T} \mid \mathbf{F} \mid 
       % \inl~\ter \mid 
       % \inr~\ter \mid
       \textbf{if }t\textbf{ then }t_1\textbf{ else }t_2 \\ %\mid
      %  \llet~\var=\ter_1~\lin~\ter_2 \mid\\
    %\case~\ter~\of~\{\inl~\var_1\to \ter_1, \inr~\var_2\to\ter_2\} \mid 
    % \{\} \mid 
    %  \{ t_1 \mapsto t_2 \} \mid 
    %  t_1 \mdoubleplus t_2 \mid
    %  t_1[t_2]
    % &\;\;\;\;\;\; \sample(\ter_1, \ter_2) \mid \observe(\ter_1, \ter_2) \mid \score(\ter)
    \text{Primitives } c \, ::= \,& + \mid +^* \mid <~ \mid \normal_{\text{REPARAM}} \mid
            \normal_{\text{REINFORCE}} \mid
            \flip_{\text{REINFORCE}} \mid
            \flip_{\text{ENUM}} \mid
            \flip_{\text{MVD}} \mid 
            \dots
    \end{align*}

    \vspace{2mm}

    \begin{mdframed}
        \centering
        \begin{tabular}{c}
           $t : \RR^*$ 
           \\\hline
           $t : \RR$
        \end{tabular}
        \quad
        \begin{tabular}{c}
           $ $ 
           \\\hline
           $+ : \RR \times \RR \to \RR$
        \end{tabular}
        \quad
        \begin{tabular}{c}
           $ $ 
           \\\hline
           $+^* : \RR^* \times \RR^* \to \RR^*$
        \end{tabular}
        \quad
        \begin{tabular}{c}
           $ $ 
           \\\hline
           $<~: \RR^* \times \RR^* \to \mathbb{B}$
        \end{tabular}
        \quad
        \begin{tabular}{c}
           $ $ 
           \\\hline
           $\normal_{\text{REPARAM}} : \RR \times \RR_{>0} \to D\,\RR$ 
        \end{tabular}
        \quad
        \begin{tabular}{c}
           $ $ 
           \\\hline
           $\normal_{\text{REINFORCE}} : \RR \times \RR_{>0} \to D\,\RR^*$
        \end{tabular}
        \quad
        \begin{tabular}{c}
           $ $ 
           \\\hline
           $\flip_{\text{MVD}} : \mathbb{I} \to D\,\mathbb{B}$
        \end{tabular}
    \end{mdframed}

    \vspace{2mm}

    {\scriptsize{
    \centering
    $\sem{1} = \{()\}$ \quad $\sem{\mathbb{I}} = [0, 1]$ \quad $\sem{\RR}=\sem{\RR^*}=\RR$ \quad $\sem{D~\sigma} = \text{Prob}_{\ll \mathcal{B}_\sigma}~\sem{\sigma}$
    \quad
    $\sem{\mathbf{normal}_{\text{REPARAM}}} = \sem{\normal_{\text{REINFORCE}}} = \lambda (\mu, \sigma). \mathcal{N}(\mu, \sigma)$
    }}

    \vspace{2mm}
    
\begin{minipage}[t]{0.48\textwidth}
    \centering
    \textcolor{blue}{\textit{\textbf{Generative Probabilistic Programming}} ($\lambda_{\text{Gen}}$)}
    \begin{align*}
        \text{Types }\tau \, ::=\, & 
            G\,\tau \text{ (generative programs)}\\
        \text{Terms }\ter \, ::=\, & 
            \return\, \ter \mid
            \sample\, \ter_1\, \ter_2 \mid\\
            & \observe\, \ter_1\, \ter_2 \mid
             \haskdo \{ m \}\\
        m \, ::= \, &
            t \mid x \gets t; m
    \end{align*}

    \vspace{2mm}

    \begin{mdframed}
        \centering
        \begin{tabular}{c}
           $t : \tau$ 
           \\\hline
           $\return\,t : G\,\tau$ 
        \end{tabular}
        \quad
        \begin{tabular}{c}
            $t_1 : \dens~\sigma$\quad
            $t_2 : \str$
             \\\hline
            $\sample\, t_1\, t_2 : G\,\sigma$ 
        \end{tabular}
        \quad
        \begin{tabular}{c}
            $t_1 : \dens~\sigma$\quad
            $t_2 : \sigma$
             \\\hline
            $\observe\, t_1\, t_2 : G\,1$ 
        \end{tabular}
        \quad
        % \begin{tabular}{c}
        %    $t_p : G\,\tau$
        %     \\\hline
        %    $t_p : M\,\tau$ 
        % \end{tabular}
        % \quad
        \begin{tabular}{c}
            $t : G\,\tau$
             \\\hline
            $\haskdo \{ t \} : G\,\tau$ 
        \end{tabular}
        % \quad
        % \begin{tabular}{c}
        %     $t_p : M\,\tau$
        %      \\\hline
        %     $M.\haskdo \{ t_p \} : M\,\tau$ 
        % \end{tabular}
        %\quad
        \\\vspace{0.4mm}
        \begin{tabular}{c}
            $t : G\,\tau_1$ \quad
            $x : \tau_1 \vdash \haskdo\, \{ m \} : G\,\tau_2$
             \\\hline
            $\haskdo \{ x \gets t; m \} : G\,\tau_2$ 
        \end{tabular}
    \end{mdframed}
    \vspace{2.4mm}    
    \scriptsize{
    \centering
        $\sem{G~\tau} = \text{Meas}_{\ll \mathcal{B}_\mathbb{T}}^{DS}~\mathbb{T} \times (\mathbb{T} \Rightarrow \sem{\tau})$
        \begin{align*}
        \sem{\return~t}_1(\gamma,U) &= \delta_{\{\}}(U)\\
        \sem{\return~t}_2(\gamma,u) &= \sem{t}(\gamma)\\
        \sem{\sample~t_1~t_2}_1(\gamma,U) &= \int \sem{t_1}(\gamma, dv) \delta_{\{\sem{t_2}(\gamma) \mapsto v\}}(U)\\
        \sem{\sample~t_1~t_2}_2(\gamma,u) &= u[\sem{t_2}(\gamma)]\\
        \sem{\observe~t_1~t_2}_1(\gamma,U) &= \frac{d(\sem{t_1}(\gamma))}{d\mathcal{B}_\sigma}(\sem{t_2}(\gamma))  \delta_{\{\}}(U)\\
        \sem{\observe~t_1~t_2}_2(\gamma,u) &= ()\\
        \sem{\haskdo\{x \gets t; m\}}_1(\gamma,U) &=\int \sem{t}_1(\gamma, du_1)\\
        &\quad \int \sem{\haskdo\{m\}}_1(\gamma', du_2)\\ 
        &\quad\quad \delta_{u_1\mdoubleplus u_2}(U) \cdot \textit{disj}(u_1, u_2)\\
        \text{where }\gamma'&=\gamma[x \mapsto \sem{t}_2(\gamma, u_1)]\\
        \sem{\haskdo\{x \gets t; m\}}_2(\gamma,u) &= \sem{\haskdo\{m\}}_2(\gamma', u')\\
        \text{where }\gamma'&=\gamma[x \mapsto \sem{t}_2(\gamma, u)]\\
        \text{and }u'&=\pi_2(\textit{split}_{\sem{t}_1(\gamma)}(u))
        \end{align*}
    }
\end{minipage}\hfill\vline\hfill%
%    \vspace{4mm}
\begin{minipage}[t]{0.48\textwidth}
    \centering
    \textcolor{blue}{\textit{\textbf{Differentiable Probabilistic Programming}} ($\lambda_{\text{ADEV}}$)}
    \begin{align*}
        \text{Types }\tau \, ::=\, & 
            P~\tau \mid \eRR \mid \trace \\
        \text{Terms }\ter\, ::=\,& \mathbb{E}~t \mid \return~t \mid \sample~t \mid \haskdo\{m\} \mid\\
        & \mathbf{score}~t\mid \{\} \mid \{ t_1 \mapsto t_2 \} \mid t_1 \mdoubleplus t_2 \mid t_1[t_2]\\
 m \, ::= \, &
            t \mid x \gets t; m
    \end{align*}
    
    \vspace{2mm}

    \begin{mdframed}
        \begin{tabular}{c}
           $t : \tau$ 
           \\\hline
           $\return\,t : P\,\tau$ 
        \end{tabular}
        \quad
        \begin{tabular}{c}
            $t : \dens~\sigma$
             \\\hline
            $\sample\, t : P\,\sigma$ 
        \end{tabular}
        % \begin{tabular}{c}
        %     $t_d : \dens~\sigma$\quad
        %     $t : \sigma$
        %      \\\hline
        %     $\observe\, t_d\, t : M\,1$ 
        % \end{tabular}
        \\\centering
        % \begin{tabular}{c}
        %    $t_p : G\,\tau$
        %     \\\hline
        %    $t_p : M\,\tau$ 
        % \end{tabular}
        % \quad
        \begin{tabular}{c}
            $t : \RR_{\geq 0}$
             \\\hline
            $\mathbf{score}~t : P\,1$ 
        \end{tabular}
        \quad
        \begin{tabular}{c}
            $t : P\,\tau$
             \\\hline
            $\haskdo \{ t \} : P\,\tau$ 
        \end{tabular}
        \quad
        \begin{tabular}{c}
            $t : P\,\RR$
             \\\hline
             \\[-1.1em]
            $\mathbb{E}~t : \eRR$ 
        \end{tabular}
        \quad
        % \begin{tabular}{c}
        %     $t_p : M\,\tau$
        %      \\\hline
        %     $M.\haskdo \{ t_p \} : M\,\tau$ 
        % \end{tabular}
        %\quad
        \begin{tabular}{c}
            $t : P\,\tau_1$ \quad
            $x : \tau_1 \vdash \haskdo\, \{ m \} : P\,\tau_2$
             \\\hline
            $\haskdo \{ x \gets t; m \} : P\,\tau_2$ 
        \end{tabular}
        % \quad
        % \begin{tabular}{c}
        %     $t_p : M\,\tau_1$ \quad
        %     $x : \tau_1 \vdash M.\haskdo\, \{ m \} : M\,\tau_2$
        %      \\\hline
        %     $M.\haskdo \{ x \gets t_p; m \} : M\,\tau_2$ 
        % \end{tabular}
    \end{mdframed}
    \vspace{2mm}
    \scriptsize{
    \centering
        $\sem{P~\tau} = \text{Meas}~\sem{\tau}\;\mid\;\sem{\eRR} = \text{Meas}~\RR\;\mid\;\sem{\trace}=\mathbb{T}$
        \begin{align*}
        \sem{\return~t}(\gamma,U) &= \delta_{\sem{t}(\gamma)}(U)\\
        \sem{\sample~t}(\gamma,U) &= \sem{t}(\gamma, U)\\
        \sem{\score~t}(\gamma, U) &= \sem{t}(\gamma) \cdot \delta_{()}(U)\\
        \sem{\haskdo\{x \gets t; m\}}(\gamma,U) &=\int \sem{t}(\gamma, du_1)\\
        &\quad\quad\; \sem{\haskdo\{m\}}(\gamma[x \mapsto u_1], U)\\
        \sem{\mathbb{E}~t}(\gamma, U) &= \sem{t}(\gamma, U)\\
        \sem{\{\}}(\gamma) &= \{\}\\
        \sem{\{t_1 \mapsto t_2\}}(\gamma) &= \{\sem{t_1}(\gamma) \mapsto \sem{t_2}(\gamma)\}\\
        \sem{t_1 \mdoubleplus t_2}(\gamma) &= \begin{cases}
            u_1 \mdoubleplus u_2 &\text{ if } \textit{disj}(u_1, u_2)\\
            \{\} &\text{ otherwise}
        \end{cases}\\
        \text{where } u_i &= \sem{t_i}(\gamma)\\
        \sem{t_1[t_2]}(\gamma) &= \begin{cases}
            u[v] &\text{ if } v \in u\\
            \text{default}_\sigma &\text{ otherwise}
        \end{cases}\\
        \text{where } (u,v) &= (\sem{t_1}(\gamma), \sem{t_2}(\gamma))
        \end{align*}
    }
\end{minipage}
}
% \vspace{3mm}

% {\it Target Language}
%     \begin{align*}
%     \text{Types }\tau \, ::=\,& 
%         P_{\text{Cont}}~\tau \mid \trace_D %\\ %\mid 
%         %\trace \\
%         %\sigma_1 + \sigma_2 
%         %\mid \sigma_1 \times \sigma_2
%         \quad\quad\quad\quad\quad\quad 
%     \text{Types }\type \, ::= \,%&
%         \sigma \mid D~\sigma \mid
%         \tau_1 \to \tau_2 \mid 
%        % \tau_1 + \tau_2 \mid 
%         \tau_1 \times \tau_2\\
%     \text{Terms }\ter \, ::= \,&
%         () \mid 
%         c \mid 
%         \var \mid 
%         \lambda x. t \mid
%         t_1 \, t_2 \mid 
%         (\ter_1, \ter_2) \mid 
%         \pi_1~\ter \mid 
%         \pi_2~\ter \mid 
%        % \inl~\ter \mid 
%        % \inr~\ter \mid
%        \textbf{if }t\textbf{ then }t_1\textbf{ else }t_2\\ %\mid
%       %  \llet~\var=\ter_1~\lin~\ter_2 \mid\\
%     %\case~\ter~\of~\{\inl~\var_1\to \ter_1, \inr~\var_2\to\ter_2\} \mid 
%     % \{\} \mid 
%     %  \{ t_1 \mapsto t_2 \} \mid 
%     %  t_1 \mdoubleplus t_2 \mid
%     %  t_1[t_2]
%     % &\;\;\;\;\;\; \sample(\ter_1, \ter_2) \mid \observe(\ter_1, \ter_2) \mid \score(\ter)
%     \text{Primitive distributions }\ter_d \, ::=\, &
%             \normal_{\text{REPARAM}} \mid
%             \normal_{\text{REINFORCE}} \mid
%             \bernoulli_{\text{REINFORCE}} \mid
%             \bernoulli_{\text{MVD}} \mid 
%             \dots
    % \end{align*}
    \Description{Grammar and typing rules for our languages.}
    \caption{Grammars, selected typing rules, and selected denotations for our core languages $\lambda_{\text{Gen}}$ and $\lambda_{ADEV}$.}
    \label{fig:syntax}
    \vspace{-4mm}
\end{figure}

We now formalize the core of our approach. Although our formal model lacks several features of our full language (see \S\ref{sec:sp} and Appx. \ref{appdx:full-system}), it is still expressive, featuring higher-order functions, continuous and discrete sampling, stochastic control flow, and discontinuous branches. Ultimately, our formalization aims to give an account of how user programs representing models, variational families, and variational objectives are transformed into compiled programs representing unbiased gradient estimators. We begin in this section by introducing two calculi for generative and differentiable probabilistic programming (Fig.~\ref{fig:syntax}).

\subsection{Shared Core}
\subsubsection{Syntax} 
The top of Fig.~\ref{fig:syntax} presents the \textit{shared core}, the $\lambda$-calculus on which both our languages build. It is largely standard, with functions, tuples, if statements, and ground types for Booleans, strings, and numbers, but two aspects merit further discussion:

\begin{itemize}[leftmargin=*]
\item\textit{Smooth and non-smooth reals.} First, following~\citet{lew_adev_2023}, we have two types for real numbers, $\RR$ and $\RR^*$. Intuitively, the type $\RR$ is the type of ``real numbers that must be used differentiably,'' whereas the type $\RR^*$ is the type of ``real numbers that may be manipulated in any (measurable) way.'' These constraints are enforced by the types of primitive functions: non-smooth primitives like $< ~: \RR^* \times \RR^* \to \mathbb{B}$ only accept $\RR^*$ inputs. Smooth primitives, by contrast, come in two varieties, smooth versions (e.g. $+ : \RR \times \RR \to \RR$) and non-smooth versions (e.g. $+^* : \RR^* \times \RR^* \to \RR^*$).\footnote{The reader may wonder if we also need primitives that accept some smooth and some non-smooth inputs, but this turns out to be unnecessary, because non-smooth inputs can always be safely promoted to smooth inputs. The output will then also be smooth, but this is by design: if \textit{any} input to a primitive has smooth type, then the primitive's output must also have smooth type, to ensure that future computation does not introduce non-differentiability.} We allow implicit promotion of terms of type $\RR^*$ into terms of type $\RR$.

\item\textit{Primitive distributions.} The shared core includes a type $D~\sigma$ of primitive probability distributions over ground types $\sigma$. Again following~\citet{lew_adev_2023}, we expose multiple versions of each primitive, e.g. $\normal_{\text{REPARAM}}$ and $\normal_{\text{REINFORCE}}$. All versions denote the same distribution, and so our correctness results (which are phrased in terms of our denotational semantics) ensure that gradients will target the same objective no matter which version a program uses (Thm.~\ref{thm:unbiased-adev}). But different versions of the same primitive employ different \textit{estimation strategies} for propagating derivatives, striking different trade-offs between variance and cost. Furthermore, because different  estimation strategies may place different requirements on the user's program, the typing rules for different versions of the same primitive may differ. For example, $\normal_{\text{REINFORCE}}$ constructs a distribution of type $D~\RR^*$, meaning that probabilistic programs that draw samples from it can freely manipulate those samples. By contrast, $\normal_{\text{REPARAM}}$ constructs a distribution of type $D~\RR$, so samples from $\normal_{\text{REPARAM}}$ must be used smoothly.
\end{itemize}

\subsubsection{Denotational Semantics.} %For a term $t$, we use the standard notation
%$\Gamma \vdash t : \tau$ for the judgment that, in context $\Gamma$ (mapping a finite list of free variable names $x$ to their types $\Gamma(x)$), the term $t$ has type $\tau$. 
We assign to each type $\tau$ a mathematical space $\sem{\tau}$, and interpret an open term $\Gamma \vdash t : \tau$ as a map from $\sem{\Gamma} \coloneqq \prod_{x \in \Gamma} \sem{\Gamma(x)}$, the space of \textit{environments} for context $\Gamma$, to $\sem{\tau}$, the space of results to which $t$ can evaluate. Formally, we work in the category of \textit{quasi-Borel spaces}~\cite{heunen_convenient_2017}, but to ease exposition, we present our semantics in terms of standard measure theory when possible. So, for example, we write that $\sem{\RR}$ is the measurable space $(\RR, \mathcal{B}(\RR))$, that $\sem{\sigma_1 \times \sigma_2}$ is the product of the measurable spaces $\sem{\sigma_1}$ and $\sem{\sigma_2}$, and so on, but with the implicit understanding that these can equivalently be viewed as quasi-Borel spaces. For semantics of higher-order types, we use quasi-Borel spaces explicitly (e.g., we set $\sem{\tau_1 \to \tau_2}$ to $\sem{\tau_1} \Rightarrow \sem{\tau_2}$, the quasi-Borel space of quasi-Borel maps from $\sem{\tau_1}$ to $\sem{\tau_2}$). %For example, for quasi-Borel spaces that happen to be standard Borel, we work with them as measurable spaces; and for quasi-Borel measures defined on those spaces, we work them as ordinary measures. %The quasi-Borel spaces are drop-in replacements for measurable spaces, with key differences appearing only at higher-order. Our presentation exploits this: at ground types, we work with ordinary measures and measurable spaces.

To give a semantics to our primitive distributions, we need to first assign to each ground type $\sigma$ a \textit{base measure} $\mathcal{B}_\sigma$ over $\sem{\sigma}$: for discrete types $\sigma \in \{1, \mathbb{B}, \str\}$, we set $\mathcal{B}_\sigma(U) = |U|$, the counting measure, and for continuous types $\sigma \in \{\RR, \RR^*\}$, we set $\mathcal{B}_\sigma(U) = \int_\RR 1_U(u) du$, the Lebesgue measure. Given two ground types $\sigma_1$ and $\sigma_2$, the base measure of the product, $\mathcal{B}_{\sigma_1 \times \sigma_2}$, is $\mathcal{B}_{\sigma_1} \otimes \mathcal{B}_{\sigma_2}$, the product of the base measures for each type. 
With base measures in hand, we can define $\sem{D~\sigma} \coloneqq \text{Prob}_{\ll \mathcal{B}_\sigma} \sem{\sigma}$, the space of probability measures on $\sem{\sigma}$ that are \textit{absolutely continuous} with respect to $\mathcal{B}_\sigma$. Absolute continuity ensures that these distributions have \textit{density functions}.% For example, discrete distributions on (countable subsets of) $\RR$ cannot be introduced as primitives. 

\subsection{Generative Probabilistic Programming with \texorpdfstring{$\lambda_{\text{Gen}}$}{Lambda Gen}}
\label{sec:gen-lang}
Users write probabilistic models and variational families in $\lambda_{\text{Gen}}$, which extends the shared core with a monadic type $G~\tau$ of \textit{generative programs} (Fig.~\ref{fig:syntax}, left). Examples of programs in $\lambda_{\text{Gen}}$ include the model and variational programs (\texttt{model} and $\texttt{q}_\texttt{NAIVE}$) in Fig.~\ref{fig:example_transcript}.

\subsubsection{Syntax} Syntactically, programs of type $G~\tau$ interleave standard functional programming logic (from the shared core) with two new kinds of statements: $\sample$ and $\observe$. The $\sample$ statement takes as input a probability distribution to sample (of type $D~\sigma$) and a unique \textit{name} for the random variable being sampled (of type $\str$). The $\observe$ statement takes as input a probability distribution representing a likelihood (of type $D~\sigma$), and a value representing an observation (of type $\sigma$). A generative program can include many calls to $\sample$ and $\observe$, ultimately inducing an \textit{unnormalized joint distribution} over \textit{traces}: finite dictionaries mapping the names of random variables to sampled values. Intuitively, ignoring the $\observe$ statements in a program, we can read off a sampling distribution over traces\textemdash the \textit{prior}. The $\observe$ statements then reweight each possible execution's trace by the likelihoods accumulated during that execution, yielding an unnormalized \textit{posterior} over traces.

\subsubsection{Denotational Semantics} Formally, we write $\mathbb{T}$ for the space of possible traces. It arises as a countable disjoint union, indexed by possible \textit{trace shapes} (finite partial maps from string-valued \textit{names} $k$ to corresponding ground types $\sigma_k$), of product spaces $\prod_{k} \sem{\sigma_k}$. We can also define a base measure $\mathcal{B}_\mathbb{T}$ over $\mathbb{T}$, by summing product measures for each possible trace shape: $$\mathcal{B}_\mathbb{T}(U) \coloneqq \sum_{s \subseteq \str \times \Sigma} \left(\bigotimes_{(k, \sigma_k) \in s} \mathcal{B}_{\sigma_k}\right)(\{\textit{values}(u) \mid u \in U \wedge \textit{shape}(u) = s\}).$$ 

Generative programs (of type $G~\tau$) induce measures $\mu$ on $\mathbb{T}$ that satisfy two properties: they are absolutely continuous with respect to $\mathcal{B}_\mathbb{T}$ (and therefore have a well-defined notion of \textit{trace density} $\frac{d\mu}{d\mathcal{B}_\mathbb{T}}$), and they are \textit{discrete-structured}: for $(\mu \otimes \mu)$-almost-all pairs of traces $(u_1, u_2)$, either $u_1 = u_2$, or there exists a string $s$ present in both $u_1$ and $u_2$ such that $u_1[s] \neq u_2[s]$. In words, if two distinct traces are both in the support of a generative program, then in the two executions that those traces represent, there must have been some $\sample$ statement at which different choices were made for the same random variable. We write $\text{Meas}_{\ll \mathcal{B}_\mathbb{T}}^{DS}~\mathbb{T}$ for the space of measures on $\mathbb{T}$ satisfying these two properties. The semantics then assigns $\sem{G~\tau} \coloneqq \text{Meas}_{\ll \mathcal{B}_\mathbb{T}}^{DS}~\mathbb{T} \times (\mathbb{T} \Rightarrow \sem{\tau})$: a generative program denotes both a (well-behaved) measure on traces, \textit{and} a \textit{return-value function} that given a trace, computes the program's $\sem{\tau}$-valued result when its random choices are as in the trace.

The semantics for terms of type $G~\tau$ (Fig.~\ref{fig:syntax}, bottom left) give a formal account of how a generative program's source code yields a particular measure on traces and return-value function. We write $\sem{\cdot}_i$ for $\pi_i \circ \sem{\cdot}$, so that $\sem{\cdot}_1$ computes the measure on traces and $\sem{\cdot}_2$ computes the return-value function. In Appx.~\ref{appx:proof-details}, we show that our term semantics really does map every program of type $G~\tau$ a well-behaved measure over traces, i.e., the absolute continuity and discrete-structure requirements are satisfied by our definitions. The semantics of each probabilistic programming construct can be understood in terms of its trace distribution and return value function: 
\begin{itemize}[leftmargin=*]
\item $\return~t$: The simplest generative programs deterministically compute a return value $t$. These programs denote deterministic (Dirac delta) distributions on the \textit{empty trace} $\{\}$, because they make no random choices. The return value function $\sem{\return~t}_2(\gamma)$ then maps any trace $u$ to the program's return value, $\sem{t}(\gamma)$. 
\item $\sample~t_1~t_2$: Slightly more complicated, the $\sample~t_1~t_2$ command denotes a measure over \textit{singleton} traces $\{ \textit{name} \mapsto \textit{value} \}$, namely the \textit{pushforward} of the measure $\sem{t_1}(\gamma)$ by the function $\lambda v. \{\sem{t_2}(\gamma) \mapsto v\}$. The return value function for $\sample~t_1~t_2$ statements accepts a trace $u$ and looks up the value associated with the name $\sem{t_2}(\gamma)$, if it exists. We define $u[\textit{name}]$ to return the value associated with \textit{name} in $u$, if \textit{name} is a key in $u$, and otherwise to return a default value of the appropriate type.\footnote{All ground types $\sigma$ are inhabited, and we can choose the default values arbitrarily.} 
\item $\observe~t_1~t_2$: Like $\return$, the statement $\observe~t_1~t_2$ makes no random choices, and thus has an empty trace. But it denotes a \textit{scaled} Dirac delta measure: the measure it assigns to the empty trace is equal to the density of the value $\sem{t_2}(\gamma)$ under the measure $\sem{t_1}(\gamma)$. 
\item $\haskdo \{x \gets t; m\}$: Sequencing two generative programs concatenates their traces (which we write using the $\mdoubleplus$  concatenation operator). The helper $\textit{disj} : \mathbb{T} \times \mathbb{T} \to \{0, 1\}$ checks whether the names used by each of two traces are distinct, returning $1$ if so and $0$ otherwise. We use it in our definition of $\sem{\haskdo \{x \gets t; m\}}$ to model that when a program uses the same name twice in a single execution, a runtime error is raised: the semantics assigns measure 0 to those executions, leaving measure $< 1$ for the remaining, valid executions.\footnote{The semantics of a runtime error are equivalent to the semantics of $\mathbf{observe}$-ing an impossible outcome. Because of this, if the user's model contains such errors, variational inference can be seen as training a guide program to approximate the model posterior \textit{given} that no errors are encountered. If the variational program itself contains either $\observe$ statements or runtime errors, it can also denote an unnormalized measure, in which case an objective like the ELBO ($\int Q(dx) \log \frac{Zp(x)}{q(x)}$) can no longer be interpreted as a lower bound on the model's log normalizing constant $\log Z$. Our system will still produce unbiased gradient estimates for the objective, but it is likely not an objective that the user intended to optimize. This suggests that better static checks for whether a program is normalized could be useful, helping users to avoid silent optimization failures if they accidentally encode an unnormalized variational family.}
\end{itemize}

\subsection{Differentiable Probabilistic Programming with \texorpdfstring{$\lambda_{\text{ADEV}}$}{Lambda ADEV}}
\label{sec:adev-lang}
\subsubsection{Syntax} The right panel of Fig.~\ref{fig:syntax} presents a separate extension of the shared core, $\lambda_{\text{ADEV}}$, a lower-level language for differentiable probabilistic programming~\cite{lew_adev_2023}. Like $\lambda_{\text{Gen}}$, $\lambda_{\text{ADEV}}$ adds a monadic type for probabilistic computations, $P~\tau$, which supports the sampling of primitive distributions (with $\sample$), deterministic computation (with $\return$), scoring by multiplicative density factors ($\score$), and sequencing (with $\haskdo$). But $\lambda_{\text{ADEV}}$ probabilistic programs do not denote distributions on \textit{traces}, and the $\sample$ statements do not specify names for random variables. Rather, programs directly denote (quasi-Borel) measures on output types $\tau$.

Even so, we do include syntax for constructing and manipulating traces \textit{as data}. This is because, in $\S\ref{sec:gen-transformations}$, we will develop program transformations that turn $\lambda_{\text{Gen}}$ programs into $\lambda_{\text{ADEV}}$ programs that (differentiably) simulate and evaluate densities of reified trace data structures. 

The language also features an \textit{expected value operator} $\mathbb{E}$,
of type $P~\RR \to \eRR$. Terms of type $\eRR$ intuitively represent ``losses (i.e., objectives) that can be unbiasedly estimated.'' 
The ultimate goal of a user in $\lambda_{\text{ADEV}}$ is to write an objective function of type $\RR^n \to \eRR$ (where $\RR^n = \RR \times \cdots \times \RR$), and then use automatic differentiation of expected values (\S\ref{sec:adev})
to obtain a gradient estimator for the loss function it denotes. 
These loss functions can be constructed by taking expectations of probabilistic programs (using $\mathbb{E}$) or by composing existing losses using new 
primitives (e.g., $+_{\eRR}$, 
$\times_{\eRR}$, and $\text{exp}_{\eRR}$). Example $\lambda_\text{ADEV}$ programs include the objectives defined in \S\ref{sec:overview} (ELBO, IWELBO, and IWHVI), as well as the automatically compiled programs on the right-hand side of Fig.~\ref{fig:example_transcript}.

\subsubsection{Denotational Semantics} Semantically, $P~\tau$ is the space of quasi-Borel measures on $\sem{\tau}$. The type $\eRR$ has the same denotation as $P~\RR$, but cannot be used monadically (i.e., it composes in a more restricted way than $P~\RR$). This is because it is intended to represent an unbiased estimator of a particular real number. For example, suppose $p : P~\RR$ denotes a probability measure over the reals with expected value $0$. Then $\mathbb{E}~p : \eRR$ denotes the same probability measure, with expected value $0$. But $p$ can be used within larger probabilistic programs; for example, we can write $p_* = \haskdo \{x \gets p; \return \exp(p)\}$, which draws a sample from $p$ and exponentiates it. By Jensen's inequality, the expected value of $\sem{p_*}$ will generally be \textit{greater than} $e^0 = 1$. By contrast, the term $\mathbb{E}~p$ cannot be freely sampled within probabilistic programs, but can be composed with certain special arithmetic operators, so we can write (for example) $p^* = \exp_{\eRR} (\mathbb{E}~p)$.
The primitive $\exp_{\eRR}$ uses special logic to construct an unbiased estimator of $e^x$ given an unbiased estimator of $x$, so the program $p^*$ denotes a probability distribution that \textit{does} have expectation $e^0 = 1$.

\section{Compiling Differentiable Simulators and Density Evaluators}
\label{sec:gen-transformations}
\begin{figure}
\scriptsize{
\begin{tabular}{ll}
    % Overall spec
   {Spec}  & \framebox{
    %\begin{tabularx}{0.9\textwidth}{@{}>{\raggedright\arraybackslash}X@{\hfill}>{\raggedleft\arraybackslash}X@{}}
    \begin{tabular*}{0.87\textwidth}{@{\extracolsep{\fill}}lr@{}}
    \textcolor{blue}{\textbf{\textit{Syntax}}}\;\;$\vdash  t : \sigma \to G~\tau \implies \vdash \mathbf{density}\{t\} : \sigma \to \trace \to \RR$ & \textcolor{blue}{\textbf{\textit{Semantics}}}\;\;$\sem{\mathbf{density}\{t\}}(\theta) = \frac{d\sem{t}_1(\theta)}{d\mathcal{B}_\mathbb{T}}$
   \end{tabular*}} \\
   % Wrapper implementation
   \rule{0pt}{4ex}{Wrapper}  & $\mathbf{density}\{t\} \coloneqq \lambda \theta. \lambda u. \llet~(x, w, u') = \xi\{t\}(\theta)(u)~\lin~\mathbf{if}~\textit{isempty}(u')~\mathbf{then}~w~\mathbf{else}~0$\\
   % Spec for helper
   \rule{0pt}{6ex}{Helper} $\xi$ & \framebox{

    \begin{tabular*}{0.87\textwidth}{@{\extracolsep{\fill}}lr@{}}
        \textcolor{blue}{\textbf{\textit{Syntax}}}\;\;$\Gamma \vdash t : \tau \implies \xi\{\Gamma\} \vdash \xi\{t\} : \xi\{\tau\}$ & \textcolor{blue}{\textbf{\textit{Semantics}}}\;\;$\forall (\gamma, \gamma') \in R^\xi_{\Gamma}, (\sem{t}(\gamma), \sem{\xi\{t\}}(\gamma')) \in R^{\xi}_\tau$
    \end{tabular*}
   }\\
   % Implementation on types
    \rule{0pt}{4ex}\;\textit{on types} & \begin{tabular}[t]{@{}l@{}l|l@{}l}
        $\xi\{\sigma\}$ &$\coloneqq \sigma$ & $R^\xi_\sigma$ &$\coloneqq \{(x, x) \mid x \in \sem{\sigma}\}$ \\
        $\xi\{D~\sigma\}$&$\coloneqq \sigma \to \RR$ & $R^\xi_{D~\sigma}$ &$\coloneqq \{(\mu, \rho) \mid \rho = \frac{d\mu}{d\mathcal{B}_\sigma}\}$\\
        $\xi\{\tau_1 \times \tau_2\}$&$\coloneqq \xi\{\tau_1\} \times \xi\{\tau_2\}$ & $R^\xi_{\tau_1 \times \tau_2}$ &$\coloneqq \{((x,y), (x',y')) \mid (x,x') \in R^\xi_{\tau_1} \wedge (y,y') \in R^\xi_{\tau_2}\}$\\
        $\xi\{\tau_1 \to \tau_2\}$&$\coloneqq \xi\{\tau_1\} \to \xi\{\tau_2\}$ & $R^\xi_{\tau_1 \to \tau_2}$ &$\coloneqq \{(f,g) \mid \forall (x,y) \in R^\xi_{\tau_1}, (f(x),g(y)) \in R^\xi_{\tau_2}\}$\\
        $\xi\{G~\tau\}$&$\coloneqq \trace \to \xi\{\tau\} \times \RR \times \trace$ & $R^\xi_{G~\tau}$ &\begin{tabular}[t]{@{}l}
        $\coloneqq \{((\mu,f),g) \mid \exists h. \forall u \in \mathbb{T}. (f(u), h(u)) \in R^\xi_{\tau} \; \wedge$ \\ 
        $\quad\quad\quad g = \lambda u. \llet~(u_1, u_2) = \textit{split}_\mu(u)~\lin (h(u), \frac{d\mu}{dB_\mathbb{T}}(u_1), u_2)\}$
        \end{tabular}\\
        % & & \text{where} $\mathbb{T}_\mu$&$\coloneqq \{u \in \mathbb{T} \mid \exists u' \in \mathbb{T}. \frac{d\mu}{d\mathcal{B}_\mathbb{T}}(u\mdoubleplus u') > 0\}, \mathbb{T}_\mu^*\coloneqq \lambda u. \text{argmax}_{\{u' \in \mathbb{T}_\mu \mid u' \subseteq u\}} |u'|$\\ % \forall_{\mu \otimes \mu} (u_1, u_2), u_1 = u_2 \vee \exists s \in u_1. s \in u_2 \wedge u_1[s] \neq u_2[s]$\\
        %& & \text{and} $\mathbb{T}_\mu^*$& $\coloneqq \lambda u. \text{argmax}_{\{u' \in \mathbb{T}_\mu \mid u' \subseteq u\}} |u'|$\\
        % & & \text{and} $\text{eat}_\mu$ & $\coloneqq \lambda u. u \setminus \mathbb{T}_\mu^*(u)$
    \end{tabular} \\
   % \rule{0pt}{4ex}\quad\textit{on types} & $\xi\{\sigma\} \coloneqq \sigma \;\mid\; \xi\{D~\sigma\} \coloneqq \sigma \to \RR \;\mid\; \xi\{\tau_1 \to \tau_2\} \coloneqq \xi\{\tau_1\} \to \xi\{\tau_2\}$ \\
   % Implementation on terms
   \rule{0pt}{4ex}\;\textit{on terms} & \begin{tabular}[t]{@{}l@{}l@{}l@{}l}
        $\xi\{()\}$ &$\coloneqq ()$ &
        %$\xi\{c\}$ &$\coloneqq c_\xi$\\
        $\xi\{x\}$ &$\coloneqq x$ \\
        $\xi\{\lambda x. t\}$ &$\coloneqq \lambda x. \xi\{t\}$ &
        $\xi\{t_1~t_2\}$ &$\coloneqq \xi\{t_1\}~\xi\{t_2\}$\\
        $\xi\{(t_1, t_2)\}$ &$\coloneqq (\xi\{t_1\}, \xi\{t_2\})$ &
        $\xi\{\pi_i~t\}$ &$\coloneqq \pi_i~\xi\{t\}$\\
        $\xi\{b\} (b \in \{\mathbf{T}, \mathbf{F}\})$ &$\coloneqq b$ &
        $\xi\{\mathbf{if}~t~\mathbf{then}~t_1~\mathbf{else}~t_2\}$ &$\coloneqq \mathbf{if}~\xi\{t\}~\mathbf{then}~\xi\{t_1\}~\mathbf{else}~\xi\{t_2\}$\\
        $\xi\{\normal_{\textit{strat}}\}$ &$\coloneqq \lambda \mu. \lambda \sigma. \lambda x. \frac{1}{\sigma\sqrt{2\pi}}e^{-\frac{1}{2}\left(\frac{x-\mu}{\sigma}\right)^2}$ &
        $\xi\{\flip_{\textit{strat}}\}$ &$\coloneqq \lambda p. \lambda b. \mathbf{if}~b~\mathbf{then}~p~\mathbf{else}~1-p$
        % $\xi\{\normal_{a}\}$ 
        % &$\coloneqq \lambda \mu. \lambda \sigma. \lambda x. \frac{1}{\sigma\sqrt{2\pi}}e^{-\frac{1}{2}\left(\frac{x-\mu}{\sigma}\right)^2}$
        \\
        % $\xi\{\flip_{\text{ENUM}}\}$ &$\coloneqq \lambda p. \lambda b. \mathbf{if}~b~\mathbf{then}~p~\mathbf{else}~1-p$ &
        % $\xi\{\flip_{\text{REINFORCE}}\}$ 
        % &$\coloneqq \lambda p. \lambda b. \mathbf{if}~b~\mathbf{then}~p~\mathbf{else}~1-p$\\
        $\xi\{\return~t\}$ &$\coloneqq \lambda u. (\xi\{t\}, 1, u)$ &
        $\xi\{\observe~t_1~t_2\}$ &$\coloneqq \lambda u. ((), \xi\{t_1\}(\xi\{t_2\}), u)$ \\
        $\xi\{\sample~t_1~t_2\}$ &\begin{tabular}[t]{@{}l}
        $\coloneqq \lambda u. \llet~(v,w,u')=\textit{pop}~u~\xi\{t_2\}$\\
        $\quad\quad\;\;\,\lin~(v,w\cdot\xi\{t_1\}(v), u')$
        \end{tabular} &
        $\xi\{\haskdo~\{x \gets t; m\}\}$ &\begin{tabular}[t]{@{}l}
        $\coloneqq \lambda u. \llet~(x,w,u')=\xi\{t\}(u)~\lin$\\
        $\quad\quad\;\;\,\llet~(y,v,u'') = \xi\{\haskdo\{m\}\}(u')~\lin$\\
        $\quad\quad\;\;\,(y,w\cdot v, u'')$
        \end{tabular}
   \end{tabular}\\
\end{tabular}
\vspace{1mm}
\hrule
\vspace{2mm}
\begin{tabular}{ll}
    % Overall spec
   {Spec}  & \framebox{
   \begin{tabular*}{0.87\textwidth}{@{\extracolsep{\fill}}lr@{}}
   \textcolor{blue}{\textbf{\textit{Syntax}}}\;\;$\vdash t : \sigma \to G~\tau \implies \vdash \mathbf{sim}\{t\} : \sigma \to P~(\trace \times \RR)$ & \textcolor{blue}{\textbf{\textit{Semantics}}}\;\;$\sem{\mathbf{sim}\{t\}}(\theta) = (\text{id} \otimes  \frac{d\sem{t}_1(\theta)}{d\mathcal{B}_\mathbb{T}})_*\sem{t}_1(\theta)$
   \end{tabular*}}\\ %\int \sem{t}_1(\theta, du) 1_U\!\left(u, \frac{d\sem{t}_1(\theta)}{d\mathcal{B}_\mathbb{T}}(u)\right)$} \\
   % Wrapper implementation
   \rule{0pt}{4ex}{Wrapper}  & $\mathbf{sim}\{t\} \coloneqq \lambda \theta. \haskdo \{ (x, w, u) \gets \chi\{t\}(\theta); \return~(u, w)\}$\\
   % Spec for helper
   \rule{0pt}{6ex}{Helper} $\chi$ & \framebox{
    \begin{tabular*}{0.87\textwidth}{@{\extracolsep{\fill}}lr@{}}

   \textcolor{blue}{\textbf{\textit{Syntax}}}\;\;$\Gamma \vdash t : \tau \implies \chi\{\Gamma\} \vdash \chi\{t\} : \chi\{\tau\}$ & \textcolor{blue}{\textbf{\textit{Semantics}}}\;\; $\forall (\gamma, \gamma') \in R^\chi_{\Gamma}, (\sem{t}(\gamma), \sem{\chi\{t\}}(\gamma')) \in R^{\chi}_\tau$
   \end{tabular*}}\\
   % Implementation on types
    \rule{0pt}{4ex}\;\textit{on types} & \begin{tabular}[t]{@{}ll|ll}
        $\chi\{\sigma\}$ &$\coloneqq \sigma$ & $R^\chi_\sigma$ &$\coloneqq \{(x, x) \mid x \in \sem{\sigma}\}$ \\
        $\chi\{D~\sigma\}$&$\coloneqq P~(\sigma \times \RR) \times (\sigma \to \RR)$ & $R^\chi_{D~\sigma}$ &$\coloneqq \{(\mu, (\nu, \rho)) \mid \nu(U) = \int \mu(du) \delta_{\left(u,\frac{d\mu}{d\mathcal{B}_\sigma}(u)\right)}(U) \wedge \rho = \frac{d\mu}{d\mathcal{B}_\sigma}\}$\\
        $\chi\{\tau_1 \times \tau_2\}$&$\coloneqq \chi\{\tau_1\} \times \chi\{\tau_2\}$ & $R^\chi_{\tau_1 \times \tau_2}$ &$\coloneqq \{((x,y), (x',y')) \mid (x,x') \in R^\chi_{\tau_1} \wedge (y,y') \in R^\chi_{\tau_2}\}$\\
        $\chi\{\tau_1 \to \tau_2\}$&$\coloneqq \chi\{\tau_1\} \to \chi\{\tau_2\}$ & $R^\chi_{\tau_1 \to \tau_2}$ &$\coloneqq \{(f,g) \mid \forall (x,y) \in R^\chi_{\tau_1}, (f(x),g(y)) \in R^\chi_{\tau_2}\}$\\
        $\chi\{G~\tau\}$&$\coloneqq P~(\chi\{\tau\} \times \RR \times \trace)$ & $R^\chi_{G~\tau}$ &\begin{tabular}[t]{@{}l}$\coloneqq \{((\mu,f),\nu) \mid \exists g. \forall u. (f(u),g(u)) \in R_\tau^\xi \wedge$\\ %
        %\mu(\mathbb{T}) < 1 \vee 
        \quad\quad\quad\quad$\nu=\lambda U.\int\mu(du) 1_U(g(u), \frac{d\mu}{d\mathcal{B}_\mathbb{T}}(u), u)\}$\end{tabular}
    \end{tabular} \\
   % \rule{0pt}{4ex}\quad\textit{on types} & $\xi\{\sigma\} \coloneqq \sigma \;\mid\; \xi\{D~\sigma\} \coloneqq \sigma \to \RR \;\mid\; \xi\{\tau_1 \to \tau_2\} \coloneqq \xi\{\tau_1\} \to \xi\{\tau_2\}$ \\
   % Implementation on terms
   \rule{0pt}{4ex}\;\textit{on terms} & \begin{tabular}[t]{@{}l@{}l@{}l@{}l}
        $\chi\{()\}$ &$\coloneqq ()$ &
        %$\xi\{c\}$ &$\coloneqq c_\xi$\\
        $\chi\{x\}$ &$\coloneqq x$ \\
        $\chi\{\lambda x. t\}$ &$\coloneqq \lambda x. \chi\{t\}$ &
        $\chi\{t_1~t_2\}$ &$\coloneqq \chi\{t_1\}~\chi\{t_2\}$\\
        $\chi\{(t_1, t_2)\}$ &$\coloneqq (\chi\{t_1\}, \chi\{t_2\})$ &
        $\chi\{\pi_i~t\}$ &$\coloneqq \pi_i~\chi\{t\}$\\
        $\chi\{b\} (b \in \{\mathbf{T}, \mathbf{F}\})$ &$\coloneqq b$ &
        $\chi\{\mathbf{if}~t~\mathbf{then}~t_1~\mathbf{else}~t_2\}$ &$\coloneqq \mathbf{if}~\chi\{t\}~\mathbf{then}~\chi\{t_1\}~\mathbf{else}~\chi\{t_2\}$\\
        $\chi\{\normal_{\textit{strat}}\}$ &\begin{tabular}[t]{@{}l}
        $\coloneqq \lambda (\mu, \sigma). (\haskdo \{$\\
        $\quad\quad x \gets \sample~\normal_{\textit{strat}}(\mu, \sigma);$\\
        $\quad\quad \llet~\rho = \frac{1}{\sigma\sqrt{2\pi}}e^{-\frac{1}{2}\left(\frac{x-\mu}{\sigma}\right)^2};$\\
        $\quad\quad\return (x, \rho)$\\
        $\quad\}, \lambda x. \frac{1}{\sigma\sqrt{2\pi}}e^{-\frac{1}{2}\left(\frac{x-\mu}{\sigma}\right)^2})$
        \end{tabular}&
        $\chi\{\flip_{\textit{strat}}\}$ &\begin{tabular}[t]{@{}l}
        $\coloneqq \lambda p. (\haskdo \{$\\
        $\quad\quad b \gets \sample~\flip_{\textit{strat}}(p);$\\
        $\quad\quad \llet~\rho = \mathbf{if}~b~\mathbf{then}~p~\mathbf{else}~1-p;$\\
        $\quad\quad\return (b, \rho)$\\
        $\quad\}, \lambda b. \mathbf{if}~b~\mathbf{then}~p~\mathbf{else}~1-p)$
        \end{tabular}\\
        % $\chi\{\normal_{\text{REINFORCE}}\}$ 
        % &\begin{tabular}[t]{@{}l}
        % $\coloneqq \lambda (\mu, \sigma). (\haskdo \{$\\
        % $\quad\quad x \gets \sample~(\normal_{\text{REINFORCE}}~(\mu, \sigma));$\\
        % $\quad\quad \llet~\rho = \frac{1}{\sigma\sqrt{2\pi}}e^{-\frac{1}{2}\left(\frac{x-\mu}{\sigma}\right)^2};$\\
        % $\quad\quad\return (x, \rho)$\\
        % $\quad\}, \lambda x. \frac{1}{\sigma\sqrt{2\pi}}e^{-\frac{1}{2}\left(\frac{x-\mu}{\sigma}\right)^2})$
        % \end{tabular}
        %  &
        % $\chi\{\flip_{\text{REINFORCE}}\}$ 
        % &\begin{tabular}[t]{@{}l}
        % $\coloneqq \lambda p. (\haskdo \{$\\
        % $\quad\quad b \gets \sample~(\flip_{\text{REINFORCE}}~p);$\\
        % $\quad\quad \llet~\rho = \mathbf{if}~b~\mathbf{then}~p~\mathbf{else}~1-p;$\\
        % $\quad\quad\return (b, \rho)$\\
        % $\quad\}, \lambda b. \mathbf{if}~b~\mathbf{then}~p~\mathbf{else}~1-p)$
        % \end{tabular}\\
        $\chi\{\return~t\}$ &$\coloneqq \return (\chi\{t\}, 1, \{\})$ &
        $\chi\{\observe~t_1~t_2\}$ &\begin{tabular}[t]{@{}l}
        $\coloneqq \haskdo \{$ \\
        $\quad\quad\llet w = \pi_2(\chi\{t_1\})(\chi\{t_2\});$\\ 
        $\quad\quad\score~w; \return~((), w, \{\})$ \\
        $\quad\}$
        \end{tabular}\\
        $\chi\{\sample~t_1~t_2\}$ &\begin{tabular}[t]{@{}l}
        $\coloneqq \haskdo \{$\\ 
        $\quad\quad (x,w) \gets \pi_1(\chi\{t_1\});$\\
        $\quad\quad \return (x, w, \{\chi\{t_2\} \mapsto x\})$\\
        $\quad \}$
        \end{tabular} &
        $\chi\{\haskdo~\{x \gets t; m\}\}$ &\begin{tabular}[t]{@{}l}
        $\coloneqq \haskdo \{$\\
        $\quad\quad (x, w, u_1) \gets \chi\{t\};$\\
        $\quad\quad (y, v, u_2) \gets \chi\{\haskdo\{m\}\};$\\
        $\quad\quad \mathbf{if} \textit{disj}(u_1, u_2)~\mathbf{then}$\\
        $\quad\quad\quad \return~(y,w \cdot v,u_1\mdoubleplus u_2)$\\
        $\quad\quad \mathbf{else}~\haskdo \{\score~0; \return~(y,0,\{\}))\}$\\ % TODO: score needs to be able to take 0 as an argument for this to work
        $\quad\}$
        \end{tabular}
   \end{tabular}\\
\end{tabular}
}
    \Description{Definitions of our program transformations for density evaluation and traced simulation.}
    \caption{Traced simulation and density evaluation as program transformations from $\lambda_{\text{Gen}}$ to $\lambda_{\text{ADEV}}$.}
    \label{fig:gen-transformations}
    %\vspace{-7mm}
\end{figure}

We now show how to take $\lambda_{\text{Gen}}$ programs and automatically compile them into $\lambda_{\text{ADEV}}$ programs that simulate traces or compute density functions. We introduce two program transformations, $\textbf{sim}$ and $\textbf{density}$, which take $\lambda_\text{Gen}$ terms to $\lambda_\text{ADEV}$ terms that implement the desired functionality. The intuition for these transformations is as follows:
\begin{itemize}[leftmargin=*]
    \item ($\textbf{sim}\{\cdot\}$) At a high level, to simulate traces, we just run the generative program and record the value of every sample we take into a growing trace data structure.\footnote{In the full specification for $\textbf{sim}$, along with the simulated trace, the transformed term also returns (with probability 1) the density of the term evaluated at the simulated trace.}
    \item ($\textbf{density}\{\cdot\}$) To compute the density of a trace, we execute the program, but fix the value of every primitive $\sample$ to equal the recorded value from the given trace, and multiply a running joint density by the density for the primitive. For $\observe$, we multiply the running joint density by the density of the primitive evaluated at the value provided to the $\observe$ statement. 
\end{itemize}    
For example, simulating from the unnormalized model in Fig.~\ref{fig:example_transcript} and then evaluating the density of the sampled trace --- using $\textbf{sim}$ and $\textbf{density}$ as introduced in Fig.~\ref{fig:gen-transformations} --- might produce:
{\small\begin{align*}
    \textbf{sim}\{\texttt{model}\} \rightsquigarrow (\{``x": 0.75, ``y": -2.2 \}, 1.3 \times 10^{-4}), \quad \textbf{density}\{\texttt{model}\}(\{``x": 0.75, ``y": -2.2\}) = 1.3\times 10^{-4}
\end{align*}}\noindent Note that although they are not usually formalized as program transformations or rigorously proven correct, the techniques we describe here are well-known and widely used in the implementations of PPLs, e.g. in Gen, ProbTorch, and Pyro.

%\vspace{0.1in}
\paragraph{Differentiability Properties of Densities.} In the context of variational inference, density functions must satisfy certain differentiability properties\textemdash with respect to the parameters being learned, and possibly with respect to the location at which the density is being queried, depending on the gradient estimators one wishes to apply. Previous work has developed specialized static analyses to determine smoothness properties of different parts of programs, in order to reason about gradient estimation for variational inference~\cite{lee_smoothness_2023}. A benefit of our approach is that it greatly simplifies this reasoning: the overall differentiability requirements for gradient estimation are enforced by the type system of the target-language ($\lambda_{\text{ADEV}}$), as in~\cite{lew_adev_2023}. We translate well-typed source-language programs into well-typed target-language density evaluators and trace simulators, which can be composed into well-typed variational objectives. This implies a non-trivial result: 
%Thus, we obtain a modular way of guaranteeing that 
the restrictions that $\lambda_{\text{Gen}}$ enforces on models and variational families are sufficient to ensure the necessary differentiability properties for unbiased estimation of variational objective gradients. Furthermore, these guarantees can be modularly extended to new variational objectives, gradient estimation strategies, and modeling language features. We discuss this further at the end of \S\ref{sec:adev}.

% Our $\lambda_{\text{Gen}}$ language enforces certain restrictions on model and guide programs via the distinction between the $\RR$ and $\RR^*$ types. This section's program transformations produce correct density functions and simulators that are well-typed $\lambda_{\text{ADEV}}$ programs. 
% What is novel in this section is that our compiler produces well-typed $\lambda_{\text{ADEV}}$ programs, and as such, the compiled programs are guaranteed to satisfy certain properties necessary for their sound within $\lambda_{\text{ADEV}}$ variational objectives. This is only possible because $\lambda_{\text{Gen}}$ enforces certain restrictions (primarily via the $\RR$ and $\RR^*$ types) 
% \begin{itemize}[leftmargin=*]
% \item \textit{Density functions are differentiable w.r.t. their parameters.} Our type system (namely, the distinction between $\RR$ and $\RR^*$ in the shared core) introduces certain restrictions on the $\lambda_{\text{Gen}}$ probabilistic programs that users can write. By formalizing density computation as a typed program transformation, we can guarantee that our source-language restrictions on generative models and guides yield density functions that adhere to our target-language restrictions.
% \item \textit{Trace simulators carry gradient estimator information.} Unlike programs in many PPLs, $\lambda_{\text{Gen}}$ programs are annotated with information about how gradients should be estimated
% \end{itemize}

\subsection{Compiling Differentiable Density Evaluators}

The density program transformation $\textbf{density}$ is given in Fig.~\ref{fig:gen-transformations} (top). As input, it processes a $\lambda_{\text{Gen}}$ term $t$ of type $\sigma \to G~\tau$: a generative program that has some (ground type) parameter or input. The result of the transformation is a $\lambda_{\text{ADEV}}$ term of type $\sigma \to \trace \to \RR$, which, given a parameter $\theta$ and a trace $u$, computes the density of $\sem{t}_1(\theta)$ at $u$.

The transformation is only defined for source programs of type $\sigma \to G~\tau$, but it is implemented using a helper transformation $\xi$ that is defined at all source-language types. The intended behavior of $\xi$ applied to a source-language term depends on the type of the term; we encode this type-dependent specification into a family of \textit{relations} $R_\tau^\xi$ indexed by types $\tau$ (see Fig.~\ref{fig:gen-transformations}). Each $R_\tau^\xi$ is a subset of $\sem{\tau} \times \sem{\xi\{\tau\}}$; if $(x, y) \in R_\tau^\xi$, it means that it is permissible for $\xi$ to translate a term denoting $x$ into a term denoting $y$. For example, $R_{D~\sigma}^\xi$ relates measures on $\sem{\sigma}$ to their density functions with respect to $\mathcal{B}_\sigma$, to encode that $\xi$ should transform primitive distribution terms into terms that compute primitive distribution densities. More interesting is the specification for $\xi$ on full probabilistic programs, of type $G~\tau$. Because the term under consideration may be only part of a larger program, $\xi$ must compute not just a density, but also a return value (for use later in the program) and a \textit{remainder} of its input trace, containing choices not consumed while processing the current term. The relation $R^\xi_{G~\tau}$ encodes this intuition using the function $\textit{split}_\mu$, which splits a trace $u$ into two parts: the largest subtrace of $u$ in the support of $\mu$ and the remaining subtrace, or, if no such subtrace exists, all of $u$ and the empty trace (see Appx.~\ref{appx:proof-details} for a formal definition).

To prove that $\mathbf{density}$ works correctly, we first prove that $\xi$ satisfies its intended specifications:

\begin{lemma}\label{lem:fundamental-xi}
    Let $\Gamma \vdash t : \tau$ be an open term of $\lambda_{\text{Gen}}$. Then $\xi\{\Gamma\} \vdash \xi\{t\} : \xi\{\tau\}$ is a well-typed open term of $\lambda_{\text{ADEV}}$, and $\forall (\gamma, \gamma') \in R^\xi_\Gamma, (\sem{\tau}(\gamma), \sem{\xi\{\tau\}}(\gamma')) \in R^\xi_\tau$.
\end{lemma}

The proof (in Appx.~\ref{appx:proof-details}) is by induction, but because the inductive hypothesis is different at each type $\tau$ (depending on our definition of $R_\tau^\xi$), it is an example of what is often called a \textit{logical relations} proof. Once it is proven, we are ready to prove the main correctness theorem for densities:

\noindent\begin{theorem}
    Let $\vdash t : \sigma \to G~\tau$ be a closed $\lambda_{\text{Gen}}$ term for some ground type $\sigma$. Then $\vdash {\mathbf{density}}\{t\} : \sigma \to \trace \to \RR$ is a well-typed $\lambda_{\text{ADEV}}$ term and for all $\theta \in \sem{\sigma}$, $\sem{\mathbf{density}\{t\}}(\theta)$ is a density function for $\pi_1(\sem{t})(\theta)$ with respect to $\mathcal{B}_\mathbb{T}$.
\end{theorem}
\begin{proof}
    Fix $\theta \in \sem{\sigma}$ and let $(\mu, f) = \sem{t}(\theta)$. By Lemma~\ref{lem:fundamental-xi}, we have that $(\sem{t}(\theta), \sem{\xi\{t\}}(\theta)) \in R^\xi_{G~\tau}$. Now consider a trace $u \in \mathbb{T}$. The $\mathbf{density}$ macro invokes $\xi\{t\}~\theta$ on $u$ to obtain a triple $(x, w, u')$. By the definition of $R^\xi_{G~\tau}$, we have that $w = \frac{d\mu}{d\mathcal{B}_\mathbb{T}}(u_1)$ and $u' = u_2$, for $(u_1, u_2) = \textit{split}_\mu(u)$. 
    Recall that if $u$ is in the support of $\mu$, \textit{or} if \textit{no subtrace} of $u$ is in the support of $\mu$, then $\textit{split}_\mu$ returns $(u, \{\})$, causing $\mathbf{density}$ to enter its $\mathbf{then}$ branch and return $w = \frac{d\mu}{d\mathcal{B}_\mathbb{T}}(u_1) = \frac{d\mu}{d\mathcal{B}_\mathbb{T}}(u)$ as desired. If $u$ is not in the support of $\mu$ but has a subtrace that is, then $\textit{split}_\mu$ returns that subtrace, along with a non-empty $u_2$. In this case the $\mathbf{else}$ branch of $\mathbf{density}$ correctly returns $0$ (because $u$ is not in the support).
\end{proof}

% split: if there exists a subtrace of u in the support, split out the rest 

\subsection{Compiling Differentiable Trace Simulators}
The simulation program transformation $\textbf{sim}$ is given in Fig.~\ref{fig:gen-transformations} (bottom). Like $\textbf{density}$, it processes as input a $\lambda_{\text{Gen}}$ term $t$ of type $\sigma \to G~\tau$, but it generates a term of type $P~(\trace \times \RR)$, satisfying the specification that the pushforward by $\pi_1$ is the original program's measure over traces, $\sem{t}_1(\theta)$, and that with probability 1, the second component is the density of $\sem{t}_1$ evaluated at the sampled trace. Like $\textbf{density}$, $\textbf{sim}$ is implemented using a helper macro $\chi$ defined at all types. Fig.~\ref{fig:gen-transformations} presents the logical relations $R^\chi_\tau$ specifying the helper's intended behavior on terms of type $\tau$. These relations are simpler than those for $\xi$; for example, on terms of type $G~\tau$, $\chi$ has almost the same specification as $\textbf{sim}$ itself, except that it must also compute a return value. 

One feature of $\chi$ that is worth noting is its translations of primitives. If a primitive used within a traced probabilistic program is annotated with a gradient estimation strategy, then the translated program uses the same annotated primitive, and then computes a density. This is only well-typed because the \textit{density functions} of primitives that return $\RR$ values (i.e., not $\RR^*$ values) are smooth. This is not a requirement of ADEV in general, but we require it in order to automate differentiable traced simulation.

To prove correctness, we again begin by showing the helper is sound:

\begin{lemma}\label{lem:fundamental-chi}
    Let $\Gamma \vdash t : \tau$ be an open term of $\lambda_{\text{Gen}}$. Then $\chi\{\Gamma\} \vdash \chi\{t\} : \chi\{\tau\}$ is a well-typed open term of $\lambda_{\text{ADEV}}$, and $\forall (\gamma, \gamma') \in R^\chi_\Gamma, (\sem{\tau}(\gamma), \sem{\chi\{\tau\}}(\gamma') \in R^\chi_\tau$.
\end{lemma}

The proof is again by logical relations, and can be found in Appx.~\ref{appx:proof-details}. We can then prove the correctness of $\textbf{sim}$ itself:
\noindent\begin{theorem}
    Let $\vdash t : G~\tau$ be a closed $\lambda_{\text{Gen}}$ term. Then $\vdash {\mathbf{sim}}\{t\} : P~(\trace \times \RR)$ is a well-typed $\lambda_{\text{ADEV}}$ term and $\sem{\mathbf{sim}\{t\}}$ is the pushforward of $\pi_1(\sem{t})$ by the function $\lambda u. \left(u, \frac{d\sem{t}_1}{d\mathcal{B}_{\mathbb{T}}}(u)\right)$.
\end{theorem}
\begin{proof}
    Fix $\theta \in \sem{\sigma}$ and let $(\mu, f) = \sem{t}(\theta)$. By Lemma~\ref{lem:fundamental-chi}, we have that $(\sem{t}(\theta), \sem{\chi\{t\}}(\theta)) \in R^\chi_{G~\tau}$. The $\mathbf{sim}$ macro invokes $\chi\{t\} \theta$ to obtain a triple $(x, w, u)$, but only returns $(u, w)$. Observe that the requirements placed by $R^\chi_{G~\tau}$ on $w$ and $u$ are precisely the conditions we aim to prove here.
\end{proof}

\section{Variational Inference via Differentiable Probabilistic Programming}
\label{sec:vi}

\label{sec:adev}
\begin{figure}[t]
\scriptsize{
\begin{tabular}{ll}
    % Overall spec
   %{Spec}  & 
   \framebox{\textcolor{blue}{\textbf{\textit{Syntax}}} \;\; \begin{tabular}[t]{@{}l}$\vdash t : \RR^n \to \eRR \implies \vdash \mathbf{adev}\{t\} : \RR^n \to \RR^n \to \eRR$\\
   $\vdash t : \RR^n \to \eRR \implies \vdash \mathbf{dom}\{t\}_{(\theta, i)} : \RR \times \RR^* \to \RR$\end{tabular} \quad \;\textcolor{blue}{\textbf{\textit{Semantics}}} \;\; \begin{tabular}[t]{@{}l}
   $(\forall \theta \in \RR^n, i \in \underline{n}. \sem{\mathbf{dom}\{t\}_{(\theta, i)}} \text{ locally dom'd}) \implies$ \\
   $\mathbb{E}_{x \sim \sem{\mathbf{adev}\{t\}}(\mathbf{\theta}, \mathbf{v})}[x] = \left(\nabla_\theta \int_{\RR} x~\sem{t}(\mathbf{\theta}, dx)\right)^T \mathbf{v}$\end{tabular}} \\
   % Wrapper implementation
   %\rule{0pt}{4ex}{Wrapper}  & 
   \rule{0pt}{4ex}{$\mathbf{adev}\{t\} \coloneqq \lambda \theta. \lambda \mathbf{v}. \mathbb{E} (\haskdo \{(y, y') \gets \mathcal{D}\{t\}~((\theta_1, \mathbf{v}_1),\dots,(\theta_n, \mathbf{v_n})); \return~y'\})$}\\
   %& 
   $\mathbf{dom}\{t\}_{(\theta, i)} \coloneqq \lambda (\phi, x). \pi_2(\pi_2(\mathcal{D}\{t\} ((\theta_1, 0), \dots, (\theta_{i-1}, 0), (\phi, 1), (\theta_{i+1}, 0), \dots, (\theta_n, 0))~ x$\\
   % Spec for helper
   %\rule{0pt}{6ex}{Helper} $\mathcal{D}$ & 
   \rule{0pt}{6ex}{\framebox{\textcolor{blue}{\textbf{\textit{Syntax}}}\;\;$\Gamma \vdash_{\text{ADEV}} t : \tau \implies \mathcal{D}\{\Gamma\} \vdash_{\text{ADEV}} \mathcal{D}\{t\} : \mathcal{D}\{\tau\}$\quad \quad\textcolor{blue}{\textbf{\textit{Semantics}}}\;\;$\forall (\gamma, \gamma') \in R^{\mathcal{D}}_{\Gamma}, (\sem{t} \circ \gamma, \sem{\mathcal{D}\{t\}} \circ \gamma') \in R^{\mathcal{D}}_\tau$\quad\quad\quad\quad\;\;}}%\\
   \end{tabular}

   % Implementation on types
    %\rule{0pt}{4ex}\;\textit{on types} & 
    \begin{center}
    \vspace{2mm}\textit{\textbf{Type Translation and Logical Relations}}\vspace{2mm}
    \begin{tabular}[t]{@{}l@{}l|l@{}l}
        $\mathcal{D}\{\RR\}$ &$\coloneqq \RR \times \RR$ & $R^\mathcal{D}_\RR$ &$\coloneqq \{(f, g) \mid g = \lambda r. (f(r), f'(r))\}$\\
        $\mathcal{D}\{\sigma\}$ &$\coloneqq \sigma \quad\quad\quad (\sigma \in \{\RR^*, \mathbb{B}, \text{Str}\})$ & $R^\mathcal{D}_{\sigma}$ &$\coloneqq \{(f, f) \mid f\text{ constant}\}$\\
        $\mathcal{D}\{\tau_1 \times \tau_2\}$&$\coloneqq \mathcal{D}\{\tau_1\} \times \mathcal{D}\{\tau_2\}$ & $R^\mathcal{D}_{\tau_1 \times \tau_2}$ &$\coloneqq \{(f,g) \mid \forall i \in \{1, 2\}. (\pi_i \circ f, \pi_i \circ g) \in R^\mathcal{D}_{\tau_i}\}$\\
        $\mathcal{D}\{\tau_1 \to \tau_2\}$&$\coloneqq \mathcal{D}\{\tau_1\} \to \mathcal{D}\{\tau_2\}$ & $R^\mathcal{D}_{\tau_1 \to \tau_2}$ &$\coloneqq \{(f,g) \mid \forall (x,y) \in R^\mathcal{D}_{\tau_1}, (\lambda r. f(r)(x(r)),g(r)(y(r))) \in R^\mathcal{D}_{\tau_2}\}$\\
        % $\mathcal{D}\{\RR^*\}$ &$\coloneqq \RR^*$ & $R^\mathcal{D}_{\RR^*}$ &$\coloneqq \{(f, f) \mid f\text{ constant}\}$\\
        % $\mathcal{D}\{\mathbb{B}\}$ &$\coloneqq \mathbb{B}$ & $R^\mathcal{D}_\mathbb{B}$ &$\coloneqq \{(f, f) \mid f \text{ constant}\}$ \\
        % $\mathcal{D}\{\text{Str}\}$ &$\coloneqq \text{Str}$ & $R^\mathcal{D}_\text{Str}$ &$\coloneqq \{(f, f) \mid f \text{ constant}\}$ \\
        $\mathcal{D}\{\trace\}$ &$\coloneqq \trace$ & $R^\mathcal{D}_\trace$ &$\coloneqq \{(f, g) \mid \forall k \in \sem{\text{Str}}. (\lambda r. f(r)[k], \lambda r. g(r)[k]) \in R^{\mathcal{D}}_\sigma\}$ \\
        $\mathcal{D}\{D~\sigma\}$&$\coloneqq \mathcal{D}\{P~\sigma\}$ & $R^\mathcal{D}_{D~\sigma}$ &$\coloneqq R^\mathcal{D}_{P~\sigma}$\\
        $\mathcal{D}\{\eRR\}$&$\coloneqq P~(\RR \times \RR) \times (\RR^* \to \RR \times \RR)$ &$R^\mathcal{D}_{\eRR}$ &\begin{tabular}[t]{@{}l} 
        $\coloneqq \{(\mu, \nu) \mid 
        \forall \theta. \int_{\RR} h_1(\theta)(s) ds = \int_{\RR} x \mu(\theta, dx) = \mathbb{E}_{x \sim {\pi_1}_*(\pi_1 \circ \nu)(\theta)}[x] \wedge$ \\
        \quad\quad$\forall \theta. \int_{\RR} h_2(\theta)(s) ds = \mathbb{E}_{x \sim {\pi_2}_*({\pi_1}_*\nu)(\theta)}[x] \wedge$\\
        \quad\quad$(\lambda \theta. \lambda s. h_1(\theta)(s), \lambda\theta.\lambda s. (h_1(\theta)(s), h_2(\theta)(s)) \in R_{\RR^* \to \RR}$\\
        % $(\lambda r. \mathbb{E}_{x \sim \mu(r)}[x], \lambda r. (\mathbb{E}_{(x,y) \sim \nu(r)}[x], \mathbb{E}_{(x,y) \sim \nu(r)}[y])) \in R_\RR^\mathcal{D}\}$\\
        \text{where } $h_i \coloneqq \lambda \theta. \lambda s. \pi_i((\pi_2 \circ \nu)(\theta)(s))$
        \end{tabular}\\
        $\mathcal{D}\{P~\tau\}$&$\coloneqq (\mathcal{D}\{\tau\} \to \mathcal{D}\{\eRR\}) \to \mathcal{D}\{\eRR\}$ & $R^\mathcal{D}_{P~\tau}$ &$\coloneqq \{(\mu, \nu) \mid (\lambda r. \lambda k. \lambda U. \mathbb{E}_{x \sim \mu(r)}[k(x, U)], \nu) \in R^{\mathcal{D}}_{(\tau \to \eRR) \to \eRR} \}$
    \end{tabular}
    \end{center}
   % \rule{0pt}{4ex}\quad\textit{on types} & $\xi\{\sigma\} \coloneqq \sigma \;\mid\; \xi\{D~\sigma\} \coloneqq \sigma \to \RR \;\mid\; \xi\{\tau_1 \to \tau_2\} \coloneqq \xi\{\tau_1\} \to \xi\{\tau_2\}$ \\
   % Implementation on terms
   %\rule{0pt}{4ex}\;\textit{on terms} & 
   \vspace{2mm}\textit{\textbf{Term Translation}}\vspace{2mm}
   \begin{tabular}[t]{@{}l@{}ll@{}l}
   $\mathcal{D}\{()\}$ &$\coloneqq ()$ &
   $\mathcal{D}\{\{\}\}$ &$\coloneqq \{\}$ \\
   $\mathcal{D}\{c\}$ &$\coloneqq c_\mathcal{D}$ &
   $\mathcal{D}\{x\}$ &$\coloneqq x$ \\
   $\mathcal{D}\{\lambda x. t\}$&$\coloneqq \lambda x. \mathcal{D}\{t\}$ &
   $\mathcal{D}\{t_1~t_2\}$&$\coloneqq \mathcal{D}\{t_1\}~\mathcal{D}\{t_2\}$ \\
   $\mathcal{D}\{(t_1,t_2)\}$&$\coloneqq (\mathcal{D}\{t_1\}, \mathcal{D}\{t_2\})$&
   $\mathcal{D}\{\pi_i~t\}$&$\coloneqq \pi_i~\mathcal{D}\{t\}$ \\
   $\mathcal{D}\{b\}~(b \in \{\mathbf{T},\mathbf{F}\})$ &$\coloneqq b$ &
   $\mathcal{D}\{\mathbf{if}~t~\mathbf{then}~t_1~\mathbf{else}~t_2\}$&$\coloneqq \mathbf{if}~\mathcal{D}\{t\}~\mathbf{then}~\mathcal{D}\{t_1\}~\mathbf{else}~\mathcal{D}\{t_2\}$ \\
   $\mathcal{D}\{\{t_1 \mapsto t_2\}\}$ &$\coloneqq \{\mathcal{D}\{t_1\} \mapsto \mathcal{D}\{t_2\}\}$ &
  $\mathcal{D}\{t_1 \mdoubleplus t_2\}$ &$\coloneqq \mathcal{D}\{t_1\} \mdoubleplus \mathcal{D}\{t_2\}$ \\
   $\mathcal{D}\{t_1[t_2]\}$ &$\coloneqq \mathcal{D}\{t_1\}[\mathcal{D}\{t_2\}]$ &
   $\mathcal{D}\{\sample~t\}$&$\coloneqq \mathcal{D}\{t\}$ \\
   $\mathcal{D}\{r\}$ &$\coloneqq (r, 0)$ &
   $\mathcal{D}\{\score~t\}$ &\begin{tabular}[t]{@{}l}
    $\coloneqq \lambda \kappa. (\haskdo \{y \gets \pi_1(\kappa(()));$\\
    \quad\quad\quad $\return (\mathcal{D}\{t\} \times_\mathcal{D} y)\},$\\
    \quad\quad\quad\quad $\lambda s. (\pi_2(\kappa(()))~s) \times_\mathcal{D} \mathcal{D}\{t\})$
    \end{tabular}\\
   $\mathcal{D}\{\return~t\}$&$\coloneqq \lambda \kappa. \kappa(\mathcal{D}\{t\})$ &
   $\mathcal{D}\{\haskdo~\{x \gets t; m\}\}$&$\coloneqq \lambda\kappa.\mathcal{D}\{t\}(\lambda x. \mathcal{D}\{\haskdo\{m\}\}(\kappa))$\\
   & & $\mathcal{D}\{\mathbb{E}~t\}$&\begin{tabular}[t]{@{}l}$\coloneqq\mathcal{D}\{t\}~(\lambda x.(\textbf{return}~x,\lambda s. x))$
    \end{tabular}\\
   $\mathcal{D}\{\mathbf{normal}_{\text{REPARAM}}\}$&\begin{tabular}[t]{@{}l}$\coloneqq\lambda ((\mu,\mu'),(\sigma, \sigma')). \lambda \kappa. (\haskdo \{$\\
   $\quad\quad \epsilon \gets \sample~(\mathbf{normal}_{\text{REPARAM}}(0, 1));$\\
   $\quad\quad \pi_1(\kappa((\sigma\epsilon + \mu, \sigma'\epsilon + \mu')))$\\
   $\quad\}, \lambda s. \dots)$
   \end{tabular} & 
   $\mathcal{D}\{\mathbf{normal}_{\text{REINFORCE}}\}$&\begin{tabular}[t]{@{}l}$\coloneqq\lambda ((\mu,\mu'), (\sigma, \sigma')). \lambda \kappa. (\haskdo \{$\\
   $\quad\quad x \gets \sample~ (\mathbf{normal}_{\text{REINFORCE}}(\mu, \sigma));$\\
   $\quad\quad (y, y') \gets \pi_1(\kappa((x, 0)));$\\
   $\quad\quad \llet~l' = \sigma'(\frac{1}{\sigma} + \frac{(y-\mu)^2}{\sigma^3}) + \mu' \frac{y-\mu}{\sigma^2};$\\
   $\quad\quad \return~(y, y'+yl')$\\
   $\quad\}, \lambda s. \dots)$
   \end{tabular}\\
   $\mathcal{D}\{\mathbf{flip}_{\text{ENUM}}\}$&\begin{tabular}[t]{@{}l}$\coloneqq\lambda (p,p'). \lambda \kappa. (\haskdo \{$\\
   $\quad\quad (y_T, y_T') \gets \pi_1(\kappa(\mathbf{T}));$\\
   $\quad\quad (y_F, y_F') \gets \pi_1(\kappa(\mathbf{F}));$\\
   $\quad\quad \llet~y=py_T + (1-p)y_F;$\\
   $\quad\quad \llet~y_1'=p'y_T + py_T';$\\
   $\quad\quad \llet~y_2'=(1-p)y_F' - p'y_F;$\\
   $\quad\quad \return~(y, y_1' + y_2')$\\
   %$\quad\quad x \gets \mathbf{flip}_{\text{ENUM}}(\mu, \sigma);$\\
   %$\quad\quad (y, dy) \gets \kappa((x, 0));$\\
   %$\quad\quad \llet~dl = d\sigma \times (\frac{1}{\sigma} + \frac{(y-\mu)^2}{\sigma^3}) + d\mu \times \frac{y-\mu}{\sigma^2};$\\
   %$\quad\quad \return~(y, dy+y\times dl)$\\
   $\quad\}, \lambda s. \dots)$
   \end{tabular} & 
   $\mathcal{D}\{\mathbf{flip}_{\text{REINFORCE}}\}$&\begin{tabular}[t]{@{}l}$\coloneqq\lambda (p,p'). \lambda \kappa. (\haskdo \{$\\
   $\quad\quad b \gets \sample~ (\mathbf{flip}_{\text{REINFORCE}}(p));$\\
   $\quad\quad (y, y') \gets \pi_1(\kappa(b));$\\
   $\quad\quad \llet~l'=\mathbf{if}~b~\mathbf{then}~\frac{p'}{p}~\mathbf{else}~\frac{p'}{p-1};$\\
   $\quad\quad \return (y, y' + yl')$\\
   $\quad\}, \lambda s. \dots)$ \\
   \end{tabular} \\

   % $\mathcal{D}\{\mathbf{normal}_{\text{REPARAM}}\}$&$\coloneqq \lambda(\mu, d\mu).\lambda(\sigma,d\sigma).\lambda\kappa.\haskdo \{\epsilon \gets \mathbf{normal}_{\text{REPARAM}}(0, 1); \kappa((\sigma \times \epsilon + \mu, d\sigma \times \epsilon + d\mu))\}$\\
   % $\mathcal{D}\{\mathbf{normal}_{\text{REINFORCE}}\}$&$\coloneqq \lambda(\mu, d\mu).\lambda(\sigma,d\sigma).\lambda\kappa.\haskdo \{x \gets \mathbf{normal}_{\text{REINFORCE}}(\mu, \sigma); (l, dl) \gets \kappa((x,0)); (p, dp) \gets \mathcal{D}\{\log\}; \return (l, dl + l \times dp)\}$\\
   \end{tabular}
   
%\end{tabular}
}
    \Description{ADEV program transformation}
    \caption{Monte Carlo gradient estimation as a program transformation from $\lambda_{\text{ADEV}}$ to $\lambda_{\text{ADEV}}$.}
    \label{fig:adev-transformation}
    \vspace{-2mm}
\end{figure}

As we saw in \S\ref{sec:overview}, the density and trace simulation programs automated in the previous section can be used to construct larger $\lambda_\text{ADEV}$ programs implementing variational objectives. Once we have a $\lambda_{\text{ADEV}}$ program representing our objective function, we need to differentiate it. Conventional AD systems do not correctly handle randomness in objective functions, or the expectation operator $\mathbb{E}$, and will produce biased gradient estimators when applied naively~\cite{lew_adev_2023}. For example, standard AD has no way of propagating derivative information through a primitive like $\text{flip}:[0,1]\rightarrow P~\mathbb{B}$ (to do so, one would need to define the notion of \textit{derivative of a Boolean with respect to the probability that it was heads}). The ADEV algorithm~\cite{lew_adev_2023} is designed to handle these features, and can be used to derive unbiased gradient estimators automatically.

%\vspace{0.05in}
\paragraph{Extending ADEV with Traces and Unnormalized Measures.} Fig.~\ref{fig:adev-transformation} gives the ADEV program transformation, extended to handle new datatypes (traces) and unnormalized measures (due to $\score$). Fig.~\ref{fig:adev-transformation} shows two top-level transformations, $\mathbf{adev}$ and $\mathbf{dom}$. The $\mathbf{dom}$ transformation exists solely for analytical purposes, to produce a term that must satisfy a \textit{local domination} condition in order for gradient estimates to be unbiased:\footnote{We have omitted several terms from Fig.~\ref{fig:adev-transformation}, denoted with ``$\dots$''. These terms are as in \cite[Fig. 26]{lew_adev_2023}, and do not affect the behavior of the $\mathbf{adev}$ transformation, only $\mathbf{dom}$.}
\begin{definition}[locally dominated]
A function $f : \RR \times \RR \to \RR$ is locally dominated if, for every $\theta \in \RR$, there is a neighborhood $U(\theta) \subseteq \RR$ of $\theta$ and an integrable function $m_{U(\theta)} : \RR \to [0, \infty)$ such that $\forall \theta' \in U(\theta), \forall x \in \RR, |f(\theta', x)| \leq m_{U(\theta)}(x)$. 
\end{definition}
\noindent Under this mild assumption, ADEV produces correct unbiased gradient estimators:
\begin{theorem}
\label{thm:unbiased-adev}
Let $\vdash t : \RR^n \to \eRR$ be a closed $\lambda_{\text{ADEV}}$ term, satisfying the following preconditions:
\begin{enumerate}
    \item $\int_\RR x~\sem{t}(\theta, dx)$ is finite for every $\theta \in \RR^n$.
    \item $\sem{\mathbf{dom}\{t\}_{(\theta, i)}}$ is locally dominated for every $\theta \in \RR^n$ and $i \in \underline{n}$.
\end{enumerate}
Then $\vdash {\mathbf{adev}}\{t\} : \RR^n \to \RR^n \to \eRR$ is a well-typed $\lambda_{\text{ADEV}}$ term, satisfying the following properties:
\begin{itemize}
    \item $\sem{\mathbf{adev}\{t\}}(\theta, \mathbf{v})$ is a probability measure with finite  expectation for all $\theta, \mathbf{v} \in \RR^n$.
    \item $\sem{t}$ is differentiable and $\mathbb{E}_{x \sim \sem{\mathbf{adev}\{t\}}(\theta, \mathbf{v})}[x] = \left(\nabla_\theta \int_\RR x~\sem{t}(\theta, dx) \right)^T \mathbf{v}$.
\end{itemize}
\end{theorem}
%\begin{proof}
%\end{proof}

\noindent The proof, as in the previous section, is a logical relations proof, and Fig.~\ref{fig:adev-transformation} gives the logical relations.
It extends the proof of~\cite{lew_adev_2023} with new cases, for $\mathbf{score}$, as well as for trace operations.

%\vspace{0.05in}
\paragraph{Static Checks and Unbiasedness.}
Probabilistic programming languages like Pyro and Gen use specialized logic — also going beyond ordinary AD — to unbiasedly estimate derivatives of particular variational objectives like the ELBO, for user-defined models and variational families. But in these systems, biased gradient estimates can still arise if the user's model or variational family violates assumptions made by the PPL backend. For example:
\begin{itemize}[leftmargin=*]
    \item The user's program may sample $x$ from a normal distribution, then branch on whether $x < k$ for some constant threshold $k$, in order to decide on the distribution of another random variable $y$. Pyro's default gradient estimation strategy assumes that the joint density of the model is differentiable with respect to the values of Gaussian random variables like $x$, but this assumption is violated by the user's program, because the joint density $p(x, y)$ is of the form $p(x) \cdot ([x > k] \cdot p_1(y \mid x) + [x \leq k] \cdot p_2(y \mid x))$, which may be discontinuous at $x = k$.
    \item Gen's default gradient estimation strategy does not place differentiability assumptions on the user's program, but does assume that the support of each primitive distribution does not depend on learned parameters. If the user's program samples from a uniform distribution with learned endpoints $a$ and $b$, this assumption will be violated, and Gen's gradient estimates will be biased.
\end{itemize}

In our design, by contrast, there is no default gradient estimation strategy. Rather, the user chooses a different gradient estimation strategy for each primitive, and the overall gradient estimator is automated compositionally. Crucially, different versions of primitives (employing different gradient estimation strategies) have static types that enforce the key assumptions necessary for their unbiasedness.
 For example:
\begin{itemize}[leftmargin=*]
    \item The $\texttt{normal}_\text{REPARAM}$  primitive has type $\mathbb{R} \times \mathbb{R}_{>0} \to D~\mathbb{R}$, so if $x$ is drawn from $\texttt{normal}_\text{REPARAM}$, then the type of the variable $x$ is $\mathbb{R}$. The type of $<$ is $\mathbb{R}^* \times \mathbb{R}^* \to \mathbb{B}$, and so the expression $x < k$ is ill-typed. Thus, the types enforce that the smoothness assumptions of the reparameterization estimator hold for the user's program, if the user chooses to apply this estimator.

    \item In our system, the uniform distribution with custom endpoints is $\mathbf{uniform} : \RR^* \times \RR^* \to D~\RR^*$, which behaves like a safe version of Gen's uniform distribution\textemdash its output can be used non-smoothly, but its bounds must not depend directly on learned parameters. (The bounds may still depend on, e.g., Gaussian random choices with learned means.)
\end{itemize}

\noindent Our smoothness-typing discipline in $\lambda_\text{Gen}$ is similar to, but slightly different from, that in~\citet{lew_adev_2023}. As an example, their version of ADEV can support a primitive $\text{uniform}_\text{MVD} : \RR \times \RR \to P~\RR$, but we cannot introduce such a primitive that returns $D~\RR^*$ or $D~\RR$. This is because the program transformation $\xi$ would need to translate such primitives into code for computing the uniform distribution's \textit{density} as a smooth function of its endpoints\textemdash which is not possible, since the density of the uniform distribution is discontinuous at the endpoints.

These static checks are necessary for proving  unbiasedness, as without them, we could easily produce estimators that do not respect the restrictions of the  estimation strategies they employ.

\section{Correctness of Gradient Estimation for Variational Inference}
\label{sec:correctness}

We can put together the results of the previous two sections to prove a general correctness theorem for our approach to variational inference. %The goal of our theorem is to characterize the conditions required to ensure that our compiler constructs an unbiased gradient estimator for $\lambda_\text{ADEV}$ programs which make use of $\textbf{sim}$ and $\textbf{density}$ to represent variational objectives.
%\vspace{0.1in}
%\begin{mdframed}[leftmargin=0pt,nobreak=true]
\noindent Suppose the user has written the following three programs:
\begin{itemize}[leftmargin=20pt]
    \item A \textbf{model program}: a closed $\lambda_\text{Gen}$ program $P : \RR^n \to G~\tau_1$.
    \item A \textbf{variational program}: a closed $\lambda_\text{Gen}$ program $Q : \RR^m \to G~\tau_2$.
    \item An \textbf{objective program}: a closed $\lambda_\text{ADEV}$ program $L : \mathbb{R}^{n+m} \to \eRR$, of the form $$L = \lambda (\theta, \phi). \mathbf{let}~(p,\mathbb{P},q,\mathbb{Q}) = (\textbf{density}\{P\}~\theta, \mathbf{sim}\{P\}~\theta, \mathbf{density}\{Q\}~\phi, \mathbf{sim}\{Q\}~\phi)~\mathbf{in}~F,$$ where $\theta$ and $\phi$ do not occur free in $F$. %$for some $F$.
    This program encodes the variational objective $\mathcal{F}(\mu_P, \mu_Q) = \int x \cdot \sem{F}(\gamma(\mu_P, \mu_Q), dx)$, where $\gamma(\mu_P, \mu_Q)$ is an environment mapping $p$ and $q$ to  densities of $\mu_P$ and $\mu_Q$ (with respect to $\mathcal{B}_\mathbb{T}$) and $\mathbb{P}$ and $\mathbb{Q}$ to simulators for $\mu_P$ and $\mu_Q$.
    %that encodes a variational objective $\mathcal{F}(\mu_P, \mu_Q)$, and is allowed to refer to \textit{density functions} $p, q : \trace \to \RR$ and \textit{simulators} $\mathbb{P}, \mathbb{Q} : P~(\trace \times \RR)$ for the measures $\mu_P$ and $\mu_Q$. (The symbols $p$, $q$, $\mathbb{P}$, and $\mathbb{Q}$ are made available through the context $\Gamma$. For a given $\mu_P$ and $\mu_Q$, we can construct the environment $\gamma(\mu_P, \mu_Q) = (p \mapsto \frac{d\mu_P}{d\mathcal{B}_\mathbb{T}}, q \mapsto \frac{d\mu_Q}{d\mathcal{B}_\mathbb{T}}, \mathbb{P} \mapsto (id \otimes \frac{d\mu_P}{d\mathcal{B}_\mathbb{T}})_*\mu_P, \mathbb{Q} \mapsto (id \otimes \frac{d\mu_Q}{d\mathcal{B}_\mathbb{T}})_*\mu_Q)$. Then the program $F$ encodes the objective $\mathcal{F}$ in that $\mathcal{F}(\mu_P, \mu_Q) = \int_{\RR} x \sem{F}(\gamma(\mu_P, \mu_Q), dx)$.)
\end{itemize}

\noindent The user wants to find $\theta$ and $\phi$ that optimize $\mathcal{F}(\sem{P}_1(\theta), \sem{Q}_1(\phi))$. Our result is that our system estimates derivatives of this objective unbiasedly, under mild technical conditions: %We compile a closed $\lambda_\text{ADEV}$ term $\vdash L : \mathbb{R}^{n + m} \to \eRR$ to represent this function: $$L = \lambda (\theta, \phi). \mathbf{let}~(p,\mathbb{P},q,\mathbb{Q}) = (\textbf{density}\{P\}~\theta, \mathbf{sim}\{P\}~\theta, \mathbf{density}\{Q\}~\phi, \mathbf{sim}\{Q\}~\phi)~\mathbf{in}~F.$$
% \noindent We can then state our unbiasedness result:
\vspace{0.05in}
\begin{mdframed}
\vspace{-0.1in}
\begin{theorem}
% Let:
% \begin{itemize}
%     \item $\vdash P : \RR^n \to G~\tau_1$, the \textit{model program}, be a closed $\lambda_{\text{Gen}}$ term;
%     \item $\vdash Q : \RR^m \to G~\tau_2$, the \textit{variational program}, be a closed $\lambda_{\text{Gen}}$ term;
%     \item $\Gamma$ be the $\lambda_{\text{ADEV}}$ environment with mappings for $p : \trace \to \RR, q : \trace \to \RR, \mathbb{P} : P (\trace \times \RR), \mathbb{Q} : P (\trace \times \RR)$;
%     \item  $\Gamma \vdash t : \eRR$, the \textit{objective program}, be an open $\lambda_{\text{ADEV}}$ term in context $\Gamma$;
%     \item $\mathcal{F}(\mu_P, \mu_Q)$, be the \textit{objective function}, with \begin{equation*}\mathcal{F}(\mu_P, \mu_Q)=\int_{\RR} x~\sem{t}(\frac{d\mu_P}{d{\mathcal{B}_\mathbb{T}}}, \frac{d\mu_Q}{d\mathcal{B}_\mathbb{T}}, (\textit{id} \otimes \frac{d\mu_P}{d{\mathcal{B}_\mathbb{T}}})_*\mu_P, (\textit{id} \otimes \frac{d\mu_Q}{d{\mathcal{B}_\mathbb{T}}})_*\mu_Q, dx)\end{equation*}
%     \item $\mathcal{L}(\theta, \phi)$, be the \textit{parametric objective function}, equal to $\mathcal{F}(\sem{P}_1(\theta), \sem{Q}_1(\phi))$
%     \item $s \coloneqq ~\vdash \lambda (\theta, \phi). t[\{p\mapsto\textbf{density}\{P\}~\theta,q\mapsto\textbf{density}\{Q\}~\phi,\mathbb{P}\mapsto\textbf{sim}\{P\}~\theta,\mathbb{Q}\mapsto\textbf{sim}\{Q\}~\phi\}] : \RR^{n+m} \to \eRR$ be the \textit{compiled parametric objective program}, a closed term of $\lambda_{\text{ADEV}}$.
% \end{itemize}
Let $\mathcal{L}(\theta, \phi) = \mathcal{F}(\sem{P}_1(\theta), \sem{Q}_1(\phi))$. If for all $\theta \in \RR^n$ and $\phi \in \RR^m$, $\mathcal{L}(\theta, \phi)$ is finite and $\sem{\textbf{dom}\{L\}_{((\theta, \phi), i)}}$ is locally dominated for each $i \in \underline{n+m}$, then for all $\mathbf{v} \in \RR^{n+m}$, $\sem{\mathbf{adev}\{L\}}((\theta, \phi), \mathbf{v})$ is a probability measure with finite expectation and \begin{equation*}\mathbb{E}_{x \sim \sem{\mathbf{adev}\{L\}}((\theta, \phi), \mathbf{v})}[x] = (\nabla_{(\theta, \phi)} \mathcal{L}(\theta, \phi))^T \mathbf{v}.\end{equation*}
\end{theorem}
\end{mdframed}
%\end{mdframed}
\vspace{0.05in}
\noindent To understand the guarantee the theorem gives more concretely, consider the following objective program defining the ELBO (Eqn.~\ref{eqn:elbo}):
{\small$$\textbf{ELBO}\coloneqq \lambda(\theta,\phi). \mathbb{E}(\haskdo~\{(z, w_q)\gets(\textbf{sim}\{Q\}~\phi); \mathbf{let}~w_p = (\textbf{density}\{P\}~\theta)~z; \return\log w_p - \log w_q\})$$}
%This program can be understood as encoding the (mathematical) variational objective $$\mathcal{F}(\mu_P, \mu_Q) = \mathbb{E}_{z \sim \mu_Q}\left[\log \frac{d\mu_P}{d\mu_Q}(z)\right].$$  
We can write it in the form required by the theorem as follows:
\begin{align*}
    L = \lambda(\theta,\phi). \mathbf{let}~&(p,\mathbb{P},q,\mathbb{Q}) = (\textbf{density}\{P\}~\theta, \mathbf{sim}\{P\}~\theta, \mathbf{density}\{Q\}~\phi, \mathbf{sim}\{Q\}~\phi)~\mathbf{in}\\
    &\mathbb{E} (\haskdo \{(z,w_q) \gets \mathbb{Q}; \mathbf{return}~(\log p(z) - \log w_q)\})
\end{align*}
% {\small$$L = \lambda(\theta,\phi). \mathbf{let}~(p,\mathbb{P},q,\mathbb{Q}) = (\textbf{density}\{P\}~\theta, \mathbf{sim}\{P\}~\theta, \mathbf{density}\{Q\}~\phi, \mathbf{sim}\{Q\}~\phi)~\mathbf{in}~\mathbb{E} (\haskdo \{(z,w_q) \gets \mathbb{Q}; \mathbf{let}~w_p = p~z; \mathbf{return}~(\log w_p - \log w_q)\})$$}
% \begin{align*}
% \Gamma &= \{q \mapsto \textbf{density}\{Q\}~\phi, p\mapsto\textbf{density}\{P\}~\theta,\mathbb{Q}\mapsto\textbf{sim}\{Q\}~\phi,\mathbb{P}\mapsto\textbf{sim}\{P\}~\theta\} \\
% t &= \mathbb{E}(\haskdo~\{(z, w_q)\gets\mathbb{Q}; w_p\gets~p~z; \return\log w_p - \log w_q\})
% \end{align*}
%With the local domination condition, we can then apply Theorem~\ref{thm:unbiased-adev} to conclude that the denotation of the gradient estimator program $\sem{\textbf{adev}\{s\}}$ has the correct unbiasedness property with respect to the objective function, for all $\theta$ and $\phi$.
% TODO: Maybe move Section 5.1 here?
The theorem then establishes that, under mild technical conditions, applying $\mathbf{adev}$ to $L$ yields an unbiased estimator of $\nabla_{(\theta, \phi)} \mathcal{F}(\sem{P}_1(\theta), \sem{Q}_1(\phi))$.

\section{Full System}
\label{sec:sp}

In the previous sections, we presented a formal model of the key features of our approach. 
Our full language and system (detailed in Appx.~\ref{appdx:full-system}) extends our formal model in three key ways:
\begin{itemize}[leftmargin=*]
\item \textbf{New language features for probabilistic models and variational families (Appx.~\ref{sec:margnorm}).} Our full language includes constructs for \textit{marginalizing} ($\textbf{marginal}$) and \textit{normalizing} ($\textbf{normalize}$) $\lambda_\text{Gen}$ programs, making it possible to express a broader class of models and variational families than in current systems. Our versions of these constructs are designed following~\citet{lew_probabilistic_2023}.

\item \textbf{Differentiable stochastic estimators of densities and density reciprocals (Appx.~\ref{sec:diffsp}).} When exact densities of $\lambda_\text{Gen}$ programs cannot be efficiently computed, our full system can compile $\lambda_\text{ADEV}$ terms implementing differentiable unbiased estimators of the required density functions and their reciprocals. These estimators can even have learnable parameters controlling their variance, which can be optimized jointly as part of the overall variational objective. %TODO: point out that variational optimization *also improves the `marginal` proposals.

\item \textbf{Reverse-mode automatic differentiation of expected values (Appx.~\ref{appdx:yolo}).} Our full language's AD algorithm computes \textit{vector-Jacobian products} for expected values of probabilistic objectives, whereas our formal development shows only \textit{Jacobian-vector products}. Algorithms for vector-Jacobian products, also known as \textit{reverse-mode} AD algorithms, are much more efficient when optimizing scalar losses with large numbers of parameters, common in deep learning.
\end{itemize}

%\noindent The remainder of this section introduces these features, explains their implications for programmable variational inference, and highlights the key innovations necessary to implement them.
% NOTE: moved everything below this into the appendix.

\section{Evaluation}
\label{sec:evaluation}
We evaluate our approach using \texttt{genjax.vi}, a prototype of our proposed architecture implemented as an extension to a JAX-hosted version of Gen \cite{cusumano-towner_gen_2019}. {All experiments were run on a single device with an AMD Ryzen 7 7800X3D @ 5.050 GHz CPU and an Nvidia RTX 4090 GPU.} We consider several case studies designed to answer the following questions:
\begin{itemize}[leftmargin=*]
    \item \textbf{Overhead.} \textit{How much overhead is incurred by using our automated gradient estimators, over hand-coded versions?} We compare the same gradient estimator for a variational autoencoder \cite{kingma_auto-encoding_2022} constructed (a) via a hand-coded implementation and (b) via our automation.
    \item \textbf{Overall performance.} \textit{How well can we solve a challenging inference problem using our system compared to other PPLs that support variational inference?} We consider the \textit{Attend-Infer-Repeat} (AIR) model \cite{eslami_attend_2016} and compare the capabilities of our system to Pyro \cite{bingham_pyro_2018}.
    \item \textbf{Expressivity and compositional correctness.} \textit{For the objectives and estimator strategies expressible in our system, is it possible to combine all objectives and estimator strategies while maintaining correctness?} We evaluate the expressiveness of our system vs. Pyro on the AIR model, and on a hierarchical variational inference problem~\cite{agakov_auxiliary_2004}. % (HVI) objective \cite{bornschein_reweighted_2015}. We investigate our system's ability to express HVI, importance weighted HVI (IWHVI), and doubly importance weighted HVI (DIWHVI), by combining our $\textbf{marginal}$ feature with SIR.
\end{itemize}

\paragraph{Overhead.} Table~\ref{tab:overhead} presents a runtime comparison between \texttt{genjax.vi} and a hand-coded implementation of the gradient estimator in JAX (Appx. \ref{appdx:overhead}). %: in JAX, using \texttt{jax.grad} and \texttt{tensorflow.probability}). 
We measure the wall time required to compute a gradient estimate for different batch sizes $n$. We find that our automation introduces a small amount of runtime overhead (around 3-10\%) compared to our hand coded implementation (Fig.~\ref{fig:overhead_transcript}). %In epoch computations, where we compute and apply gradients to the parameters of a model and guide in a loop ($\sim$940 iterations), the overhead leads to a runtime penalty on the order of $\sim$1ms (Fig.~\ref{fig:overhead_transcript}, top right, bottom). %Our program transformations also induce a compile time penalty - but we found this penalty to be negligible and not statistically significant for the VAE program. In other implementations of our language, we'd expect the runtime cost associated with continuation passing style (CPS) to incur a penalty: our JAX implementation negates this by performing partial evaluation to remove continuations at compile time. On this case study, the cost is cheap to utilize our abstractions, and the benefits are numerous, including improved readability from separation of concerns, and ease of changing estimators.
%\newpage
\setlength{\columnsep}{9pt}%
\setlength{\intextsep}{5pt}%
\begin{wraptable}{r}{0.38\textwidth}
  \setlength{\abovecaptionskip}{1pt}
  \caption{Timing (ms) our estimators versus hand coded estimators for the VAE.}
 % \vspace{-10pt}
  {\small\begin{tabular}{|c|c|c|}
  \hline
  Batch size   & Ours & Hand coded \\ \hline
  64           & 0.11 $\pm$ 0.02 &  0.09 $\pm$ 0.04 \\
  128          & 0.22 $\pm$ 0.2  &  0.16 $\pm$ 0.08 \\
  256          & 0.31 $\pm$ 0.18 &  0.29 $\pm$ 0.17 \\
  512          & 0.56 $\pm$ 0.35 &  0.54 $\pm$ 0.34 \\
  1024         & 1.58 $\pm$ 1.13 &  1.07 $\pm$ 0.70 \\
  \hline
  \end{tabular}}
  \label{tab:overhead}
\end{wraptable}
\vspace{-4mm}

\paragraph{Overall Performance.} We consider the \textit{Attend, Infer, Repeat} \cite{eslami_attend_2016} (AIR) model (Fig.~\ref{fig:air_arch}). We plot accuracy and loss curves over time in Fig.~\ref{fig:air_collat}, for several estimators expressed in our system and in Pyro. We also compare timing results, shown in Table ~\ref{fig:benchmark_table}. Our implementation's performance is competitive, and we support a broader class of estimators and objectives than Pyro.
%- e.g. allowing measure valued derivative estimators for Bernoulli distributions as a gradient estimator strategy, for the IWELBO objective. We can easily swap objectives (as programs in our loss language), e.g. using Reweighted Wake-Sleep (RWS) instead of the ELBO (or IWELBO instead of ELBO) without compromising the correctness of our gradient estimators. 
We find that some estimators we support (in particular those based on measure-valued derivatives) lead to faster convergence than the estimators automated by Pyro.

\begin{figure}[t]
\setlength{\belowcaptionskip}{-5mm}
\includegraphics[width=0.9\linewidth]{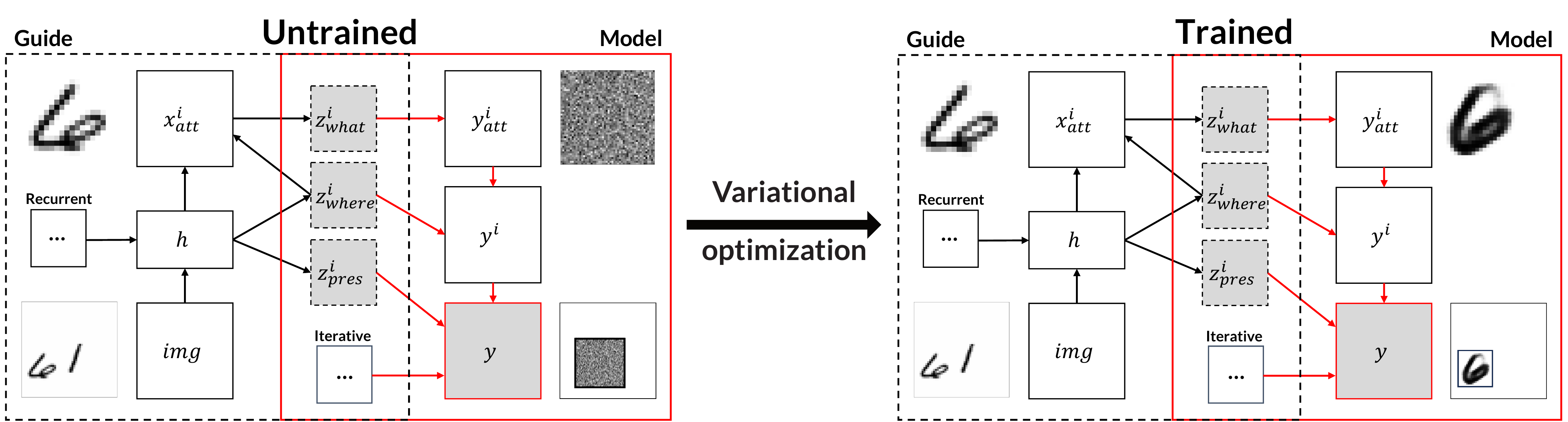}
\vspace{-2mm}
\caption{AIR is a generative model for multi-object images, trained with variational inference. The model randomly selects a number of patches to render onto a canvas, and a location, scale, and latent code for each. The variational family predicts these latent variables from an image. The model is trained on a dataset of images constructed by randomly translating and scaling  MNIST digits onto a canvas.}
\label{fig:air_arch}
\Description{AIR figure.}
\end{figure}

\begin{figure}[htb!]
\centering
\captionof{table}{Time (in seconds) to train the AIR model~\cite{eslami_attend_2016} for one epoch (batch size 64) with different objectives and estimators. All discrete variables use the same estimation strategy. IWELBO runs have $n = 2$ particles.}
\vspace{-3mm}
\footnotesize{\begin{tabular}{ | c c | c | c | c | c | c | }
    \hline
    System & Compiler & REINFORCE & ENUM & MVD & IWELBO + REINFORCE & IWELBO + MVD \\
    \hline
    \texttt{genjax.vi}    & JAX (XLA) & 1.52 $\pm$ 0.05 & 6.22 $\pm$ 0.29 & 1.74 $\pm$ 0.04 & 2.28 $\pm$ 0.12 & 3.74 $\pm$ 0.05 \\
    \texttt{pyro}            & Torch & 12.28 $\pm$ 0.55 & 122.93 $\pm$ 1.74 & \text{\sffamily X} & 22.17 $\pm$ 1.2 & \text{\sffamily X} \\
    \hline
\end{tabular}}
\label{fig:benchmark_table}
\end{figure}

\paragraph{Expressivity and Compositional Correctness.} In Table \ref{tab:checks_exes}, we enumerate several possible combinations of gradient estimation strategies and objectives for the AIR model. In Table \ref{fig:hvi_table}, we consider the model shown in Fig.~\ref{fig:example_transcript}, and implement the model and naive variational guide in \texttt{genjax.vi}, Pyro and NumPyro, as well as the auxiliary variable variational guide in \texttt{genjax.vi}. We show statistics on final mean objective values for different variational objectives. While ELBO and IWELBO are standard, our system allows using more expressive approximations and tighter lower bound objectives \textit{compositionally} (such as DIWHVI~\cite{sobolev_importance_2019}, which uses SIR to estimate densities of marginals, in the IWELBO objective) to achieve tighter variational bounds.

\setlength{\dashlinedash}{0.2pt}
\setlength{\dashlinegap}{2pt}
\setlength{\arrayrulewidth}{0.2pt}

% \begin{table}[h]

% \newfloatcommand{capbtabbox}{table}[][\FBwidth]
% \begin{figure}
% \begin{floatrow}
% \ffigbox{%
%   \rule{3cm}{3cm}%
% }{%
%   \caption{A figure}%
% }
% \capbtabbox{%
%   \centering
%   \input{fig/expressivity_table/checks_exes_tabular}
% }{%
% \caption{\label{tab:my_label}We consider a combinatorial space of gradient estimators and objective functions that can be optimized for the Attend-Infer-Repeat (AIR) model. Here, we focus on strategies for the discrete random choices and assume that multivariate normal choices use reparametrization. Pyro supports only a pre-selected menu of options within this combinatorial space, whereas \texttt{genjax.vi} can express them all, compositionally.}%
% }
% \end{floatrow}
% \end{figure}

\begin{figure}[htbp]
  \begin{minipage}[t]{0.48\linewidth}
  \centering
    \captionsetup{type=table}
    \caption{{Combinatorial space of gradient estimators and objective functions for the AIR model, which our programmable approach helps to explore. %Our programmable approach allows users to rapidly explore a broad class of gradient estimators.
    %Here, we focus on strategies for the discrete random choices and assume that multivariate normal choices use reparametrization. 
    }}
    \label{tab:checks_exes}
    \vspace{-3mm}
    \tiny{
RE = REINFORCE estimator, EN = enumeration estimator, BL = REINFORCE with learned baselines, MV = measure-valued derivative estimator
\begin{tabular}[t]{|cccc|c|c|cc|}
\hline
\multicolumn{4}{|c|}{\textbf{Grad. Estimation Strategies}} & \multirow{2}{*}{\textbf{Objective}} & \multirow{2}{*}{\textbf{Batch}} & \multicolumn{2}{|c|}{\textbf{System}} \\
{\tiny RE.} & {\tiny EN.} & {\tiny BL.} & {\tiny MV.} &  &  & Pyro & Ours \\\hline
\multirow{2}{*}{$\checkmark$} & \multirow{2}{*}{ } & \multirow{2}{*}{ } & \multirow{2}{*}{ } & ELBO & $\geq$ 1 & \textcolor{teal}{\textbf{$\checkmark$}} & \textcolor{teal}{\textbf{$\checkmark$}} \\
\cdashline{5-8}
& & & & IWAE & $\geq$ 1 & \textcolor{teal}{\textbf{$\checkmark$}} & \textcolor{teal}{\textbf{$\checkmark$}} \\
\cdashline{5-8}
\hdashline
\multirow{2}{*}{ } & \multirow{2}{*}{$\checkmark$} & \multirow{2}{*}{ } & \multirow{2}{*}{ } & ELBO & $\geq$ 1 & \textcolor{teal}{\textbf{$\checkmark$}} & \textcolor{teal}{\textbf{$\checkmark$}} \\
\cdashline{5-8}
& & & & IWAE & $\geq$ 1 & \textcolor{purple}{\textbf{$\times$}} & \textcolor{teal}{\textbf{$\checkmark$}} \\
\cdashline{5-8}
\hdashline
\multirow{2}{*}{ } & \multirow{2}{*}{ } & \multirow{2}{*}{$\checkmark$} & \multirow{2}{*}{ } & ELBO & $\geq$ 1 & \textcolor{teal}{\textbf{$\checkmark$}} & \textcolor{teal}{\textbf{$\checkmark$}} \\
\cdashline{5-8}
& & & & IWAE & $\geq$ 1 & \textcolor{purple}{\textbf{$\times$}} & \textcolor{teal}{\textbf{$\checkmark$}} \\
\cdashline{5-8}
\hdashline
\multirow{2}{*}{ } & \multirow{2}{*}{ } & \multirow{2}{*}{ } & \multirow{2}{*}{$\checkmark$} & ELBO & $\geq$ 1 & \textcolor{purple}{\textbf{$\times$}} & \textcolor{teal}{\textbf{$\checkmark$}} \\
\cdashline{5-8}
& & & & IWAE & $\geq$ 1 & \textcolor{purple}{\textbf{$\times$}} & \textcolor{teal}{\textbf{$\checkmark$}} \\
\cdashline{5-8}
\hdashline
\multirow{2}{*}{$\checkmark$} & \multirow{2}{*}{$\checkmark$} & \multirow{2}{*}{ } & \multirow{2}{*}{ } & ELBO & $\geq$ 1 & \textcolor{teal}{\textbf{$\checkmark$}} & \textcolor{teal}{\textbf{$\checkmark$}} \\
\cdashline{5-8}
& & & & IWAE & $\geq$ 1 & \textcolor{purple}{\textbf{$\times$}} & \textcolor{teal}{\textbf{$\checkmark$}} \\
\cdashline{5-8}
\hdashline
\multirow{2}{*}{$\checkmark$} & \multirow{2}{*}{ } & \multirow{2}{*}{$\checkmark$} & \multirow{2}{*}{ } & ELBO & $\geq$ 1 & \textcolor{teal}{\textbf{$\checkmark$}} & \textcolor{teal}{\textbf{$\checkmark$}} \\
\cdashline{5-8}
& & & & IWAE & $\geq$ 1 & \textcolor{purple}{\textbf{$\times$}} & \textcolor{teal}{\textbf{$\checkmark$}} \\
\cdashline{5-8}
\hdashline
\multirow{2}{*}{$\checkmark$} & \multirow{2}{*}{ } & \multirow{2}{*}{ } & \multirow{2}{*}{$\checkmark$} & ELBO & $\geq$ 1 & \textcolor{purple}{\textbf{$\times$}} & \textcolor{teal}{\textbf{$\checkmark$}} \\
\cdashline{5-8}
& & & & IWAE & $\geq$ 1 & \textcolor{purple}{\textbf{$\times$}} & \textcolor{teal}{\textbf{$\checkmark$}} \\
\cdashline{5-8}
\hdashline
\multirow{2}{*}{ } & \multirow{2}{*}{$\checkmark$} & \multirow{2}{*}{$\checkmark$} & \multirow{2}{*}{ } & ELBO & $\geq$ 1 & \textcolor{purple}{\textbf{$\times$}} & \textcolor{teal}{\textbf{$\checkmark$}} \\
\cdashline{5-8}
& & & & IWAE & $\geq$ 1 & \textcolor{purple}{\textbf{$\times$}} & \textcolor{teal}{\textbf{$\checkmark$}} \\
\cdashline{5-8}
\hdashline
\multirow{2}{*}{ } & \multirow{2}{*}{$\checkmark$} & \multirow{2}{*}{ } & \multirow{2}{*}{$\checkmark$} & ELBO & $\geq$ 1 & \textcolor{purple}{\textbf{$\times$}} & \textcolor{teal}{\textbf{$\checkmark$}} \\
\cdashline{5-8}
& & & & IWAE & $\geq$ 1 & \textcolor{purple}{\textbf{$\times$}} & \textcolor{teal}{\textbf{$\checkmark$}} \\
\cdashline{5-8}
\hdashline
\multirow{2}{*}{ } & \multirow{2}{*}{ } & \multirow{2}{*}{$\checkmark$} & \multirow{2}{*}{$\checkmark$} & ELBO & $\geq$ 1 & \textcolor{purple}{\textbf{$\times$}} & \textcolor{teal}{\textbf{$\checkmark$}} \\
\cdashline{5-8}
& & & & IWAE & $\geq$ 1 & \textcolor{purple}{\textbf{$\times$}} & \textcolor{teal}{\textbf{$\checkmark$}} \\
\cdashline{5-8}
\hdashline
\multirow{2}{*}{$\checkmark$} & \multirow{2}{*}{$\checkmark$} & \multirow{2}{*}{$\checkmark$} & \multirow{2}{*}{ } & ELBO & $\geq$ 1 & \textcolor{purple}{\textbf{$\times$}} & \textcolor{teal}{\textbf{$\checkmark$}} \\
\cdashline{5-8}
& & & & IWAE & $\geq$ 1 & \textcolor{purple}{\textbf{$\times$}} & \textcolor{teal}{\textbf{$\checkmark$}} \\
\cdashline{5-8}
\hdashline
\multirow{2}{*}{$\checkmark$} & \multirow{2}{*}{$\checkmark$} & \multirow{2}{*}{ } & \multirow{2}{*}{$\checkmark$} & ELBO & $\geq$ 1 & \textcolor{purple}{\textbf{$\times$}} & \textcolor{teal}{\textbf{$\checkmark$}} \\
\cdashline{5-8}
& & & & IWAE & $\geq$ 1 & \textcolor{purple}{\textbf{$\times$}} & \textcolor{teal}{\textbf{$\checkmark$}} \\
\cdashline{5-8}
\hdashline
\multirow{2}{*}{$\checkmark$} & \multirow{2}{*}{ } & \multirow{2}{*}{$\checkmark$} & \multirow{2}{*}{$\checkmark$} & ELBO & $\geq$ 1 & \textcolor{purple}{\textbf{$\times$}} & \textcolor{teal}{\textbf{$\checkmark$}} \\
\cdashline{5-8}
& & & & IWAE & $\geq$ 1 & \textcolor{purple}{\textbf{$\times$}} & \textcolor{teal}{\textbf{$\checkmark$}} \\
\cdashline{5-8}
\hdashline
\multirow{2}{*}{ } & \multirow{2}{*}{$\checkmark$} & \multirow{2}{*}{$\checkmark$} & \multirow{2}{*}{$\checkmark$} & ELBO & $\geq$ 1 & \textcolor{purple}{\textbf{$\times$}} & \textcolor{teal}{\textbf{$\checkmark$}} \\
\cdashline{5-8}
& & & & IWAE & $\geq$ 1 & \textcolor{purple}{\textbf{$\times$}} & \textcolor{teal}{\textbf{$\checkmark$}} \\
\cdashline{5-8}
\hdashline
\multicolumn{4}{|c|}{\multirow{2}{*}{Not Exploited}} & \multirow{2}{*}{RWS} & 1 & \textcolor{teal}{\textbf{$\checkmark$}} & \textcolor{teal}{\textbf{$\checkmark$}} \\
\cdashline{6-8}
\multicolumn{4}{|c|}{} &  & $>$ 1 & \textcolor{purple}{\textbf{$\times$}} & \textcolor{teal}{\textbf{$\checkmark$}} \\
\cdashline{6-8}
\hline
\end{tabular}}
  \end{minipage}%
  \hfill
  \begin{minipage}[t]{0.49\linewidth}
    \centering
    \includegraphics[width=\linewidth, valign=t]{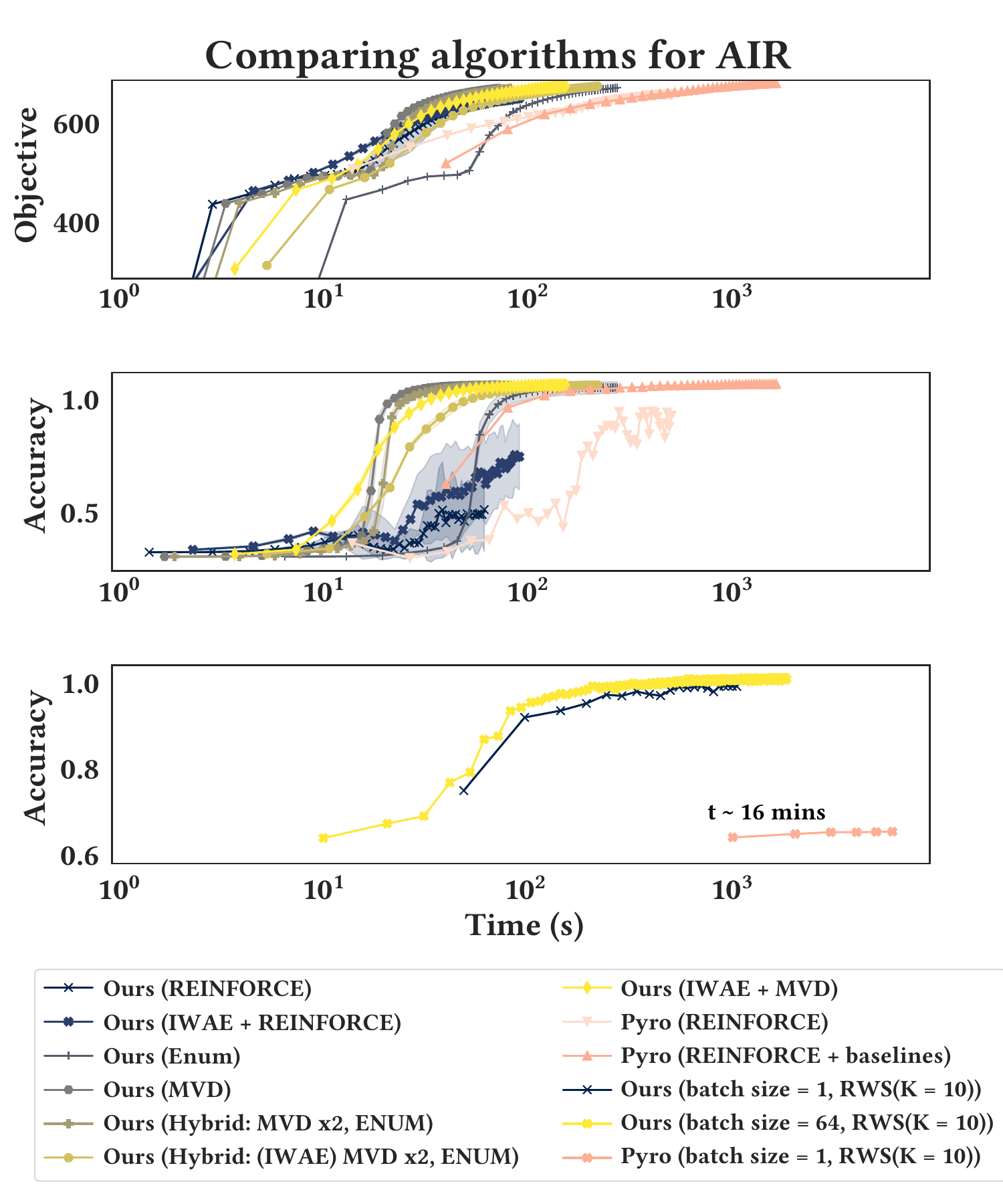}
    \vspace{-3mm}
    \caption{We evaluate a variety of custom estimators and objectives (ELBO, IWAE, RWS) using our system. Our on average best estimator (IWAE + MVD, not expressible in Pyro) converges an order of magnitude faster than Pyro's recommended estimator.}
    \label{fig:air_collat}
  \end{minipage}
\end{figure}

\begin{table}[htb!]
\captionsetup{skip=3pt}
\captionof{table}{Mean objective value (in nats) on repeated runs for several variational objectives, including ones which utilize $\textbf{marginal}$. $n$ and $m$ denote particle sizes for SIR algorithms.}
\footnotesize{\begin{tabular}[t]{ | c | c | c | c | c | c | }
    \hline
    System & ELBO & IWELBO ($n = 5$) & HVI & IWHVI ($m = 5$) & DIWHVI ($n = 5, m = 5$) \\
    \hline
    \texttt{genjax.vi}    & -8.08 & -7.79 & -9.75 & -8.18 & -7.33 \\
    \texttt{numpyro}      & -8.08 & -7.77 & \checkmark / \text{\sffamily X} & \text{\sffamily X} & \text{\sffamily X} \\
    \texttt{pyro}         & -8.08  & -7.75 & \checkmark / \text{\sffamily X} & \text{\sffamily X} & \text{\sffamily X} \\
    \hline
\end{tabular}}
\label{fig:hvi_table}
\end{table}

\section{Related work}
\label{sec:related}

\paragraph{Variational Inference in PPLs.} Many PPLs support some form of variational inference~\cite{bingham_pyro_2018,stites_learning_2021,carpenter2017stan,ge2018turing,cusumano-towner_gen_2019},
and it is the primary focus of Pyro~\cite{bingham_pyro_2018} and ProbTorch~\cite{stites_learning_2021}. Both
Pyro and ProbTorch have endeavored to make inference more programmable. For example, ProbTorch has introduced
\textit{inference combinators} that make it easy to express certain nested variational inference algorithms~\cite{zimmermann_nested_2021,stites_learning_2021}. Pyro has perhaps the most mature support for variational inference, with many gradient estimators
and objectives supported. Pyro also implements some variance reduction strategies not yet supported by our system, e.g.
exploiting conditional independence using their \textbf{plate} operator. 
However, extending Pyro with new variational objectives or gradient estimation strategies requires a deep understanding of
its internals. Furthermore, Pyro's modeling language does not have our constructs for marginalization and normalization (although Pyro can marginalize discrete, finite-support auxiliary variables from models). Concurrently with this work, \citet{wagner2023fast} presented a PPL with correct-by-construction variational inference, where their notion of ``correctness'' is stronger in some ways and weaker in others than that of this work. In particular, \citet{wagner2023fast} works with \textit{smoothed approximations} of the user's probabilistic programs, and so the gradient estimates it computes are biased for the original objective. However, the degree of error in this smoothing is gradually annealed over the course of optimization, leading to convergence to a stationary point of the original objective. 

\paragraph{Static Analyses for Differentiability Criteria.} Within the PL community, researchers have made some progress toward formalizing and developing static analyses for ensuring soundness properties of variational inference~\cite{lee_towards_2019,lee_smoothness_2023,wang2021sound}, as well as program transformations for automatically constructing variational families informed by the model's structure~\cite{li2023type}. By contrast, we formalize the process by which user model, inference, and objective code is transformed into a gradient estimator, tracking the interactions between density computation, simulation, and automatic differentiation. One interesting direction would be to extend~\citet{lee_smoothness_2023}'s analysis to automatically annotate $\lambda_{\text{Gen}}$ programs with gradient estimators. Note that our current type-based analysis does not verify the local domination condition of Thm.~\ref{thm:unbiased-adev}, which~\citet{wagner2023fast} does manage to check statically, by imposing restrictions on the PPL and considering only the ELBO objective. 

\paragraph{Automated Gradient Estimation.} Due to its centrality to many applications in computer science and beyond, there has been intense interest in automating unbiased gradient estimation for objective functions expressed as expectations, yielding several frameworks for unbiasedly differentiating first-order stochastic computation graphs~\cite{krieken2021storchastic,schulman2015gradient,weber2019credit}, imperative programs with discrete randomness~\cite{arya_automatic_2023}, and higher-order probabilistic programs~\cite{lew_adev_2023}. Some works have also investigated automated computation of \textit{biased} gradient estimates via smoothing~\citep{kreikemeyer2023smoothing}.  However, these frameworks cannot be directly applied to variational inference problems, which require differentiating not just user code, but also traced simulators and log density evaluators of the probabilistic programs that users write. We address these challenges, by compiling density functions of user-defined probabilistic programs into a target language compatible with ADEV~\cite{lew_adev_2023}. We also extend ADEV in several other ways: our version adds \textbf{score} to differentiate not just expected values but general integrals against (potentially unnormalized) measures; is implemented performantly on GPU; and has been extended with a reverse-mode.

% \noindent\textbf{Sound and modular inference} % also cite minh?

% %% VI + programmable inference?
\paragraph{Programmable Inference.} We build on a long line of work that aims to make inference more \textit{programmable} in PPLs~\cite{mansinghka2014venture,mansinghka2018probabilistic,narayanan2016probabilistic,cusumano-towner_gen_2019,lew_probabilistic_2023}. In much the same way that this prior work has aimed to expose compositional structure in Monte Carlo algorithms and help users explore a broad class of algorithm settings, our work aims to do the same for variational inference, exposing compositional structure in gradient estimators and variational objectives.

% The automation of gradient computations for variational inference has led to the emergence of specialized programming languages and software frameworks that streamline these processes in machine learning and statistical modeling. We build on prior approaches: our core calculus is closest to languages like Gen \cite{cusumano-towner_gen_2019}, Pyro \cite{bingham_pyro_2018}, and ProbTorch \cite{siddharth_learning_2017}. We go beyond these languages by formalizing the interaction between density computation and gradient estimation, and significantly extend the landscape of learnable programs (including ones with stochastic control flow, discrete random choices, and marginal density estimation). Our work is inspired by work like \cite{stites_learning_2021}, which considers extending variational objectives to stacks of properly weighted importance samplers. Our work considers differentiable languages for model and proposal specification, estimator design questions, and a more flexible interface (the stochastic probability interface is less restrictive than enforcing properly weighted samplers). We're intrigued to determine if we can recover \cite{stites_learning_2021}, as well as the nested objectives described in \cite{zimmermann_nested_2021}, in our setting. Practically, our language development is also inspired by compositional inference frameworks like \texttt{monad-bayes} \cite{scibior_functional_2018}.

%% Formalism for inference
\paragraph{Formal Reasoning about Inference and Program Transformations.} Our semantics builds on recent work in the {denotational semantics and validation of Bayesian inference}~\cite{scibior_denotational_2017,heunen_convenient_2017,scibior_functional_2018}, as well as semantic foundations for differentiable programming~\cite{huot_correctness_2020,sherman_s_2021}. Our soundness proofs are based on logical relations~\cite{ahmed_step-indexed_2006,huot_correctness_2020}. We also draw on a tradition of deriving probabilistic inference algorithms via program transformations~\cite{narayanan2016probabilistic}.

\section{Discussion}
\label{sec:discussion}

This work gives a formal account of the automation that PPLs provide for variational inference — a powerful and widely used suite of features that has not previously been completely understood in formal programming language terms. This work's account provides a careful separation of the interactions between tracing, density computation, gradient estimation strategies, and automatic differentiation. Simultaneously, this work shows how to implement these features modularly and extensibly, addressing a number of pain points in existing implementations, and expanding the class of variational inference algorithms that users can easily express.

\paragraph{Limitations.} That said, the modular approach we have presented has several key limitations:
\begin{itemize}[leftmargin=*]
\item \textit{Static checks may be unnecessarily restrictive.} For example, although the rectified linear unit (ReLU) is not differentiable at 0, in many (but not all) contexts it is safe to use it as though it were differentiable, without compromising unbiasedness. The type system of $\lambda_\text{ADEV}$ is not sophisticated enough to distinguish when ReLU is safe or not, and so we must give it the restrictive type $\mathbb{R}^*\rightarrow\mathbb{R}^*$. Of course, at their own risk, users of the system are free to ignore these static checks.
\item \textit{No parametric discontinuities.} A key limitation of our language, shared by Pyro, ProbTorch, and Gen, is that \textit{parametric discontinuities} (expressions that compute discontinuous functions of the input parameters) are not permitted. Variational inference is possible in these settings, and \citet{lee2018reparameterization} proposed a gradient estimator that can be automated for a restricted PPL with affine discontinuities. More recently, \citet{bangaru2021systematically} and \citet{michel2024distributions} have presented techniques for differentiating integral expressions with parametric discontinuities. It is not yet clear to what extent the design we present could be cleanly extended to exploit these techniques.
\item \textit{User-specified variational objectives may be ill-defined.} Our correctness theorem assumes as a precondition that the objective the user has specified is well-defined (i.e., if it is defined as an expected value, that it is finite for all input parameters). Previous work has identified the verification of this condition as an important challenge for safe variational inference~\citep{lee_towards_2019}, but our system includes no static checks to do so.
\item \textit{Continuation-passing style (CPS) is unnatural in many host languages.} $\textbf{adev}\{\cdot\}$ transforms a program to CPS in order to automate unbiased gradient estimation. This is easy in our Haskell implementation, but in Python and Julia, it introduces some amount of friction. For example, in Julia, certain host-language features (like mutation, and therefore most imperative loops) are incompatible with our CPS implementation. By contrast, Pyro and Gen place few restrictions on the host-language features that can be used to define models and variational families.
\end{itemize}

In addition, our implementation does not yet incorporate several important insights and capabilities from existing deep PPLs, including Pyro's use of tensor contractions to marginalize discrete variables~\citep{obermeyer2019functional,obermeyer2019tensor}, or the use of fine-grained control flow information to reduce the variance of gradient estimates~\citep{schulman2015gradient}.

\paragraph{Future Work.} We comment on several intriguing avenues for future work:
\begin{itemize}[leftmargin=*]
\item \textit{The search for low variance estimators.} Our approach automates the derivation of unbiased gradient estimators, but it says little about what estimators one should choose to achieve low variance on particular problems. Our approach should make it \textit{easier} to address this challenge, by allowing rapid exploration of a large space of estimation strategies, some of which have not been previously automated. We hope that our work and implementations might be used to carry out a study of the behavior of different gradient estimators, on a broader variety of problems than those enabled by current automation.
\item \textit{Interaction with ADEV.} Our system is the first to use ADEV \cite{lew_adev_2023} for scalable learning of large parameter spaces (\S\ref{appdx:yolo}). ADEV provides a fundamentally new perspective on automatic differentiation, and extends the technique to expected value loss functions. Our integration with ADEV has several implications: by virtue of the fact that the target language of our language's transformations is ADEV's source language, our system may be used with ADEV loss functions beyond the variational loss functions which we've discussed in this work. Extensions to ADEV which improve variance properties of gradient estimators symbiotically improve the performance of gradient estimators in our language. Indeed, we expect new investigations into low variance gradient estimators for discrete random choices \cite{arya_automatic_2023} to open up novel variational guide families, hitherto unexplored due to poor variance or computational intractability of discrete enumeration.
\end{itemize}

%% Closing statement
% It is our hope that this work will further weaken the boundaries between learning and probabilistic computation, and provide firm ground for old ideas (e.g. deep, structured generative modeling) and new (approximations which include normalization and marginalization, with estimated densities) to find growth.

%\newpage

\section*{Data-Availability Statement}

An artifact providing a version of \href{https://gen.dev/genjax/vi}{\texttt{genjax.vi}}, and reproducing our experiments, is available \cite{mccoy_r_becker_2024_10935596}.

\begin{acks}
The authors are grateful to Tuan Anh Le, Tan Zhi-Xuan, Cameron Freer, Andrew Bolton, George Matheos, Nishad Gothoskar, Martin Jankowiak, and Sam Witty for useful conversations and feedback, and to our anonymous referees for helpful feedback on earlier drafts of the paper. We are grateful for  support from DARPA, under the DARPA Machine Common Sense (Award ID: 030523-00001) and JUMP (CoCoSys, Prime Contract No. 2023-JU-3131) programs, and DSTA, under Master Agreement No. 801899998, as well as gifts from Google, and philanthropic gifts from an anonymous donor and the Siegel Family Foundation.
\end{acks}

\bibliography{morerefs}


\begin{thebibliography}{76}

%%% the Web address.

\ifx \showCODEN    \undefined \def \showCODEN     #1{\unskip}     \fi
\ifx \showDOI      \undefined \def \showDOI       #1{#1}\fi
\ifx \showISBNx    \undefined \def \showISBNx     #1{\unskip}     \fi
\ifx \showISBNxiii \undefined \def \showISBNxiii  #1{\unskip}     \fi
\ifx \showISSN     \undefined \def \showISSN      #1{\unskip}     \fi
\ifx \showLCCN     \undefined \def \showLCCN      #1{\unskip}     \fi
\ifx \shownote     \undefined \def \shownote      #1{#1}          \fi
\ifx \showarticletitle \undefined \def \showarticletitle #1{#1}   \fi
\ifx \showURL      \undefined \def \showURL       {\relax}        \fi
% invisible to TeX
\providecommand\bibfield[2]{#2}
\providecommand\bibinfo[2]{#2}
\providecommand\natexlab[1]{#1}
\providecommand\showeprint[2][]{arXiv:#2}

\bibitem[Agakov and Barber(2004)]%
        {agakov_auxiliary_2004}
\bibfield{author}{\bibinfo{person}{Felix~V. Agakov} {and} \bibinfo{person}{David Barber}.} \bibinfo{year}{2004}\natexlab{}.
\newblock \showarticletitle{An {Auxiliary} {Variational} {Method}}. In \bibinfo{booktitle}{\emph{Neural {Information} {Processing}}} \emph{(\bibinfo{series}{Lecture {Notes} in {Computer} {Science}})}. \bibinfo{publisher}{Springer}, \bibinfo{address}{Berlin, Heidelberg}, \bibinfo{pages}{561--566}.
\newblock
\showISBNx{978-3-540-30499-9}
\urldef\tempurl%
\url{https://doi.org/10.1007/978-3-540-30499-9_86}
\showDOI{\tempurl}


\bibitem[Ahmed(2006)]%
        {ahmed_step-indexed_2006}
\bibfield{author}{\bibinfo{person}{Amal Ahmed}.} \bibinfo{year}{2006}\natexlab{}.
\newblock \showarticletitle{Step-{Indexed} {Syntactic} {Logical} {Relations} for {Recursive} and {Quantified} {Types}}. In \bibinfo{booktitle}{\emph{Programming {Languages} and {Systems}}} \emph{(\bibinfo{series}{Lecture {Notes} in {Computer} {Science}})}. \bibinfo{publisher}{Springer}, \bibinfo{address}{Berlin, Heidelberg}, \bibinfo{pages}{69--83}.
\newblock
\showISBNx{978-3-540-33096-7}
\urldef\tempurl%
\url{https://doi.org/10.1007/11693024_6}
\showDOI{\tempurl}


\bibitem[Arya et~al\mbox{.}(2022)]%
        {arya_automatic_2023}
\bibfield{author}{\bibinfo{person}{Gaurav Arya}, \bibinfo{person}{Moritz Schauer}, \bibinfo{person}{Frank Sch{\"{a}}fer}, {and} \bibinfo{person}{Christopher Rackauckas}.} \bibinfo{year}{2022}\natexlab{}.
\newblock \showarticletitle{Automatic Differentiation of Programs with Discrete Randomness}. In \bibinfo{booktitle}{\emph{Advances in Neural Information Processing Systems 35: Annual Conference on Neural Information Processing Systems 2022, NeurIPS 2022, New Orleans, LA, USA, November 28 - December 9, 2022}}, \bibfield{editor}{\bibinfo{person}{Sanmi Koyejo}, \bibinfo{person}{S.~Mohamed}, \bibinfo{person}{A.~Agarwal}, \bibinfo{person}{Danielle Belgrave}, \bibinfo{person}{K.~Cho}, {and} \bibinfo{person}{A.~Oh}} (Eds.).
\newblock
\urldef\tempurl%
\url{http://papers.nips.cc/paper\_files/paper/2022/hash/43d8e5fc816c692f342493331d5e98fc-Abstract-Conference.html}
\showURL{%
\tempurl}


\bibitem[Bangaru et~al\mbox{.}(2021)]%
        {bangaru2021systematically}
\bibfield{author}{\bibinfo{person}{Sai~Praveen Bangaru}, \bibinfo{person}{Jesse Michel}, \bibinfo{person}{Kevin Mu}, \bibinfo{person}{Gilbert Bernstein}, \bibinfo{person}{Tzu{-}Mao Li}, {and} \bibinfo{person}{Jonathan Ragan{-}Kelley}.} \bibinfo{year}{2021}\natexlab{}.
\newblock \showarticletitle{Systematically differentiating parametric discontinuities}.
\newblock \bibinfo{journal}{\emph{{ACM} Trans. Graph.}} \bibinfo{volume}{40}, \bibinfo{number}{4} (\bibinfo{year}{2021}), \bibinfo{pages}{107:1--107:18}.
\newblock
\urldef\tempurl%
\url{https://doi.org/10.1145/3450626.3459775}
\showDOI{\tempurl}


\bibitem[Becker et~al\mbox{.}(2024)]%
        {mccoy_r_becker_2024_10935596}
\bibfield{author}{\bibinfo{person}{McCoy~R. Becker}, \bibinfo{person}{Alexander~K. Lew}, {and} \bibinfo{person}{Xiaoyan Wang}.} \bibinfo{year}{2024}\natexlab{}.
\newblock \bibinfo{title}{probcomp/programmable-vi-pldi-2024: v0.1.2}.
\newblock \bibinfo{howpublished}{Zenodo}.
\newblock
\urldef\tempurl%
\url{https://doi.org/10.5281/zenodo.10935596}
\showDOI{\tempurl}


\bibitem[Bingham et~al\mbox{.}(2019)]%
        {bingham_pyro_2018}
\bibfield{author}{\bibinfo{person}{Eli Bingham}, \bibinfo{person}{Jonathan~P. Chen}, \bibinfo{person}{Martin Jankowiak}, \bibinfo{person}{Fritz Obermeyer}, \bibinfo{person}{Neeraj Pradhan}, \bibinfo{person}{Theofanis Karaletsos}, \bibinfo{person}{Rohit Singh}, \bibinfo{person}{Paul~A. Szerlip}, \bibinfo{person}{Paul Horsfall}, {and} \bibinfo{person}{Noah~D. Goodman}.} \bibinfo{year}{2019}\natexlab{}.
\newblock \bibinfo{title}{Pyro: Deep Universal Probabilistic Programming}.
\newblock , \bibinfo{numpages}{28:1--28:6}~pages.
\newblock
\urldef\tempurl%
\url{http://jmlr.org/papers/v20/18-403.html}
\showURL{%
\tempurl}


\bibitem[Blei and Jordan(2006)]%
        {blei2006variational}
\bibfield{author}{\bibinfo{person}{David~M Blei} {and} \bibinfo{person}{Michael~I Jordan}.} \bibinfo{year}{2006}\natexlab{}.
\newblock \showarticletitle{Variational inference for Dirichlet process mixtures}.
\newblock  (\bibinfo{year}{2006}).
\newblock


\bibitem[Blei et~al\mbox{.}(2016)]%
        {blei2017variational}
\bibfield{author}{\bibinfo{person}{David~M. Blei}, \bibinfo{person}{Alp Kucukelbir}, {and} \bibinfo{person}{Jon~D. McAuliffe}.} \bibinfo{year}{2016}\natexlab{}.
\newblock \showarticletitle{Variational Inference: {A} Review for Statisticians}.
\newblock \bibinfo{journal}{\emph{CoRR}}  \bibinfo{volume}{abs/1601.00670} (\bibinfo{year}{2016}).
\newblock
\showeprint[arXiv]{1601.00670}
\urldef\tempurl%
\url{http://arxiv.org/abs/1601.00670}
\showURL{%
\tempurl}


\bibitem[Borgstr{\"o}m et~al\mbox{.}(2016)]%
        {borgstrom2016lambda}
\bibfield{author}{\bibinfo{person}{Johannes Borgstr{\"o}m}, \bibinfo{person}{Ugo Dal~Lago}, \bibinfo{person}{Andrew~D Gordon}, {and} \bibinfo{person}{Marcin Szymczak}.} \bibinfo{year}{2016}\natexlab{}.
\newblock \showarticletitle{A lambda-calculus foundation for universal probabilistic programming}.
\newblock \bibinfo{journal}{\emph{ACM SIGPLAN Notices}} \bibinfo{volume}{51}, \bibinfo{number}{9} (\bibinfo{year}{2016}), \bibinfo{pages}{33--46}.
\newblock


\bibitem[Bornschein and Bengio(2015)]%
        {bornschein_reweighted_2015}
\bibfield{author}{\bibinfo{person}{Jörg Bornschein} {and} \bibinfo{person}{Yoshua Bengio}.} \bibinfo{year}{2015}\natexlab{}.
\newblock \bibinfo{title}{Reweighted {Wake}-{Sleep}}.
\newblock
\newblock
\urldef\tempurl%
\url{https://doi.org/10.48550/arXiv.1406.2751}
\showDOI{\tempurl}
\newblock
\shownote{arXiv:1406.2751 [cs]}.


\bibitem[Burda et~al\mbox{.}(2016)]%
        {burda_importance_2016}
\bibfield{author}{\bibinfo{person}{Yuri Burda}, \bibinfo{person}{Roger Grosse}, {and} \bibinfo{person}{Ruslan Salakhutdinov}.} \bibinfo{year}{2016}\natexlab{}.
\newblock \bibinfo{title}{Importance {Weighted} {Autoencoders}}.
\newblock
\newblock
\urldef\tempurl%
\url{https://doi.org/10.48550/arXiv.1509.00519}
\showDOI{\tempurl}
\newblock
\shownote{arXiv:1509.00519 [cs, stat]}.


\bibitem[Carpenter et~al\mbox{.}(2017)]%
        {carpenter2017stan}
\bibfield{author}{\bibinfo{person}{Bob Carpenter}, \bibinfo{person}{Andrew Gelman}, \bibinfo{person}{Matthew~D Hoffman}, \bibinfo{person}{Daniel Lee}, \bibinfo{person}{Ben Goodrich}, \bibinfo{person}{Michael Betancourt}, \bibinfo{person}{Marcus~A Brubaker}, \bibinfo{person}{Jiqiang Guo}, \bibinfo{person}{Peter Li}, {and} \bibinfo{person}{Allen Riddell}.} \bibinfo{year}{2017}\natexlab{}.
\newblock \showarticletitle{Stan: A probabilistic programming language}.
\newblock \bibinfo{journal}{\emph{Journal of Statistical Software}}  \bibinfo{volume}{76} (\bibinfo{year}{2017}).
\newblock


\bibitem[Cusumano-Towner and Mansinghka(2017)]%
        {cusumano-towner_aide_2017}
\bibfield{author}{\bibinfo{person}{Marco Cusumano-Towner} {and} \bibinfo{person}{Vikash~K Mansinghka}.} \bibinfo{year}{2017}\natexlab{}.
\newblock \showarticletitle{{AIDE}: {An} algorithm for measuring the accuracy of probabilistic inference algorithms}. In \bibinfo{booktitle}{\emph{Advances in {Neural} {Information} {Processing} {Systems}}}, Vol.~\bibinfo{volume}{30}. \bibinfo{publisher}{Curran Associates, Inc.}
\newblock
\urldef\tempurl%
\url{https://proceedings.neurips.cc/paper/2017/hash/acab0116c354964a558e65bdd07ff047-Abstract.html}
\showURL{%
\tempurl}


\bibitem[Cusumano-Towner et~al\mbox{.}(2019)]%
        {cusumano-towner_gen_2019}
\bibfield{author}{\bibinfo{person}{Marco~F. Cusumano-Towner}, \bibinfo{person}{Feras~A. Saad}, \bibinfo{person}{Alexander~K. Lew}, {and} \bibinfo{person}{Vikash~K. Mansinghka}.} \bibinfo{year}{2019}\natexlab{}.
\newblock \showarticletitle{Gen: a general-purpose probabilistic programming system with programmable inference}. In \bibinfo{booktitle}{\emph{Proceedings of the 40th {ACM} {SIGPLAN} {Conference} on {Programming} {Language} {Design} and {Implementation}}} \emph{(\bibinfo{series}{{PLDI} 2019})}. \bibinfo{publisher}{Association for Computing Machinery}, \bibinfo{address}{New York, NY, USA}, \bibinfo{pages}{221--236}.
\newblock
\showISBNx{978-1-4503-6712-7}
\urldef\tempurl%
\url{https://doi.org/10.1145/3314221.3314642}
\showDOI{\tempurl}


\bibitem[Dempster et~al\mbox{.}(1977)]%
        {dempster_maximum_1977}
\bibfield{author}{\bibinfo{person}{A.~P. Dempster}, \bibinfo{person}{N.~M. Laird}, {and} \bibinfo{person}{D.~B. Rubin}.} \bibinfo{year}{1977}\natexlab{}.
\newblock \showarticletitle{Maximum {Likelihood} from {Incomplete} {Data} via the {EM} {Algorithm}}.
\newblock \bibinfo{journal}{\emph{Journal of the Royal Statistical Society. Series B (Methodological)}} \bibinfo{volume}{39}, \bibinfo{number}{1} (\bibinfo{year}{1977}), \bibinfo{pages}{1--38}.
\newblock
\showISSN{0035-9246}
\urldef\tempurl%
\url{https://www.jstor.org/stable/2984875}
\showURL{%
\tempurl}
\newblock
\shownote{Publisher: [Royal Statistical Society, Wiley]}.


\bibitem[Domke(2021)]%
        {domke_easy_2021}
\bibfield{author}{\bibinfo{person}{Justin Domke}.} \bibinfo{year}{2021}\natexlab{}.
\newblock \bibinfo{title}{An {Easy} to {Interpret} {Diagnostic} for {Approximate} {Inference}: {Symmetric} {Divergence} {Over} {Simulations}}.
\newblock
\newblock
\urldef\tempurl%
\url{https://doi.org/10.48550/arXiv.2103.01030}
\showDOI{\tempurl}
\newblock
\shownote{arXiv:2103.01030 [cs, stat]}.


\bibitem[Eslami et~al\mbox{.}(2016)]%
        {eslami_attend_2016}
\bibfield{author}{\bibinfo{person}{S.~M.~Ali Eslami}, \bibinfo{person}{Nicolas Heess}, \bibinfo{person}{Theophane Weber}, \bibinfo{person}{Yuval Tassa}, \bibinfo{person}{David Szepesvari}, \bibinfo{person}{Koray Kavukcuoglu}, {and} \bibinfo{person}{Geoffrey~E. Hinton}.} \bibinfo{year}{2016}\natexlab{}.
\newblock \bibinfo{title}{Attend, {Infer}, {Repeat}: {Fast} {Scene} {Understanding} with {Generative} {Models}}.
\newblock
\newblock
\urldef\tempurl%
\url{https://doi.org/10.48550/arXiv.1603.08575}
\showDOI{\tempurl}
\newblock
\shownote{arXiv:1603.08575 [cs]}.


\bibitem[Foerster et~al\mbox{.}(2018)]%
        {foerster2018dice}
\bibfield{author}{\bibinfo{person}{Jakob Foerster}, \bibinfo{person}{Gregory Farquhar}, \bibinfo{person}{Maruan Al-Shedivat}, \bibinfo{person}{Tim Rockt{\"a}schel}, \bibinfo{person}{Eric Xing}, {and} \bibinfo{person}{Shimon Whiteson}.} \bibinfo{year}{2018}\natexlab{}.
\newblock \showarticletitle{Dice: The infinitely differentiable {M}onte {C}arlo estimator}. In \bibinfo{booktitle}{\emph{International Conference on Machine Learning}}. PMLR, \bibinfo{pages}{1529--1538}.
\newblock


\bibitem[Fox and Roberts(2012)]%
        {fox2012tutorial}
\bibfield{author}{\bibinfo{person}{Charles~W Fox} {and} \bibinfo{person}{Stephen~J Roberts}.} \bibinfo{year}{2012}\natexlab{}.
\newblock \showarticletitle{A tutorial on variational {B}ayesian inference}.
\newblock \bibinfo{journal}{\emph{Artificial intelligence review}}  \bibinfo{volume}{38} (\bibinfo{year}{2012}), \bibinfo{pages}{85--95}.
\newblock


\bibitem[Frostig et~al\mbox{.}(2018)]%
        {frostig2018compiling}
\bibfield{author}{\bibinfo{person}{Roy Frostig}, \bibinfo{person}{Matthew~James Johnson}, {and} \bibinfo{person}{Chris Leary}.} \bibinfo{year}{2018}\natexlab{}.
\newblock \showarticletitle{Compiling machine learning programs via high-level tracing}.
\newblock \bibinfo{journal}{\emph{Systems for Machine Learning}} \bibinfo{volume}{4}, \bibinfo{number}{9} (\bibinfo{year}{2018}).
\newblock


\bibitem[Ge et~al\mbox{.}(2018)]%
        {ge2018turing}
\bibfield{author}{\bibinfo{person}{Hong Ge}, \bibinfo{person}{Kai Xu}, {and} \bibinfo{person}{Zoubin Ghahramani}.} \bibinfo{year}{2018}\natexlab{}.
\newblock \showarticletitle{Turing: a language for flexible probabilistic inference}. In \bibinfo{booktitle}{\emph{International conference on artificial intelligence and statistics}}. PMLR, \bibinfo{pages}{1682--1690}.
\newblock


\bibitem[Gu et~al\mbox{.}(2015)]%
        {gu_neural_2015}
\bibfield{author}{\bibinfo{person}{Shixiang~(Shane) Gu}, \bibinfo{person}{Zoubin Ghahramani}, {and} \bibinfo{person}{Richard~E Turner}.} \bibinfo{year}{2015}\natexlab{}.
\newblock \showarticletitle{Neural {Adaptive} {Sequential} {Monte} {Carlo}}. In \bibinfo{booktitle}{\emph{Advances in {Neural} {Information} {Processing} {Systems}}}, Vol.~\bibinfo{volume}{28}. \bibinfo{publisher}{Curran Associates, Inc.}
\newblock
\urldef\tempurl%
\url{https://papers.nips.cc/paper_files/paper/2015/hash/99adff456950dd9629a5260c4de21858-Abstract.html}
\showURL{%
\tempurl}


\bibitem[Heunen et~al\mbox{.}(2017)]%
        {heunen_convenient_2017}
\bibfield{author}{\bibinfo{person}{Chris Heunen}, \bibinfo{person}{Ohad Kammar}, \bibinfo{person}{Sam Staton}, {and} \bibinfo{person}{Hongseok Yang}.} \bibinfo{year}{2017}\natexlab{}.
\newblock \showarticletitle{A convenient category for higher-order probability theory}. In \bibinfo{booktitle}{\emph{Proceedings of the 32nd {Annual} {ACM}/{IEEE} {Symposium} on {Logic} in {Computer} {Science}}} \emph{(\bibinfo{series}{{LICS} '17})}. \bibinfo{publisher}{IEEE Press}, \bibinfo{address}{Reykjavík, Iceland}, \bibinfo{pages}{1--12}.
\newblock
\showISBNx{978-1-5090-3018-7}


\bibitem[Hinton et~al\mbox{.}(1995)]%
        {hinton_wake-sleep_1995}
\bibfield{author}{\bibinfo{person}{Geoffrey~E. Hinton}, \bibinfo{person}{Peter Dayan}, \bibinfo{person}{Brendan~J. Frey}, {and} \bibinfo{person}{Radford~M. Neal}.} \bibinfo{year}{1995}\natexlab{}.
\newblock \showarticletitle{The "{Wake}-{Sleep}" {Algorithm} for {Unsupervised} {Neural} {Networks}}.
\newblock \bibinfo{journal}{\emph{Science}} \bibinfo{volume}{268}, \bibinfo{number}{5214} (\bibinfo{date}{May} \bibinfo{year}{1995}), \bibinfo{pages}{1158--1161}.
\newblock
\showISSN{0036-8075, 1095-9203}
\urldef\tempurl%
\url{https://doi.org/10.1126/science.7761831}
\showDOI{\tempurl}


\bibitem[Hoffman et~al\mbox{.}(2013)]%
        {hoffman2013stochastic}
\bibfield{author}{\bibinfo{person}{Matthew~D Hoffman}, \bibinfo{person}{David~M Blei}, \bibinfo{person}{Chong Wang}, {and} \bibinfo{person}{John Paisley}.} \bibinfo{year}{2013}\natexlab{}.
\newblock \showarticletitle{Stochastic variational inference}.
\newblock \bibinfo{journal}{\emph{Journal of Machine Learning Research}} (\bibinfo{year}{2013}).
\newblock


\bibitem[Huot et~al\mbox{.}(2020)]%
        {huot_correctness_2020}
\bibfield{author}{\bibinfo{person}{Mathieu Huot}, \bibinfo{person}{Sam Staton}, {and} \bibinfo{person}{Matthijs Vákár}.} \bibinfo{year}{2020}\natexlab{}.
\newblock \showarticletitle{Correctness of {Automatic} {Differentiation} via {Diffeologies} and {Categorical} {Gluing}}. In \bibinfo{booktitle}{\emph{Foundations of {Software} {Science} and {Computation} {Structures}}} \emph{(\bibinfo{series}{Lecture {Notes} in {Computer} {Science}})}. \bibinfo{publisher}{Springer International Publishing}, \bibinfo{address}{Cham}, \bibinfo{pages}{319--338}.
\newblock
\showISBNx{978-3-030-45231-5}
\urldef\tempurl%
\url{https://doi.org/10.1007/978-3-030-45231-5_17}
\showDOI{\tempurl}


\bibitem[Kingma et~al\mbox{.}(2021)]%
        {kingma2021variational}
\bibfield{author}{\bibinfo{person}{Diederik Kingma}, \bibinfo{person}{Tim Salimans}, \bibinfo{person}{Ben Poole}, {and} \bibinfo{person}{Jonathan Ho}.} \bibinfo{year}{2021}\natexlab{}.
\newblock \showarticletitle{Variational diffusion models}.
\newblock \bibinfo{journal}{\emph{Advances in Neural Information Processing Systems}}  \bibinfo{volume}{34} (\bibinfo{year}{2021}), \bibinfo{pages}{21696--21707}.
\newblock


\bibitem[Kingma et~al\mbox{.}(2014)]%
        {kingma_semi-supervised_2014}
\bibfield{author}{\bibinfo{person}{Diederik~P. Kingma}, \bibinfo{person}{Danilo~J. Rezende}, \bibinfo{person}{Shakir Mohamed}, {and} \bibinfo{person}{Max Welling}.} \bibinfo{year}{2014}\natexlab{}.
\newblock \showarticletitle{Semi-supervised learning with deep generative models}. In \bibinfo{booktitle}{\emph{Proceedings of the 27th {International} {Conference} on {Neural} {Information} {Processing} {Systems} - {Volume} 2}} \emph{(\bibinfo{series}{{NIPS}'14})}. \bibinfo{publisher}{MIT Press}, \bibinfo{address}{Cambridge, MA, USA}, \bibinfo{pages}{3581--3589}.
\newblock


\bibitem[Kingma and Welling(2022)]%
        {kingma_auto-encoding_2022}
\bibfield{author}{\bibinfo{person}{Diederik~P. Kingma} {and} \bibinfo{person}{Max Welling}.} \bibinfo{year}{2022}\natexlab{}.
\newblock \bibinfo{title}{Auto-{Encoding} {Variational} {{B}ayes}}.
\newblock
\newblock
\urldef\tempurl%
\url{https://doi.org/10.48550/arXiv.1312.6114}
\showDOI{\tempurl}
\newblock
\shownote{arXiv:1312.6114 [cs, stat]}.


\bibitem[Kreikemeyer and Andelfinger(2023)]%
        {kreikemeyer2023smoothing}
\bibfield{author}{\bibinfo{person}{Justin~N Kreikemeyer} {and} \bibinfo{person}{Philipp Andelfinger}.} \bibinfo{year}{2023}\natexlab{}.
\newblock \showarticletitle{Smoothing Methods for Automatic Differentiation Across Conditional Branches}.
\newblock \bibinfo{journal}{\emph{IEEE Access}} (\bibinfo{year}{2023}).
\newblock


\bibitem[Krieken et~al\mbox{.}(2021)]%
        {krieken2021storchastic}
\bibfield{author}{\bibinfo{person}{Emile Krieken}, \bibinfo{person}{Jakub Tomczak}, {and} \bibinfo{person}{Annette Ten~Teije}.} \bibinfo{year}{2021}\natexlab{}.
\newblock \showarticletitle{Storchastic: A framework for general stochastic automatic differentiation}.
\newblock \bibinfo{journal}{\emph{Advances in Neural Information Processing Systems}}  \bibinfo{volume}{34} (\bibinfo{year}{2021}), \bibinfo{pages}{7574--7587}.
\newblock


\bibitem[Kucukelbir et~al\mbox{.}(2017)]%
        {kucukelbir2017automatic}
\bibfield{author}{\bibinfo{person}{Alp Kucukelbir}, \bibinfo{person}{Dustin Tran}, \bibinfo{person}{Rajesh Ranganath}, \bibinfo{person}{Andrew Gelman}, {and} \bibinfo{person}{David~M Blei}.} \bibinfo{year}{2017}\natexlab{}.
\newblock \showarticletitle{Automatic differentiation variational inference}.
\newblock \bibinfo{journal}{\emph{Journal of machine learning research}} (\bibinfo{year}{2017}).
\newblock


\bibitem[Le et~al\mbox{.}(2019)]%
        {le_revisiting_2019}
\bibfield{author}{\bibinfo{person}{Tuan~Anh Le}, \bibinfo{person}{Adam~R. Kosiorek}, \bibinfo{person}{N. Siddharth}, \bibinfo{person}{Yee~Whye Teh}, {and} \bibinfo{person}{Frank Wood}.} \bibinfo{year}{2019}\natexlab{}.
\newblock \bibinfo{title}{Revisiting Reweighted Wake-Sleep for Models with Stochastic Control Flow}.
\newblock , \bibinfo{numpages}{1039--1049}~pages.
\newblock
\urldef\tempurl%
\url{http://proceedings.mlr.press/v115/le20a.html}
\showURL{%
\tempurl}


\bibitem[Lee et~al\mbox{.}(2023)]%
        {lee_smoothness_2023}
\bibfield{author}{\bibinfo{person}{Wonyeol Lee}, \bibinfo{person}{Xavier Rival}, {and} \bibinfo{person}{Hongseok Yang}.} \bibinfo{year}{2023}\natexlab{}.
\newblock \showarticletitle{Smoothness {Analysis} for {Probabilistic} {Programs} with {Application} to {Optimised} {Variational} {Inference}}.
\newblock \bibinfo{journal}{\emph{Proceedings of the ACM on Programming Languages}} \bibinfo{volume}{7}, \bibinfo{number}{POPL} (\bibinfo{date}{Jan.} \bibinfo{year}{2023}), \bibinfo{pages}{12:335--12:366}.
\newblock
\urldef\tempurl%
\url{https://doi.org/10.1145/3571205}
\showDOI{\tempurl}


\bibitem[Lee et~al\mbox{.}(2019)]%
        {lee_towards_2019}
\bibfield{author}{\bibinfo{person}{Wonyeol Lee}, \bibinfo{person}{Hangyeol Yu}, \bibinfo{person}{Xavier Rival}, {and} \bibinfo{person}{Hongseok Yang}.} \bibinfo{year}{2019}\natexlab{}.
\newblock \showarticletitle{Towards verified stochastic variational inference for probabilistic programs}.
\newblock \bibinfo{journal}{\emph{Proceedings of the ACM on Programming Languages}} \bibinfo{volume}{4}, \bibinfo{number}{POPL} (\bibinfo{date}{Dec.} \bibinfo{year}{2019}), \bibinfo{pages}{16:1--16:33}.
\newblock
\urldef\tempurl%
\url{https://doi.org/10.1145/3371084}
\showDOI{\tempurl}


\bibitem[Lee et~al\mbox{.}(2018)]%
        {lee2018reparameterization}
\bibfield{author}{\bibinfo{person}{Wonyeol Lee}, \bibinfo{person}{Hangyeol Yu}, {and} \bibinfo{person}{Hongseok Yang}.} \bibinfo{year}{2018}\natexlab{}.
\newblock \showarticletitle{Reparameterization gradient for non-differentiable models}.
\newblock \bibinfo{journal}{\emph{Advances in Neural Information Processing Systems}}  \bibinfo{volume}{31} (\bibinfo{year}{2018}).
\newblock


\bibitem[Lew et~al\mbox{.}(2022)]%
        {lew_recursive_2022}
\bibfield{author}{\bibinfo{person}{Alexander~K. Lew}, \bibinfo{person}{Marco~F. Cusumano{-}Towner}, {and} \bibinfo{person}{Vikash~K. Mansinghka}.} \bibinfo{year}{2022}\natexlab{}.
\newblock \showarticletitle{Recursive {M}onte {C}arlo and variational inference with auxiliary variables}. In \bibinfo{booktitle}{\emph{Uncertainty in Artificial Intelligence, Proceedings of the Thirty-Eighth Conference on Uncertainty in Artificial Intelligence, {UAI} 2022, 1-5 August 2022, Eindhoven, The Netherlands}} \emph{(\bibinfo{series}{Proceedings of Machine Learning Research}, Vol.~\bibinfo{volume}{180})}. \bibinfo{publisher}{{PMLR}}, \bibinfo{pages}{1096--1106}.
\newblock
\urldef\tempurl%
\url{https://proceedings.mlr.press/v180/lew22a.html}
\showURL{%
\tempurl}


\bibitem[Lew et~al\mbox{.}(2019)]%
        {lew2019trace}
\bibfield{author}{\bibinfo{person}{Alexander~K Lew}, \bibinfo{person}{Marco~F Cusumano-Towner}, \bibinfo{person}{Benjamin Sherman}, \bibinfo{person}{Michael Carbin}, {and} \bibinfo{person}{Vikash~K Mansinghka}.} \bibinfo{year}{2019}\natexlab{}.
\newblock \showarticletitle{Trace types and denotational semantics for sound programmable inference in probabilistic languages}.
\newblock \bibinfo{journal}{\emph{Proceedings of the ACM on Programming Languages}} \bibinfo{volume}{4}, \bibinfo{number}{POPL} (\bibinfo{year}{2019}), \bibinfo{pages}{1--32}.
\newblock


\bibitem[Lew et~al\mbox{.}(2023a)]%
        {lew_probabilistic_2023}
\bibfield{author}{\bibinfo{person}{Alexander~K. Lew}, \bibinfo{person}{Matin Ghavamizadeh}, \bibinfo{person}{Martin~C. Rinard}, {and} \bibinfo{person}{Vikash~K. Mansinghka}.} \bibinfo{year}{2023}\natexlab{a}.
\newblock \showarticletitle{Probabilistic {Programming} with {Stochastic} {Probabilities}}.
\newblock \bibinfo{journal}{\emph{Proceedings of the ACM on Programming Languages}} \bibinfo{volume}{7}, \bibinfo{number}{PLDI} (\bibinfo{date}{June} \bibinfo{year}{2023}), \bibinfo{pages}{176:1708--176:1732}.
\newblock
\urldef\tempurl%
\url{https://doi.org/10.1145/3591290}
\showDOI{\tempurl}


\bibitem[Lew et~al\mbox{.}(2023b)]%
        {lew_adev_2023}
\bibfield{author}{\bibinfo{person}{Alexander~K. Lew}, \bibinfo{person}{Mathieu Huot}, \bibinfo{person}{Sam Staton}, {and} \bibinfo{person}{Vikash~K. Mansinghka}.} \bibinfo{year}{2023}\natexlab{b}.
\newblock \showarticletitle{{ADEV}: {Sound} {Automatic} {Differentiation} of {Expected} {Values} of {Probabilistic} {Programs}}.
\newblock \bibinfo{journal}{\emph{Proceedings of the ACM on Programming Languages}} \bibinfo{volume}{7}, \bibinfo{number}{POPL} (\bibinfo{date}{Jan.} \bibinfo{year}{2023}), \bibinfo{pages}{121--153}.
\newblock
\showISSN{2475-1421}
\urldef\tempurl%
\url{https://doi.org/10.1145/3571198}
\showDOI{\tempurl}
\newblock
\shownote{arXiv:2212.06386 [cs, stat]}.


\bibitem[Li et~al\mbox{.}(2023b)]%
        {li2023type}
\bibfield{author}{\bibinfo{person}{Jianlin Li}, \bibinfo{person}{Leni Ven}, \bibinfo{person}{Pengyuan Shi}, {and} \bibinfo{person}{Yizhou Zhang}.} \bibinfo{year}{2023}\natexlab{b}.
\newblock \showarticletitle{Type-preserving, dependence-aware guide generation for sound, effective amortized probabilistic inference}.
\newblock \bibinfo{journal}{\emph{Proceedings of the ACM on Programming Languages}} \bibinfo{volume}{7}, \bibinfo{number}{POPL} (\bibinfo{year}{2023}), \bibinfo{pages}{1454--1482}.
\newblock


\bibitem[Li et~al\mbox{.}(2023a)]%
        {li_neural_2023}
\bibfield{author}{\bibinfo{person}{Michael~Y. Li}, \bibinfo{person}{Dieterich Lawson}, {and} \bibinfo{person}{Scott Linderman}.} \bibinfo{year}{2023}\natexlab{a}.
\newblock \showarticletitle{Neural {Adaptive} {Smoothing} via {Twisting}}.
\newblock
\urldef\tempurl%
\url{https://openreview.net/forum?id=rC6-kGN-0v}
\showURL{%
\tempurl}


\bibitem[Lund{\'e}n et~al\mbox{.}(2021)]%
        {lunden2021correctness}
\bibfield{author}{\bibinfo{person}{Daniel Lund{\'e}n}, \bibinfo{person}{Johannes Borgstr{\"o}m}, {and} \bibinfo{person}{David Broman}.} \bibinfo{year}{2021}\natexlab{}.
\newblock \showarticletitle{Correctness of Sequential {M}onte {C}arlo Inference for Probabilistic Programming Languages.}. In \bibinfo{booktitle}{\emph{ESOP}}. \bibinfo{pages}{404--431}.
\newblock


\bibitem[Maaløe et~al\mbox{.}(2016)]%
        {maaloe_auxiliary_2016}
\bibfield{author}{\bibinfo{person}{Lars Maaløe}, \bibinfo{person}{Casper~Kaae Sønderby}, \bibinfo{person}{Søren~Kaae Sønderby}, {and} \bibinfo{person}{Ole Winther}.} \bibinfo{year}{2016}\natexlab{}.
\newblock \showarticletitle{Auxiliary {Deep} {Generative} {Models}}. In \bibinfo{booktitle}{\emph{Proceedings of {The} 33rd {International} {Conference} on {Machine} {Learning}}}. \bibinfo{publisher}{PMLR}, \bibinfo{pages}{1445--1453}.
\newblock
\urldef\tempurl%
\url{https://proceedings.mlr.press/v48/maaloe16.html}
\showURL{%
\tempurl}
\newblock
\shownote{ISSN: 1938-7228}.


\bibitem[Maddison et~al\mbox{.}(2017)]%
        {maddison_filtering_2017}
\bibfield{author}{\bibinfo{person}{Chris~J. Maddison}, \bibinfo{person}{Dieterich Lawson}, \bibinfo{person}{George Tucker}, \bibinfo{person}{Nicolas Heess}, \bibinfo{person}{Mohammad Norouzi}, \bibinfo{person}{Andriy Mnih}, \bibinfo{person}{Arnaud Doucet}, {and} \bibinfo{person}{Yee~Whye Teh}.} \bibinfo{year}{2017}\natexlab{}.
\newblock \bibinfo{title}{Filtering {Variational} {Objectives}}.
\newblock
\newblock
\urldef\tempurl%
\url{https://doi.org/10.48550/arXiv.1705.09279}
\showDOI{\tempurl}
\newblock
\shownote{arXiv:1705.09279 [cs, stat]}.


\bibitem[Malkin et~al\mbox{.}(2022)]%
        {malkin2022gflownets}
\bibfield{author}{\bibinfo{person}{Nikolay Malkin}, \bibinfo{person}{Salem Lahlou}, \bibinfo{person}{Tristan Deleu}, \bibinfo{person}{Xu Ji}, \bibinfo{person}{Edward Hu}, \bibinfo{person}{Katie Everett}, \bibinfo{person}{Dinghuai Zhang}, {and} \bibinfo{person}{Yoshua Bengio}.} \bibinfo{year}{2022}\natexlab{}.
\newblock \showarticletitle{GFlowNets and variational inference}.
\newblock \bibinfo{journal}{\emph{arXiv preprint arXiv:2210.00580}} (\bibinfo{year}{2022}).
\newblock


\bibitem[Mansinghka et~al\mbox{.}(2014)]%
        {mansinghka2014venture}
\bibfield{author}{\bibinfo{person}{Vikash Mansinghka}, \bibinfo{person}{Daniel Selsam}, {and} \bibinfo{person}{Yura Perov}.} \bibinfo{year}{2014}\natexlab{}.
\newblock \showarticletitle{Venture: a higher-order probabilistic programming platform with programmable inference}.
\newblock \bibinfo{journal}{\emph{arXiv preprint arXiv:1404.0099}} (\bibinfo{year}{2014}).
\newblock


\bibitem[Mansinghka et~al\mbox{.}(2018)]%
        {mansinghka2018probabilistic}
\bibfield{author}{\bibinfo{person}{Vikash~K Mansinghka}, \bibinfo{person}{Ulrich Schaechtle}, \bibinfo{person}{Shivam Handa}, \bibinfo{person}{Alexey Radul}, \bibinfo{person}{Yutian Chen}, {and} \bibinfo{person}{Martin Rinard}.} \bibinfo{year}{2018}\natexlab{}.
\newblock \showarticletitle{Probabilistic programming with programmable inference}. In \bibinfo{booktitle}{\emph{Proceedings of the 39th ACM SIGPLAN Conference on Programming Language Design and Implementation}}. \bibinfo{pages}{603--616}.
\newblock


\bibitem[Michel et~al\mbox{.}(2024)]%
        {michel2024distributions}
\bibfield{author}{\bibinfo{person}{Jesse Michel}, \bibinfo{person}{Kevin Mu}, \bibinfo{person}{Xuanda Yang}, \bibinfo{person}{Sai~Praveen Bangaru}, \bibinfo{person}{Elias~Rojas Collins}, \bibinfo{person}{Gilbert Bernstein}, \bibinfo{person}{Jonathan Ragan-Kelley}, \bibinfo{person}{Michael Carbin}, {and} \bibinfo{person}{Tzu-Mao Li}.} \bibinfo{year}{2024}\natexlab{}.
\newblock \showarticletitle{Distributions for Compositionally Differentiating Parametric Discontinuities}.
\newblock \bibinfo{journal}{\emph{Proceedings of the ACM on Programming Languages}} \bibinfo{volume}{8}, \bibinfo{number}{OOPSLA1} (\bibinfo{year}{2024}), \bibinfo{pages}{893--922}.
\newblock


\bibitem[Mohamed et~al\mbox{.}(2020)]%
        {mohamed2020monte}
\bibfield{author}{\bibinfo{person}{Shakir Mohamed}, \bibinfo{person}{Mihaela Rosca}, \bibinfo{person}{Michael Figurnov}, {and} \bibinfo{person}{Andriy Mnih}.} \bibinfo{year}{2020}\natexlab{}.
\newblock \showarticletitle{{M}onte {C}arlo gradient estimation in machine learning}.
\newblock \bibinfo{journal}{\emph{Journal of Machine Learning Research}} \bibinfo{volume}{21}, \bibinfo{number}{132} (\bibinfo{year}{2020}), \bibinfo{pages}{1--62}.
\newblock


\bibitem[Naesseth et~al\mbox{.}(2018)]%
        {naesseth2018variational}
\bibfield{author}{\bibinfo{person}{Christian Naesseth}, \bibinfo{person}{Scott Linderman}, \bibinfo{person}{Rajesh Ranganath}, {and} \bibinfo{person}{David Blei}.} \bibinfo{year}{2018}\natexlab{}.
\newblock \showarticletitle{Variational sequential {M}onte {C}arlo}. In \bibinfo{booktitle}{\emph{International conference on artificial intelligence and statistics}}. PMLR, \bibinfo{pages}{968--977}.
\newblock


\bibitem[Naesseth et~al\mbox{.}(2020)]%
        {naesseth_markovian_2020}
\bibfield{author}{\bibinfo{person}{Christian~A. Naesseth}, \bibinfo{person}{Fredrik Lindsten}, {and} \bibinfo{person}{David Blei}.} \bibinfo{year}{2020}\natexlab{}.
\newblock \showarticletitle{Markovian score climbing: variational inference with {KL}(p{\textbar}{\textbar}q)}. In \bibinfo{booktitle}{\emph{Proceedings of the 34th {International} {Conference} on {Neural} {Information} {Processing} {Systems}}} \emph{(\bibinfo{series}{{NIPS}'20})}. \bibinfo{publisher}{Curran Associates Inc.}, \bibinfo{address}{Red Hook, NY, USA}, \bibinfo{pages}{15499--15510}.
\newblock
\showISBNx{978-1-71382-954-6}


\bibitem[Naesseth et~al\mbox{.}(2019)]%
        {naesseth2019elements}
\bibfield{author}{\bibinfo{person}{Christian~A Naesseth}, \bibinfo{person}{Fredrik Lindsten}, \bibinfo{person}{Thomas~B Sch{\"o}n}, {et~al\mbox{.}}} \bibinfo{year}{2019}\natexlab{}.
\newblock \showarticletitle{Elements of sequential {M}onte {C}arlo}.
\newblock \bibinfo{journal}{\emph{Foundations and Trends{\textregistered} in Machine Learning}} \bibinfo{volume}{12}, \bibinfo{number}{3} (\bibinfo{year}{2019}), \bibinfo{pages}{307--392}.
\newblock


\bibitem[Narayanan et~al\mbox{.}(2016)]%
        {narayanan2016probabilistic}
\bibfield{author}{\bibinfo{person}{Praveen Narayanan}, \bibinfo{person}{Jacques Carette}, \bibinfo{person}{Wren Romano}, \bibinfo{person}{Chung-chieh Shan}, {and} \bibinfo{person}{Robert Zinkov}.} \bibinfo{year}{2016}\natexlab{}.
\newblock \showarticletitle{Probabilistic inference by program transformation in Hakaru (system description)}. In \bibinfo{booktitle}{\emph{Functional and Logic Programming: 13th International Symposium, FLOPS 2016, Kochi, Japan, March 4-6, 2016, Proceedings 13}}. Springer, \bibinfo{pages}{62--79}.
\newblock


\bibitem[Obermeyer et~al\mbox{.}(2019a)]%
        {obermeyer2019functional}
\bibfield{author}{\bibinfo{person}{Fritz Obermeyer}, \bibinfo{person}{Eli Bingham}, \bibinfo{person}{Martin Jankowiak}, \bibinfo{person}{Du Phan}, {and} \bibinfo{person}{Jonathan~P Chen}.} \bibinfo{year}{2019}\natexlab{a}.
\newblock \showarticletitle{Functional tensors for probabilistic programming}.
\newblock \bibinfo{journal}{\emph{arXiv preprint arXiv:1910.10775}} (\bibinfo{year}{2019}).
\newblock


\bibitem[Obermeyer et~al\mbox{.}(2019b)]%
        {obermeyer2019tensor}
\bibfield{author}{\bibinfo{person}{Fritz Obermeyer}, \bibinfo{person}{Eli Bingham}, \bibinfo{person}{Martin Jankowiak}, \bibinfo{person}{Neeraj Pradhan}, \bibinfo{person}{Justin Chiu}, \bibinfo{person}{Alexander Rush}, {and} \bibinfo{person}{Noah Goodman}.} \bibinfo{year}{2019}\natexlab{b}.
\newblock \showarticletitle{Tensor variable elimination for plated factor graphs}. In \bibinfo{booktitle}{\emph{International Conference on Machine Learning}}. PMLR, \bibinfo{pages}{4871--4880}.
\newblock


\bibitem[Pu et~al\mbox{.}(2016)]%
        {pu2016variational}
\bibfield{author}{\bibinfo{person}{Yunchen Pu}, \bibinfo{person}{Zhe Gan}, \bibinfo{person}{Ricardo Henao}, \bibinfo{person}{Xin Yuan}, \bibinfo{person}{Chunyuan Li}, \bibinfo{person}{Andrew Stevens}, {and} \bibinfo{person}{Lawrence Carin}.} \bibinfo{year}{2016}\natexlab{}.
\newblock \showarticletitle{Variational autoencoder for deep learning of images, labels and captions}.
\newblock \bibinfo{journal}{\emph{Advances in Neural Information Processing Systems}}  \bibinfo{volume}{29} (\bibinfo{year}{2016}).
\newblock


\bibitem[Radul et~al\mbox{.}(2022)]%
        {radul2022you}
\bibfield{author}{\bibinfo{person}{Alexey Radul}, \bibinfo{person}{Adam Paszke}, \bibinfo{person}{Roy Frostig}, \bibinfo{person}{Matthew Johnson}, {and} \bibinfo{person}{Dougal Maclaurin}.} \bibinfo{year}{2022}\natexlab{}.
\newblock \showarticletitle{You only linearize once: Tangents transpose to gradients}.
\newblock \bibinfo{journal}{\emph{arXiv preprint arXiv:2204.10923}} (\bibinfo{year}{2022}).
\newblock


\bibitem[Rainforth et~al\mbox{.}(2018)]%
        {rainforth_tighter_2018}
\bibfield{author}{\bibinfo{person}{Tom Rainforth}, \bibinfo{person}{Adam~R. Kosiorek}, \bibinfo{person}{Tuan~Anh Le}, \bibinfo{person}{Chris~J. Maddison}, \bibinfo{person}{Maximilian Igl}, \bibinfo{person}{Frank Wood}, {and} \bibinfo{person}{Yee~Whye Teh}.} \bibinfo{year}{2018}\natexlab{}.
\newblock \bibinfo{title}{Tighter {Variational} {Bounds} are {Not} {Necessarily} {Better}}.
\newblock
\newblock
\urldef\tempurl%
\url{https://arxiv.org/abs/1802.04537v3}
\showURL{%
\tempurl}


\bibitem[Ranganath et~al\mbox{.}(2016)]%
        {ranganath_hierarchical_2016}
\bibfield{author}{\bibinfo{person}{Rajesh Ranganath}, \bibinfo{person}{Dustin Tran}, {and} \bibinfo{person}{David Blei}.} \bibinfo{year}{2016}\natexlab{}.
\newblock \showarticletitle{Hierarchical {Variational} {Models}}. In \bibinfo{booktitle}{\emph{Proceedings of {The} 33rd {International} {Conference} on {Machine} {Learning}}}. \bibinfo{publisher}{PMLR}, \bibinfo{pages}{324--333}.
\newblock
\urldef\tempurl%
\url{https://proceedings.mlr.press/v48/ranganath16.html}
\showURL{%
\tempurl}
\newblock
\shownote{ISSN: 1938-7228}.


\bibitem[Rubin(1988)]%
        {Rubin1988UsingTS}
\bibfield{author}{\bibinfo{person}{D.B. Rubin}.} \bibinfo{year}{1988}\natexlab{}.
\newblock \showarticletitle{Using the {SIR} algorithm to simulate posterior distributions}.
\newblock
\urldef\tempurl%
\url{https://api.semanticscholar.org/CorpusID:115305396}
\showURL{%
\tempurl}


\bibitem[Salimans et~al\mbox{.}(2015)]%
        {salimans2015markov}
\bibfield{author}{\bibinfo{person}{Tim Salimans}, \bibinfo{person}{Diederik Kingma}, {and} \bibinfo{person}{Max Welling}.} \bibinfo{year}{2015}\natexlab{}.
\newblock \showarticletitle{Markov chain {M}onte {C}arlo and variational inference: Bridging the gap}. In \bibinfo{booktitle}{\emph{International Conference on Machine Learning}}. PMLR, \bibinfo{pages}{1218--1226}.
\newblock


\bibitem[Schulman et~al\mbox{.}(2015)]%
        {schulman2015gradient}
\bibfield{author}{\bibinfo{person}{John Schulman}, \bibinfo{person}{Nicolas Heess}, \bibinfo{person}{Theophane Weber}, {and} \bibinfo{person}{Pieter Abbeel}.} \bibinfo{year}{2015}\natexlab{}.
\newblock \showarticletitle{Gradient estimation using stochastic computation graphs}.
\newblock \bibinfo{journal}{\emph{Advances in Neural Information Processing Systems}}  \bibinfo{volume}{28} (\bibinfo{year}{2015}).
\newblock


\bibitem[Sherman et~al\mbox{.}(2021)]%
        {sherman_s_2021}
\bibfield{author}{\bibinfo{person}{Benjamin Sherman}, \bibinfo{person}{Jesse Michel}, {and} \bibinfo{person}{Michael Carbin}.} \bibinfo{year}{2021}\natexlab{}.
\newblock \showarticletitle{$\lambda${S}: computable semantics for differentiable programming with higher-order functions and datatypes}.
\newblock \bibinfo{journal}{\emph{Proceedings of the ACM on Programming Languages}} \bibinfo{volume}{5}, \bibinfo{number}{POPL} (\bibinfo{date}{Jan.} \bibinfo{year}{2021}), \bibinfo{pages}{3:1--3:31}.
\newblock
\urldef\tempurl%
\url{https://doi.org/10.1145/3434284}
\showDOI{\tempurl}


\bibitem[Sobolev and Vetrov(2019)]%
        {sobolev_importance_2019}
\bibfield{author}{\bibinfo{person}{Artem Sobolev} {and} \bibinfo{person}{Dmitry Vetrov}.} \bibinfo{year}{2019}\natexlab{}.
\newblock \bibinfo{title}{Importance {Weighted} {Hierarchical} {Variational} {Inference}}.
\newblock
\newblock
\urldef\tempurl%
\url{https://doi.org/10.48550/arXiv.1905.03290}
\showDOI{\tempurl}
\newblock
\shownote{arXiv:1905.03290 [cs, stat]}.


\bibitem[Sohn et~al\mbox{.}(2015)]%
        {sohn_learning_2015}
\bibfield{author}{\bibinfo{person}{Kihyuk Sohn}, \bibinfo{person}{Honglak Lee}, {and} \bibinfo{person}{Xinchen Yan}.} \bibinfo{year}{2015}\natexlab{}.
\newblock \showarticletitle{Learning {Structured} {Output} {Representation} using {Deep} {Conditional} {Generative} {Models}}. In \bibinfo{booktitle}{\emph{Advances in {Neural} {Information} {Processing} {Systems}}}, Vol.~\bibinfo{volume}{28}. \bibinfo{publisher}{Curran Associates, Inc.}
\newblock
\urldef\tempurl%
\url{https://proceedings.neurips.cc/paper_files/paper/2015/hash/8d55a249e6baa5c06772297520da2051-Abstract.html}
\showURL{%
\tempurl}


\bibitem[Stites et~al\mbox{.}(2021)]%
        {stites_learning_2021}
\bibfield{author}{\bibinfo{person}{Sam Stites}, \bibinfo{person}{Heiko Zimmermann}, \bibinfo{person}{Hao Wu}, \bibinfo{person}{Eli Sennesh}, {and} \bibinfo{person}{Jan-Willem van~de Meent}.} \bibinfo{year}{2021}\natexlab{}.
\newblock \showarticletitle{Learning proposals for probabilistic programs with inference combinators}. In \bibinfo{booktitle}{\emph{Proceedings of the {Thirty}-{Seventh} {Conference} on {Uncertainty} in {Artificial} {Intelligence}}}. \bibinfo{publisher}{PMLR}, \bibinfo{pages}{1056--1066}.
\newblock
\urldef\tempurl%
\url{https://proceedings.mlr.press/v161/stites21a.html}
\showURL{%
\tempurl}
\newblock
\shownote{ISSN: 2640-3498}.


\bibitem[Tran et~al\mbox{.}(2018)]%
        {tran2018simple}
\bibfield{author}{\bibinfo{person}{Dustin Tran}, \bibinfo{person}{Matthew~D. Hoffman}, \bibinfo{person}{Dave Moore}, \bibinfo{person}{Christopher Suter}, \bibinfo{person}{Srinivas Vasudevan}, {and} \bibinfo{person}{Alexey Radul}.} \bibinfo{year}{2018}\natexlab{}.
\newblock \showarticletitle{Simple, Distributed, and Accelerated Probabilistic Programming}. In \bibinfo{booktitle}{\emph{Advances in Neural Information Processing Systems 31: Annual Conference on Neural Information Processing Systems 2018, NeurIPS 2018, December 3-8, 2018, Montr{\'{e}}al, Canada}}, \bibfield{editor}{\bibinfo{person}{Samy Bengio}, \bibinfo{person}{Hanna~M. Wallach}, \bibinfo{person}{Hugo Larochelle}, \bibinfo{person}{Kristen Grauman}, \bibinfo{person}{Nicol{\`{o}} Cesa{-}Bianchi}, {and} \bibinfo{person}{Roman Garnett}} (Eds.). \bibinfo{pages}{7609--7620}.
\newblock
\urldef\tempurl%
\url{https://proceedings.neurips.cc/paper/2018/hash/201e5bacd665709851b77148e225b332-Abstract.html}
\showURL{%
\tempurl}


\bibitem[Tran et~al\mbox{.}(2017)]%
        {tran_2017_deep}
\bibfield{author}{\bibinfo{person}{Dustin Tran}, \bibinfo{person}{Matthew~D. Hoffman}, \bibinfo{person}{Rif~A. Saurous}, \bibinfo{person}{Eugene Brevdo}, \bibinfo{person}{Kevin Murphy}, {and} \bibinfo{person}{David~M. Blei}.} \bibinfo{year}{2017}\natexlab{}.
\newblock \showarticletitle{Deep Probabilistic Programming}. In \bibinfo{booktitle}{\emph{5th International Conference on Learning Representations, {ICLR} 2017, Toulon, France, April 24-26, 2017, Conference Track Proceedings}}. \bibinfo{publisher}{OpenReview.net}.
\newblock
\urldef\tempurl%
\url{https://openreview.net/forum?id=Hy6b4Pqee}
\showURL{%
\tempurl}


\bibitem[Vahdat and Kautz(2020)]%
        {vahdat2020nvae}
\bibfield{author}{\bibinfo{person}{Arash Vahdat} {and} \bibinfo{person}{Jan Kautz}.} \bibinfo{year}{2020}\natexlab{}.
\newblock \showarticletitle{NVAE: A deep hierarchical variational autoencoder}.
\newblock \bibinfo{journal}{\emph{Advances in Neural Information Processing Systems}}  \bibinfo{volume}{33} (\bibinfo{year}{2020}), \bibinfo{pages}{19667--19679}.
\newblock


\bibitem[Wagner(2023)]%
        {wagner2023fast}
\bibfield{author}{\bibinfo{person}{Dominik Wagner}.} \bibinfo{year}{2023}\natexlab{}.
\newblock \emph{\bibinfo{title}{Fast and correct variational inference for probabilistic programming: Differentiability, reparameterisation and smoothing}}.
\newblock \bibinfo{thesistype}{Ph.\,D. Dissertation}. \bibinfo{school}{University of Oxford}.
\newblock


\bibitem[Wang et~al\mbox{.}(2021)]%
        {wang2021sound}
\bibfield{author}{\bibinfo{person}{Di Wang}, \bibinfo{person}{Jan Hoffmann}, {and} \bibinfo{person}{Thomas Reps}.} \bibinfo{year}{2021}\natexlab{}.
\newblock \showarticletitle{Sound probabilistic inference via guide types}. In \bibinfo{booktitle}{\emph{Proceedings of the 42nd ACM SIGPLAN International Conference on Programming Language Design and Implementation}}. \bibinfo{pages}{788--803}.
\newblock


\bibitem[Weber et~al\mbox{.}(2019)]%
        {weber2019credit}
\bibfield{author}{\bibinfo{person}{Th{\'e}ophane Weber}, \bibinfo{person}{Nicolas Heess}, \bibinfo{person}{Lars Buesing}, {and} \bibinfo{person}{David Silver}.} \bibinfo{year}{2019}\natexlab{}.
\newblock \showarticletitle{Credit assignment techniques in stochastic computation graphs}. In \bibinfo{booktitle}{\emph{The 22nd International Conference on Artificial Intelligence and Statistics}}. PMLR, \bibinfo{pages}{2650--2660}.
\newblock


\bibitem[Zimmermann et~al\mbox{.}(2021)]%
        {zimmermann_nested_2021}
\bibfield{author}{\bibinfo{person}{Heiko Zimmermann}, \bibinfo{person}{Hao Wu}, \bibinfo{person}{Babak Esmaeili}, {and} \bibinfo{person}{Jan-Willem van~de Meent}.} \bibinfo{year}{2021}\natexlab{}.
\newblock \showarticletitle{Nested {Variational} {Inference}}.
\newblock
\urldef\tempurl%
\url{https://openreview.net/forum?id=kBrHzFtwdp}
\showURL{%
\tempurl}


\bibitem[Ścibior et~al\mbox{.}(2018)]%
        {scibior_functional_2018}
\bibfield{author}{\bibinfo{person}{Adam Ścibior}, \bibinfo{person}{Ohad Kammar}, {and} \bibinfo{person}{Zoubin Ghahramani}.} \bibinfo{year}{2018}\natexlab{}.
\newblock \showarticletitle{Functional programming for modular {{B}ayesian} inference}.
\newblock \bibinfo{journal}{\emph{Proceedings of the ACM on Programming Languages}} \bibinfo{volume}{2}, \bibinfo{number}{ICFP} (\bibinfo{date}{July} \bibinfo{year}{2018}), \bibinfo{pages}{83:1--83:29}.
\newblock
\urldef\tempurl%
\url{https://doi.org/10.1145/3236778}
\showDOI{\tempurl}


\bibitem[Ścibior et~al\mbox{.}(2017)]%
        {scibior_denotational_2017}
\bibfield{author}{\bibinfo{person}{Adam Ścibior}, \bibinfo{person}{Ohad Kammar}, \bibinfo{person}{Matthijs Vákár}, \bibinfo{person}{Sam Staton}, \bibinfo{person}{Hongseok Yang}, \bibinfo{person}{Yufei Cai}, \bibinfo{person}{Klaus Ostermann}, \bibinfo{person}{Sean~K. Moss}, \bibinfo{person}{Chris Heunen}, {and} \bibinfo{person}{Zoubin Ghahramani}.} \bibinfo{year}{2017}\natexlab{}.
\newblock \showarticletitle{Denotational validation of higher-order {{B}ayesian} inference}.
\newblock \bibinfo{journal}{\emph{Proceedings of the ACM on Programming Languages}} \bibinfo{volume}{2}, \bibinfo{number}{POPL} (\bibinfo{date}{Dec.} \bibinfo{year}{2017}), \bibinfo{pages}{60:1--60:29}.
\newblock
\urldef\tempurl%
\url{https://doi.org/10.1145/3158148}
\showDOI{\tempurl}


\end{thebibliography}

\appendix

\newpage
\section{Full Language and System}
\label{appdx:full-system}
We elaborate on our discussion in \S\ref{sec:sp}, and provide a sketch of the changes that would be necessary to adapt our formal model and proofs to the generative language extended with $\textbf{marginal}$ and $\textbf{normalize}$, and the reverse mode implementation of ADEV.

\subsection{New Language Features for Expressive Models and Variational Families}
\label{sec:margnorm}
The full language for models and variational families modularly extends our formalized core with two new constructs: $\marginal$, for marginalizing auxiliary variables from the model or variational family; and $\normalize$, for building variational families that use Monte Carlo to approximate the normalized posteriors of other programs. We discuss each feature in turn.\\\vspace{-2mm}

\noindent\textbf{Marginalization of auxiliary variables.} A key requirement in variational inference is that the model and variational family be defined over the same set of random variables~\cite{lee_towards_2019,lew2019trace,wang2021sound,li2023type}. But it is often desirable to introduce \textit{auxiliary variables} to only the variational family, or only the model, in order to make the \textit{marginal} distribution on the shared variables more interesting. Fig.~\ref{fig:example_transcript} illustrates one example. Both the model and variational family are defined over $x$ and $y$, and intuitively, we want the variational family to learn a distribution that concentrates mass around the circumference of a circle of radius $\sqrt{5}$. But using Gaussian primitives alone we can never express such a distribution. Therefore, $\texttt{guide\_joint}$ introduces an auxiliary variable $v \sim U([0, 2\pi])$, which informs the location of both $x$ and $y$. Because $\texttt{guide\_joint}$ is now defined over the three variables $(v, x, y)$, it is no longer directly usable as a variational family for the model. Using $\marginal$, we can fix the problem, constructing a new program representing $\texttt{guide\_joint}$'s marginal on $x$ and $y$.

The syntax for marginalizing is $\marginal~\textit{names}~\textit{program}~\textit{algorithm}$. The argument \textit{names} is a list of strings indicating which of the variables from \textit{program} should be kept. Marginal distributions may have intractable density functions, and so our system automates \textit{differentiable stochastic estimates} of their densities instead (\S\ref{sec:diffsp}). The argument \textit{algorithm} specifies a strategy for calculating these estimates. It is given as a function, taking in values for the variables in \textit{names}, and outputting an \textit{algorithm} for inferring the posterior over auxiliary variables (see \S\ref{sec:diffsp} for details).%of the form $\importance~\textit{proposal}~K$, specifying a \textit{proposal distribution} over the auxiliary variables and a number of particles $K$.

\textit{Usage.} With $\marginal$, users can encode within a PPL many variational inference algorithms that cannot be easily expressed in existing systems. For example, hierarchical variational inference (HVI)~\cite{ranganath_hierarchical_2016,agakov_auxiliary_2004} can be implemented by introducing then marginalizing auxiliary variables in a variational family, using 1-particle importance sampling as the \textit{algorithm} argument. Increasing the number of particles used to estimate the marginal density of the variational family yields importance-weighted HVI~\cite{sobolev_importance_2019}. Parameterizing the proposal used to marginalize the variational family as a neural network recovers auxiliary deep generative models (ADGM)~\cite{maaloe_auxiliary_2016}. Nesting a $\marginal$ distribution as the proposal to marginalize a variational family is an instance of recursive auxiliary-variable inference (RAVI)~\cite{lew_recursive_2022}. Extending the variational family with steps from a Markov chain Monte Carlo algorithm, then marginalizing those, implements Markov chain variational inference (MCVI)~\cite{salimans2015markov}. And many aspects of these algorithms could in principle be composed to develop new inference algorithms altogether.\\

\noindent\textbf{Approximate normalization.} Another way to make variational families more expressive is to add general-purpose inference logic to the variational family itself. To do this, users can apply our $\normalize$ construct, which builds a probabilistic program that represents the output distribution of a given inference algorithm applied to a given (unnormalized) probabilistic program. Its syntax is $\normalize~\textit{program}~\textit{algorithm}$.

\textit{Usage.} Fig.~\ref{fig:objectives} illustrates a typical usage: after defining a simple variational family ($\texttt{q}_{\texttt{NAIVE}}$), we define an improved variational family ($\texttt{q}_{\texttt{SIR}}$) that performs $N$-particle importance sampling, targeting the model's posterior, using $\texttt{q}_{\texttt{NAIVE}}$ as a proposal distribution. This normalized variational family can then be used in several ways. If we use $\texttt{q}_{\texttt{SIR}}$ as a variational family and optimize the usual ELBO objective (using our stochastically estimated densities for $\texttt{q}_{\texttt{SIR}}$, see \S\ref{sec:diffsp}), this is equivalent to optimizing the IWAE (or IWELBO) objective with $\texttt{q}_{\texttt{NAIVE}}$ as the variational family~\cite{burda_importance_2016}. Alternatively, we can optimize the objective $\phi \mapsto \mathbb{E}_{x \sim \sem{\texttt{q}_{\texttt{SIR}} {\phi_0}}}[-\log ~\sem{\textbf{density}\{\texttt{q}_{\texttt{NAIVE}}\}}_\phi(x)]$, which tunes $\texttt{q}_{\texttt{NAIVE}}$ to minimize its KL divergence to the approximate inference algorithm $\texttt{q}_{\texttt{SIR}}$. This latter objective is the wake-phase $\phi$ objective in the reweighted wake sleep algorithm~\cite{bornschein_reweighted_2015,le_revisiting_2019}, and is illustrated as a $\lambda_\text{ADEV}$ program in \S\ref{appdx:objectives}.

\subsection{Differentiable Stochastic Estimators of Densities and their Reciprocals}
\label{sec:diffsp}

With the new features from \S\ref{sec:margnorm}, it is possible to define models and variational families whose exact densities cannot be efficiently computed. As such, in our full language, $\textbf{density}$ and $\textbf{sim}$ do not produce exact density evaluators, but rather density estimators that may be stochastic, depending on the input program. Importantly, these estimators are still implemented in the ADEV language, still come equipped with gradient estimation strategies, and are still guaranteed to satisfy the necessary differentiability requirements for gradients to be unbiased. We now discuss the implications of this stochastic estimation for users of our language; see \S~\ref{appdx:estimator-impl} for an overview of how the stochastic estimates are computed.

\textit{For users.} On the surface, very little changes when densities are made stochastic. The automated $\textbf{density}$ procedures now have type $\trace \to P~\RR$, instead of $\trace \to \RR$, and the specifications for both $\textbf{density}$ and $\textbf{sim}$ change slightly: now, $\textbf{density}\{t\}$ produces only \textit{unbiased estimates} of the density of $\sem{t}_1$, and $\textbf{sim}\{t\}$ produces pairs $(x, w)$ with the property that $\mathbb{E}[\frac{1}{w} \mid x]$ is the reciprocal density of $\sem{t}_1$ at $x$. 
%Note that this new specification of $\mathbf{density}$ and $\mathbf{sim}$ still allows us to relatively easily express variational objectives and produce sound gradient estimators that minimize them. 
But using density estimators instead of exact densities can change the meaning of user-specified variational objectives. Consider, for example, the ELBO, whose implementation is very slightly modified to account for the new type of $\textbf{density}$:
% As a representative example we consider the ELBO objective for a model $P_\theta$ and a variational family $Q_\phi$. If $t_p: G \tau$ and $t_q: \trace \to G \tau'$ implement $P_\theta$ and $Q_\phi$ in $\lambda_\text{Gen}$ and $y: \trace$ is denotes an observation, then the program
% \begin{align*}
% &\mathbb{E} \{ \haskdo. x, w_q \gets \mathbf{sim}\{t_q ~y\}; \\
% &\hspace{4mm} \llet\ w_p = \mathbf{density}\{t_p\} (x \mdoubleplus y); \\
% &\hspace{4mm} \return \log w_p / w_q\\
% &\hspace{2mm}\}
% \end{align*}
% implements the ELBO objective, given by
% $$
% \int \log \frac{P_\theta(x, y)}{Q_\phi(x|y)}\ \ Q_\phi(x|y)dx.
% $$
% With the new specification, we have a very similar program
% \begin{align*}
% &\mathbb{E} \{ \haskdo. x, w_q \gets \mathbf{sim}\{t_q ~y\}; \\
% &\hspace{4mm} w_p \gets \mathbf{density}\{t_p\} (x \mdoubleplus y); \\
% &\hspace{4mm} \return \log w_p / w_q\\
% &\hspace{2mm}\}
% \end{align*}
$$\textbf{ELBO}\coloneqq \lambda(\theta, \phi). \mathbb{E}(\haskdo~\{(x, w_q)\gets\textbf{sim}\{Q\} ~\phi; w_p\gets\textbf{density}\{P\} ~\theta ~x; \return\log w_p - \log w_q\}).$$

\noindent This program may no longer denote the usual ELBO objective, $\int \log \frac{\tilde{p}_\theta(x)}{q_\phi(x)}\ \ Q_\phi(x)dx.$ Instead, it denotes the objective
$$
(\theta, \phi) \mapsto \iint \log\frac{w_p}{w_q} \kappa_{p;\theta}(x, dw_p)\kappa_{q;\phi}(dx, dw_q),
$$
where we write $\kappa_p$ for the unbiased density estimation kernel for $P_\theta$ (the new specification for $\mathbf{density}\{t_p\}$) and $\kappa_q$ for the `reciprocal density simulation' measure for $Q_\phi$ (the new specification for $\mathbf{sim}\{t_q\}$). Fortunately, however, this is still a sensible variational objective. Because $\mathbb{E}[w_p \mid x] = \tilde{p}_\theta(x)$ and $\mathbb{E}[\frac{1}{w_q} \mid x] = \frac{1}{q_\phi(x)}$, we can apply Jensen's inequality and obtain that $\mathbb{E}[\log w_p \mid x] \leq \log \tilde{p}_\theta(x)$ and $\mathbb{E}[\log \frac{1}{w_q} \mid x] \leq \log \frac{1}{q_\phi(x)}$. This in turn implies that our new objective is a \textit{lower bound} on the original ELBO. In fact, it can be shown that the new bound decomposes as the original ELBO plus a non-negative term related to the average variation of the stochastic density estimates~\cite{lew_recursive_2022}. Thus, optimizing this modified ELBO simultaneously maximizes the original ELBO and minimizes the variability of the density estimates.

\subsection{Estimator implementation} 
\label{appdx:estimator-impl}
To implement density estimation, we take inspiration from~\cite{lew_probabilistic_2023}. We first extend our language with a new type for \textit{inference algorithms}; currently we have implemented only importance-sampling algorithms, constructed with the syntax $\textbf{importance}~\textit{proposal}~N$. The $\xi$ and $\chi$ transformations from \S\ref{sec:gen-transformations} are then extended to transform importance-sampling algorithms into ADEV-compatible procedures for, given an unnormalized model of type $G~\tau$, (1) generating collections of importance-weighted particles from the proposal, and (2) given a sample $x$, generating collections of \textit{conditional} importance samples (at least one of which is equal to $x$)~\cite{cusumano-towner_aide_2017,naesseth2019elements}. The importance sampling is implemented in terms of the proposal's \textbf{sim} method and the model's \textbf{density} method, so any gradient estimation logic is incorporated into the generated ADEV procedures. The $\marginal$ and $\normalize$ constructs expect algorithms as inputs, which they use as follows to estimate their densities and density reciprocals:
\begin{itemize}[leftmargin=*]
\item To implement \textbf{density} for $\marginal~\textit{names}~\textit{prog}~\textit{alg}$, we first construct an unnormalized posterior over auxiliary variables to marginalize by transforming \textit{prog} to $\observe$ all random variables in \textit{names}. We then estimate its normalizing constant\textemdash the marginal density of the observed choices\textemdash by generating (unconditional) importance samples from $\textit{alg}$ and averaging the importance weights. For $\textbf{sim}$, we simulate a complete trace $t$ from $\textit{prog}$ and project out the variables from $\textit{names}$ to obtain a marginal sample. We then proceed as before to estimate the density, but use conditional importance samples, conditioned on the auxiliary variables sampled when generating $t$. This can be shown to have the desired unbiasedness property~\cite{lew_probabilistic_2023}.

\item To implement \textbf{density} for $\normalize~\textit{prog}~\textit{alg}$ for a given trace $t$, we first call $\textbf{density}\{\textit{prog}\}$ on $t$ to get its \textit{unnormalized} density. We then use \textit{conditional} importance sampling (conditioned on $t$ using \textit{alg} targeting \textit{prog}, and compute the average weight. The ratio of the unnormalized density to the average weight is an unbiased estimate of the output density of sampling-importance-resampling (SIR)~\cite{Rubin1988UsingTS} at $t$. To implement $\textbf{sim}$, we draw unconditional importance samples, and resample one from a categorical distribution based on the importance weights. We then call  $\textbf{density}\{\textit{prog}\}$ on the chosen trace $t$ and return $(t, w)$, where $w$ is  the ratio of the chosen trace's unnormalized density to the average importance weight. The categorical sample needs a gradient estimator\textemdash we use $\textbf{categorical}_\text{ENUM}$.
\end{itemize}
\subsection{Reverse-Mode Automatic Differentiation of Expected Values}
\label{appdx:yolo}
The algorithm we present in \S\ref{fig:adev-transformation} is a \textit{forward-mode} AD algorithm, which works by propagating dual numbers. Unfortunately, forward-mode algorithms do not scale to typical deep learning applications, because they can compute only one dimension of the gradient per execution. Our implementations extend ADEV to support \textit{reverse-mode}, which can compute gradients in one pass:
\begin{itemize}[leftmargin=*]
    \item Our \texttt{genjax.vi} implementation is built on top of JAX, which takes a unique approach to reverse-mode for deterministic programs~\cite{radul2022you} (YOLO). It involves first applying a forward-mode program transformation, then separating the transformed program into linear and non-linear components, and finally transposing the linear component via another program transformation. This is a powerful technique that avoids the need to write a specialized reverse-mode algorithm, but only applies to first-order deterministic programs, which at first glance rules out ADEV's reliance on higher-order continuation-passing style and on random sampling. The first restriction can be lifted by partial evaluation: if we begin with a first-order \textit{user} program in \texttt{genjax.vi} --- after apply ADEV's transformations --- we can $\beta$-reduce away (via JAX's partial evaluation) any higher-order subterms introduced by the CPS transformation, allowing us to take first-order programs to first-order programs. Following YOLO, we then apply the unzip transformation $\mathcal{U}$, and the transposition transformation $\mathcal{T}$ to the deterministic parts of the resulting first-order programs. The restriction on random sampling can be lifted by the observation that sampling can be safely included in the nonlinear fragment of YOLO's language -- meaning that it is safe to extend YOLO with ADEV sampling primitives, so long as they implement gradient strategies that do not introduce sampling statements which depend on the linear tangent perturbations. We exemplify the process of using YOLO's transformations to derive a reverse mode gradient estimator in Fig.~\ref{fig:yolo}.

    \item In Julia and Haskell, we take a different approach. We begin with existing libraries for deterministic, reverse-mode AD, and attempt to transform the user's program into a host-language program that, when standard reverse-mode AD is applied, will yield a correct gradient estimator. To achieve this, we still apply a continuation-passing transformation to the user's program, and replace stochastic primitives with special CPS-style gradient estimators. But we must write the primitives so that, instead of stochastically returning manually constructed dual numbers, they stochastically return scalar losses whose reverse-mode derivatives are as desired. To achieve this, we use phony ``functions'' analogous to the ``magic-box'' operators in DICE~\cite{foerster2018dice} or Storchastic~\cite{krieken2021storchastic}, which introduce phantom dependencies on scalars so that their derivatives will include extra terms. This sort of trick is widely used to implement gradient estimators without giving up on deterministic AD's benefit; what is nice about combining the trick with ADEV is that (1) the hacky use of phony functions is limited to the implementations of primitives, which still have local specifications that can be separately reasoned about; and (2) we still benefit from the expressive power of ADEV's continuations, which can express gradient estimators that require much more substantial changes to the original program than simply tracking a few extra gradient terms.
\end{itemize}

% Mathieu's macros
\newcommand{\KK}{\kappa}
\newcommand{\II}{\mathbb{I}}
\newcommand{\BB}{\mathbb{B}}
\newcommand{\loss}{\mathcal{L}}
\newcommand{\syntaxcolor}[0]{\color{blue!70!black}}
\newcommand{\constantcolor}[1]{\mathbf{\syntaxcolor{#1}}}
\newcommand{\colforsyntax}[1]{\mathbf{\syntaxcolor{#1}}}
\newcommand{\iin}[0]{\colforsyntax{in}}
\newcommand{\normalreparam}[0]{\constantcolor{normal_\texttt{REPARAM}}}
\newcommand{\normalreinforce}[0]{\constantcolor{normal_\texttt{REINFORCE}}}
\newcommand{\geometricreinforce}[0]{\constantcolor{geometric_\texttt{REINFORCE}}}
\newcommand{\flipreinforce}[0]{\constantcolor{flip_\texttt{REINFORCE}}}
\newcommand{\flipenum}[0]{\constantcolor{flip_\texttt{ENUM}}}
\newcommand{\gammareparam}[0]{\constantcolor{gamma_\texttt{REPARAM}}}
\newcommand{\studenttreparam}[0]{\constantcolor{student\_t_\texttt{REPARAM}}}
\newcommand{\betareparam}[0]{\constantcolor{beta_\texttt{REPARAM}}}
\newcommand{\dirichletreparam}[0]{\constantcolor{dirichlet_\texttt{REPARAM}}}
\newcommand{\baseline}[0]{\constantcolor{baseline}}
\newcommand{\flipimportance}[0]{\constantcolor{flip_\texttt{IMPORTANCE}}}
\newcommand{\addcost}[0]{\constantcolor{add\_cost}}
\newcommand{\poissonweak}[0]{\constantcolor{poisson_\texttt{WEAK\_DERIV}}}
\newcommand{\gammaimplicit}[0]{\constantcolor{gamma_\texttt{IMPLICIT}}}
\newcommand{\reinforce}[0]{\constantcolor{reinforce}}
\newcommand{\leaveoneout}[0]{\constantcolor{leave\_one\_out}}
\newcommand{\implicitdiff}[0]{\constantcolor{implicit\_differentiation}}
\newcommand{\reparamreject}[0]{\constantcolor{reparam\_rejection\_sampler}}
\newcommand{\forget}[1]{\lfloor #1 \rfloor}
\newcommand{\kleisli}[0]{\mathbf{Kleisli}}
\newcommand{\mbe}{\constantcolor{\mathbb{E}}}
\newcommand{\mbb}{\mathbb{B}}
\newcommand{\mba}{\mathbb{A}}
\newcommand{\mbc}{\mathbb{C}}
\newcommand{\mbd}{\mathbb{D}}
\newcommand{\mbl}{\mathbb{L}}
% \newcommand{\dwRR}{\wRR_\mathcal{D}}
% \newcommand{\posreal}[0]{\RR_{> 0}}
% \newcommand{\eplus}{+^{\wRR}}
% \newcommand{\etimes}{\times^{\wRR}}
% \newcommand{\eexp}{exp^{\wRR}}
% \newcommand{\minibatch}{\constantcolor{minibatch}}
% \newcommand{\exact}{\constantcolor{exact}}
% \newcommand{\monadsymbol}{P}
% \newcommand{\pmonad}{\monadsymbol\,}
% \newcommand{\dpmonad}{\monadsymbol_\der\,}
% \newcommand{\vpmonad}{\monadsymbol_\ver\,}
% \newcommand{\rel}{\mathscr{R}}
% \newcommand{\arr}[1]{#1^{\to}}
% \newcommand{\pushoutcorner}[1][dr]{\save*!/#1+1.2pc/#1:(1,-1)@^{|-}\restore}
% \newcommand{\pullbackcorner}[1][dr]{\save*!/#1-1.2pc/#1:(-1,1)@^{|-}\restore}
% \newcommand{\lcstruct}[0]{\lambda_c(\Sigma)}
% \newcommand{\lifted}[0]{\stackrel{\cdot}{\to}}
% \newcommand{\Gl}[0]{\mathbf{Gl}}
% \newcommand{\semgl}[1]{\llparenthesis #1\rrparenthesis}
% \newcommand{\facto}[0]{\ensuremath{(\mathcal{E},\mathcal{M})}}
% \newcommand{\qbd}{\mathbf{Qbd}}
% \newcommand{\bdiff}{\mathbf{BDiff}}
% \newcommand{\wRR}{\widetilde{\RR}}

% Fig. 1 (RADEV LAFI)
\begin{figure*}[htb]
\scriptsize{
    \begin{tabular}{|l|l|l|}
         \multicolumn{1}{c}{\bf (a) Input Loss as a} & 
         \multicolumn{1}{c}{\bf(b) Apply Forward-mode} & 
         \multicolumn{1}{c}{\bf (c) Inline $\ad{\mbe}, \ad{\normalreparam}$ } \\
         \multicolumn{1}{c}{\bf Probabilistic Program} & 
         \multicolumn{1}{c}{\bf ADEV $\ad{-}$} & 
         \multicolumn{1}{c}{\bf and $\beta$-reduce $\lambda$'s} \\\hline
    \hspace{-1mm}\begin{tabular}{l}
    $\loss = \lambda (\otimes \underline{\theta};):\RR^2.\,  \mbe (\haskdo~\{$\\
    $\quad \llet\ (x;) \gets \normalreparam(\underline{\theta}_1,1)$\\
    $\quad \llet\ (y;) = \sin(x)$ \\
    $\quad \llet\ (z;) = y+\underline{\theta}_2$ \\
    $\quad\quad \return~z \})$ \\
     \\
     \\
     \\
     \\
     \\
     \\
    % $l = \lambda \theta. \mbe (f(\theta))$
    \end{tabular}
    & 
    \hspace{-1mm}\begin{tabular}{l}
    $d\loss = \lambda (\otimes \underline{\theta};\otimes \underline{\dot{\theta}}):\RR^2\times \RR^2.\,  \ad{\mbe} $\\
    $\big(\lambda dl.\haskdo~\{$\\
    $\quad \ad{\normalreparam}\big((\underline{\theta}_1,1),(\underline{\dot{\theta}}_1,0)\big)$\\
    $\quad \big(\lambda (x;\dot{x}). $\\
    $\quad\quad \llet\ (y;) = \sin(x)$\\
    $\quad\quad \llet\ (dy;) = \cos(x)$\\
    $\quad\quad \llet\ (;\dot{y}) = dy*\dot{x}$\\
    $\quad\quad \llet\ (z;) = y+\underline{\theta}_2$ \\
    $\quad\quad \llet\ (;\dot{z}) = \dot{y}+\underline{\dot{\theta}}_2$ \\
    $\quad\quad\quad dl(z,\dot{z})\big)\}\big)$
    \\
    \\
    \\
     \\
    %$l' = \lambda \theta. \mbe (g(\theta))$
    \end{tabular}
    & 
    \hspace{-1mm}\begin{tabular}{l}
    $d\loss = \lambda (\otimes \underline{\theta};\otimes \underline{\dot{\theta}}):\RR^2\times \RR^2.\, \haskdo~\{$\\
    $\quad \llet\ (\epsilon;) \gets \normalreparam(0,1)$\\
    $\quad \llet\ (x_1;) = \epsilon * 1$ \\
    $\quad \llet\ (;\dot{x_1}) = \epsilon * 0$ \\
    $\quad \llet\ (x;) = x_1 + \underline{\theta}_1$\\
    $\quad \llet\ (;\dot{x}) = \dot{x_1} + \underline{\dot{\theta}}_1$\\
    $\quad \llet\ (y;) = \sin(x)$\\
    $\quad\llet\ (dy;) = \cos(x)$\\
    $\quad \llet\ (;\dot{y}) = dy*\dot{x}$\\
    $\quad \llet\ (z;) = y+\underline{\theta}_2$ \\
    $\quad \llet\ (;\dot{z}) = \dot{y}+\underline{\dot{\theta}}_2$ \\
    $\quad\quad \return (z,\dot{z})\}$
    \end{tabular}\\\hline
    \end{tabular}
    
     \begin{tabular}{|l|l|l|}
         \multicolumn{1}{c}{\bf (d) Apply Unzip macro} & 
         \multicolumn{1}{c}{\bf (e) Apply Transpose on linear part} & 
         \multicolumn{1}{c}{\bf (f) Simplify and initialise reverse-mode} \\\hline
    \hspace{-1mm}\begin{tabular}{l}
     $d\loss.nonlin = \lambda (\otimes \underline{\theta};):\RR^2.\, \haskdo~\{$\\
    $\quad \llet\ (\epsilon;) \gets \normalreparam(0,1)$\\
    $\quad \llet\ (x_1;) = \epsilon * 1$ \\
    $\quad \llet\ (x;) = x_1 + \underline{\theta}_1$\\
    $\quad \llet\ (y;) = \sin(x)$\\
    $\quad\llet\ (dy;) = \cos(x)$\\
    $\quad \llet\ (z;) = y+\underline{\theta}_2$ \\
    $\quad \llet\ (\underline{trace};) =\otimes(\epsilon, x_1,x,y,dy,z)$ \\
    $\quad\quad \return (z,trace)\big)$\\ \\
     $d\loss.lin = \lambda (\otimes\underline{trace};\otimes \underline{\dot{\theta}}):\otimes \RR\times \RR^2.\,$\\
    $\quad \llet\ (\epsilon, x_1,x,y,dy,z; ) = \underline{trace}$\\
    $\quad \llet\ (;\dot{x_1}) = \epsilon * 0$ \\
    $\quad \llet\ (;\dot{x}) = \dot{x_1} + \underline{\dot{\theta}}_1$\\
    $\quad \llet\ (;\dot{y}) = dy*\dot{x}$\\
    $\quad \llet\ (;\dot{z}) = \dot{y}+\underline{\dot{\theta}}_2$ \\
    $\quad\quad \return\ \dot{z}$
    \end{tabular}
    & 
    \hspace{-1mm}\begin{tabular}{l}
     $d\loss.nonlin = \lambda (\otimes \underline{\theta};):\RR^2.\, \haskdo~\{$\\
    $\quad \llet\ (\epsilon;) \gets \normalreparam(0,1)$\\
    $\quad \llet\ (x_1;) = \epsilon * 1$ \\
    $\quad \llet\ (x;) = x_1 + \underline{\theta}_1$\\
    $\quad \llet\ (y;) = \sin(x)$\\
    $\quad\llet\ (dy;) = \cos(x)$\\
    $\quad \llet\ (z;) = y+\underline{\theta}_2$ \\
    $\quad \llet\ (\underline{trace};) =\otimes(dy)$ \\
    $\quad\quad \return (z,\underline{trace})\big)$\\ \\
     $d\loss.lin = \lambda (\otimes\underline{trace}; \dot{z}):\otimes \RR\times \RR.\,$\\
    $\quad \llet\ (dy; ) = \underline{trace}$\\
        $\quad \llet\ (;\dot{y}) = \dot{z}$ \\
    $\quad \llet\ (;\underline{\dot{\theta}}_2) = \dot{z}$ \\
    $\quad \llet\ (;\dot{x}) = dy*\dot{y}$\\
    $\quad \llet\ (;\dot{x_1}) = \dot{x}$ \\ 
 $\quad \llet\ (;\underline{\dot{\theta}}_1) = \dot{x}$ \\ 
    $\quad\quad \return\ \underline{\dot{\theta}}$
    \end{tabular}
    & 
    \hspace{-1mm}\begin{tabular}{l}
    $d\loss.nonlin = \lambda (\otimes \underline{\theta};):\RR^2.\, \haskdo~\{$\\
    $\quad \llet\ (\epsilon;) \gets \normalreparam(0,1)$\\
    $\quad \llet\ (x;) = \epsilon + \underline{\theta}_1$\\
    $\quad \llet\ (y;) = \sin(x)$\\
    $\quad\llet\ (dy;) = \cos(x)$\\
    $\quad \llet\ (z;) = y+\underline{\theta}_2$ \\
    $\quad \llet\ (\underline{trace};) =\otimes(dy)$ \\
    $\quad\quad \return (z,\underline{trace})\big)$\\ \\
     $d\loss.grad = \lambda (\otimes\underline{trace};):\otimes \RR.\,$\\
    $\quad \llet\ (dy; ) = \underline{trace}$\\
    $\quad \llet\ (;\underline{\dot{\theta}}_2) = 1$ \\
 $\quad \llet\ (;\underline{\dot{\theta}}_1) =  dy$ \\
    $\quad\quad \return\ \underline{\dot{\theta}}$ 
    \\
     \\
     \\
     \\
    \end{tabular}\\\hline
    \end{tabular}
}
\vspace{-.3cm}
\caption{Deriving a reverse mode gradient estimator for a program in $\lambda_\text{ADEV}$ via YOLO \cite{radul2022you}. We've dropped $\iin$ for readability.}
\label{fig:yolo}
\end{figure*}

\newpage
\section{Additional Variational Objectives in \texorpdfstring{$\lambda_{\text{ADEV}}$}{Lambda ADEV}}
\label{appdx:objectives}
IWELBO, IWHVI, and DIWHVI are all objectives which are based on \textit{evidence lower bounds}. Many variational algorithms also make use of objectives motivated by the forward KL divergence
\begin{align*}
\mathcal{F}_{\text{KL}(P, Q)} := \mathbb{E}_{z \sim P}[\log \frac{\textbf{density}\{P\}(z)}{\textbf{density}\{Q\}(z)}]
\end{align*}
where, for instance, $Q$ has learnable parameters and $P$ is taken to be a posterior. $P$ may be intractable --- we may not have access to simulators or density evaluators for $P$. In this setting, several variational algorithms \cite{hinton_wake-sleep_1995,bornschein_reweighted_2015, gu_neural_2015} consider a modified objective \textit{where inference is used to approximate $P$} --- these algorithms utilize $Q$ as part of the $\text{Alg}$ approximation:
\begin{align*}
\mathcal{F}_{\text{KL}(\text{Alg}(Q), Q)} := \mathbb{E}_{z \sim \text{Alg}(Q)}[\log \frac{\textbf{density}\{\text{Alg}(Q)\}(z)}{\textbf{density}\{Q\}(z)}]
\end{align*}
Of course, for a fixed $\text{Alg}$, the estimators constructed from the modified objective will be biased estimators of the original forward KL objective --- where the bias is controlled by the accuracy of the approximation of $P$ induced by $\text{Alg}$. Importantly, our language includes the possibility of using these objectives in variational algorithms: for instance, the wake phase $Q$-update gradient estimator from \cite{bornschein_reweighted_2015} can be derived in $\lambda_\text{ADEV}$ (plus our full system extensions for estimated densities) from the objective:

\vspace{0.1in}
\begin{mdframed}
{\small 
\begin{tabular}{l}
$\textbf{QWake}(N, P, Q, \phi^\prime)\coloneqq \lambda\phi. \mathbb{E}(\haskdo~\{$ \\
$\quad\quad (z, \_)\gets\textbf{sim}\{\lambda\phi'.\textbf{normalize} \ P \ (\textbf{importance} \ (Q \ \phi') \ N) \}~\phi^\prime;$ \\ 
$\quad\quad q\gets(\textbf{density}\{Q\}~\phi)~z;$ \\ 
$\quad\quad \return - log~q$ \\
$\})$
\end{tabular}}
\end{mdframed}
\vspace{0.1in}

\noindent and also, the wake-phase $P$-update:

\vspace{0.1in}
\begin{mdframed}
{\small 
\begin{tabular}{l}
$\textbf{PWake}(N, P, Q, \phi')\coloneqq \lambda\theta. \mathbb{E}(\haskdo~\{$ \\
$\quad\quad (z, w_q)\gets\textbf{sim}\{\lambda\theta .\textbf{normalize} \ (P \ \theta) \ (\textbf{importance} \ (Q \ \phi^\prime) \ N) \}~\theta;$ \\ 
$\quad\quad p\gets(\textbf{density}\{P\}~\theta) ~z;$ \\ 
$\quad\quad \return~(log~p) - (log~w_q)$ \\
$\})$
\end{tabular}}
\end{mdframed}
\vspace{0.1in}

\noindent where, in this objective, $P$ is a term of type $\mathbb{R}^M \rightarrow D \ \tau$. Indeed, this allows one to implement learning architectures like \textit{reweighted wake-sleep}~\cite{bornschein_reweighted_2015} in $\lambda_\text{ADEV}$.

\newpage
\section{Evaluation: Hand Coded Program vs. Our Implementation}
\label{appdx:overhead}
In \S\ref{sec:evaluation} and Table~\ref{tab:overhead}, we provided evaluation time comparisons between gradient estimators for the ELBO of a variational autoencoder derived using our system and a manually coded implementation. Below, we illustrate the programs used to construct this comparison, and runtime curves as a function of batch size.
\vspace{0.1in}

\begin{figure}[htb!]
\centering
% Row 1: Code on the left, Image on the right
\hfill
\begin{minipage}{0.5\textwidth}
\textbf{Hand coded gradient estimator}
\begin{mdframed}[innerleftmargin=13pt, innerrightmargin=15pt]
\input{fig/overhead/handcoded}
\end{mdframed}
\end{minipage}
\hfill
\begin{minipage}{0.48\textwidth}
\includegraphics[width=\linewidth]{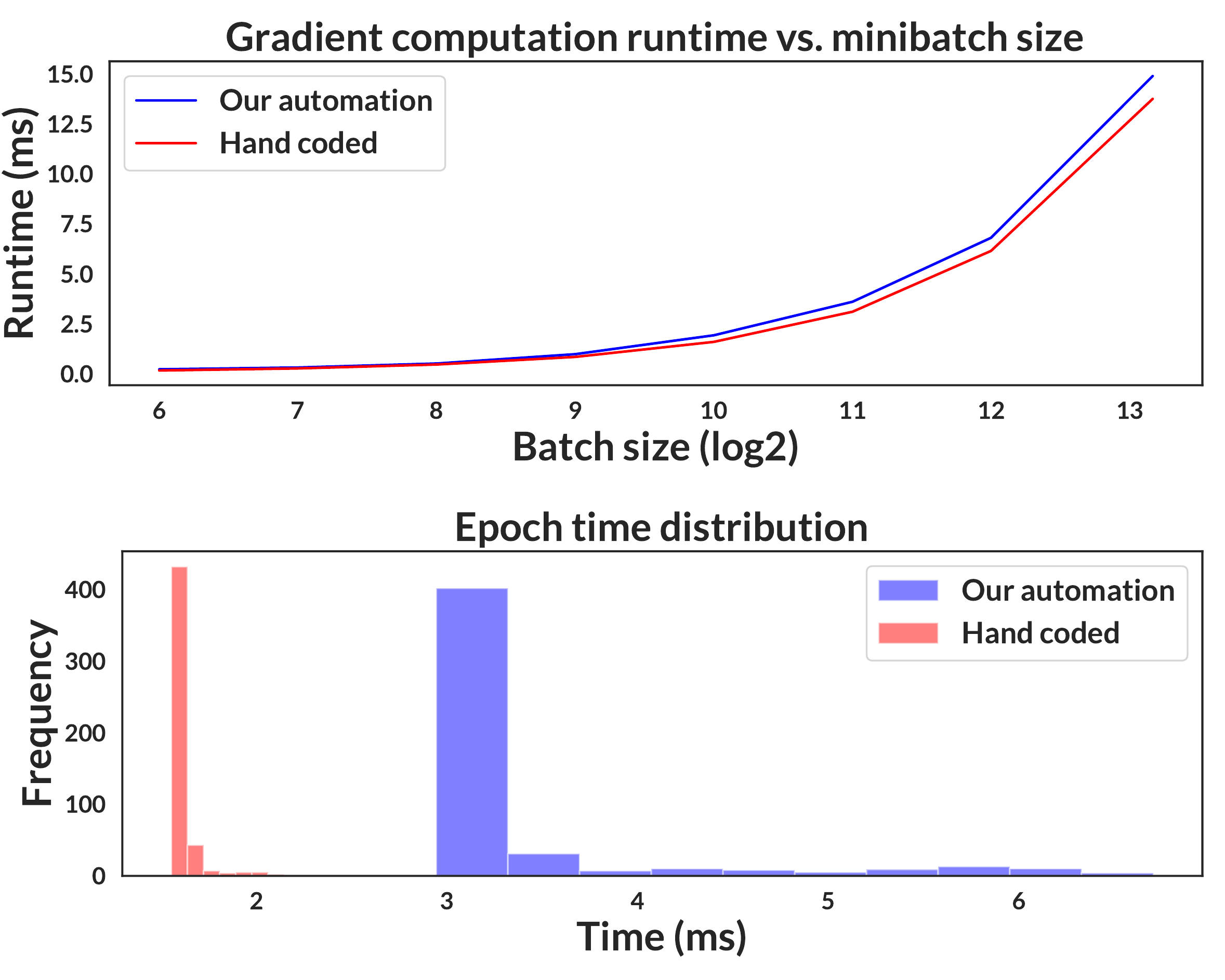}
\end{minipage}
\vspace{0.2cm} % Adds vertical space after the line
\hrule
\vspace{0.2cm} % Adds vertical space after the line
% Row 2: Code on the left, Image on the right
\hfill
\begin{minipage}{0.48\textwidth}
\textbf{Model program}
\begin{mdframed}
\input{fig/overhead/model}
\end{mdframed}
\vspace{0.1cm}
\textbf{Guide program}
\begin{mdframed}
\input{fig/overhead/guide}
\end{mdframed}
\end{minipage}
\hfill
\hfill
\hfill
\begin{minipage}{0.5\textwidth}
\textbf{Loss program (with automated estimator)}
\begin{mdframed}[innerleftmargin=13pt]
\input{fig/overhead/elbo}
\end{mdframed}
\end{minipage}
\caption{Comparing a hand coded VAE ELBO gradient estimator (top left) to our implementation. (Bottom) Our language and implementation supports a modular separation of gradient strategies, model and guide programs, and loss specification. (Top right) We investigate the runtime cost of our abstractions in a parallelization experiment (increasing batch size), and find a negligible impact on runtime performance.}
\label{fig:overhead_transcript}
\end{figure}

\newpage
\section{Evaluation: Additional Experiments}
\label{appdx:additional-experiments}
We performed additional experiments comparing $\texttt{genjax.vi}$ to Pyro and NumPyro, on tutorial examples drawn from Pyro's website. These examples include applying variational inference to a coin flip model, Bayesian linear regression, and semi-supervised and conditional variational autoencoders. In this section, we provide a discussion of the experiments and a selection of our results. For all experiments, we use the standard definition of ELBO in $\lambda_\text{ADEV}$:
{\small $$\textbf{ELBO}(P, Q) \coloneqq \lambda(\theta,\phi). \mathbb{E}(\haskdo~\{(x, w_q)\gets\textbf{sim}\{Q\} ~\phi; w_p\gets\textbf{density}\{P\}~\theta~x; \return\log w_p~-~\log w_q\})$$}

\subsection{Fairness inference of a noisy coin}
We consider a model over \href{https://pyro.ai/examples/svi_part_i.html#A-simple-example}{the latent fairness of a coin}. We observe a sequence of coin flips of length $10$, and seek to infer the distribution over the fairness (the probability) of the coin flip, a Bernoulli random variable. The model and guide programs can be expressed in $\lambda_\text{Gen}$:

\begin{figure}[htb]
\begin{mdframed}
\begin{minipage}[t]{.49\linewidth}
\vspace{0pt}
{\small\begin{tabular}{@{}l@{}}
$\textbf{model} \coloneqq \lambda~flips. (\haskdo~\{$ \\
$\quad\quad p\gets\textbf{sample}~\textbf{beta}(10.0, 10.0)~\text{"fairness"};$ \\
$\quad\quad  \textbf{mapM}~(\lambda x.~\textbf{observe}~\textbf{bernoulli}(p)~x)~flips$ \\
$\})$
\end{tabular}}
\end{minipage}
\hfill
\begin{minipage}[t]{.49\linewidth}
\vspace{0pt}
{\small\begin{tabular}{@{}l@{}}
$\textbf{guide} \coloneqq \lambda\phi. (\haskdo~\{$ \\
$\quad\quad \alpha = \pi_1~\phi;$ \\
$\quad\quad \beta = \pi_2~\phi;$ \\
$\quad\quad \_\gets\sample~\textbf{beta}_\text{IMPLICIT}(\alpha, \beta)$~\text{"fairness"} \\
$\})$
\end{tabular}}
\end{minipage}
\end{mdframed}
\end{figure}

\noindent We train the guide using gradient estimators for the ELBO. We've summarized the results across the 3 systems in the table below, comparing the average time per epoch, the average of ELBO estimates on the last 100 training steps, and, taking the learned parameter values of the guide, the inferred posterior mean (given by $\mu = \frac{\alpha}{\alpha + \beta}$, where $\alpha$ and $\beta$ are the parameters to the beta sampler (above, $\pi_1~\phi$ and $\pi_2~\phi$)).

\begin{figure}[htb!]
\centering
\begin{tabular}{|l|l|l|l|}
\hline
System & \makecell{Average wall time \\ per epoch} & \makecell{Average of ELBO estimates \\ (last $100$ steps)} & Inferred posterior mean \\
\hline
\texttt{genjax.vi} & 0.97 (ms) & -7.06 & $0.529\pm0.09$ \\
NumPyro & 0.66 (ms) & -7.07 & $0.531\pm0.09$ \\
Pyro & 0.77 (ms) & -7.07 & $0.532\pm0.09$ \\
\hline
\end{tabular}
\end{figure}

\subsection{Bayesian linear regression}
We consider the linear regression model from \href{https://pyro.ai/examples/intro_long.html#Example:-mean-field-variational-approximation-for-Bayesian-linear-regression-in-Pyro}{Introduction to Pyro}: the goal is to fit a regression model for the GDP of a country given a Boolean value (whether or not the country is in Africa) and a $\mathbb{R}$ value (the Terrain Ruggedness Index measurement of the country). The model and guide programs can be expressed in $\lambda_\text{Gen}$:

\begin{figure}[htb]
\begin{mdframed}
\begin{minipage}[t]{.44\linewidth}
\vspace{0pt}
{\small\begin{tabular}{@{}l@{}}
$\textbf{model} \coloneqq \lambda(gdps,cafrica,ruggedness). (\haskdo~\{$ \\
$\quad\quad a\gets\textbf{sample}~\textbf{normal}(0.0,10.0)~\text{"a"};$ \\
$\quad\quad ba\gets\textbf{sample}~\textbf{normal}(0.0,1.0)~\text{"bA"};$ \\
$\quad\quad br\gets\textbf{sample}~\textbf{normal}(0.0,1.0)~\text{"bR"};$ \\
$\quad\quad bar\gets\textbf{sample}~\textbf{normal}(0.0,1.0)~\text{"bAR"};$ \\
$\quad\quad s\gets\textbf{sample}~\textbf{uniform}(0.0,10.0)~\text{"sigma"};$ \\
$\quad\quad m =(a + ba \times cafrica~+ $ \\
$\quad\quad\quad\quad\quad\quad br \times ruggedness~+$ \\ 
$\quad\quad\quad\quad\quad\quad bar \times cafrica \times ruggedness);$ \\
$\quad\quad \textbf{mapM}~(\lambda x.~\textbf{observe}~\textbf{normal}(m, s)~x)~gdps$ \\
$\})$
\end{tabular}}
\end{minipage}
\hfill
\begin{minipage}[t]{.525\linewidth}
\vspace{0pt}
{\small\begin{tabular}{@{}l@{}}
$\textbf{guide} \coloneqq \lambda(\phi, cafrica, ruggedness). (\haskdo~\{$ \\
$\quad\quad (al, as, wl0, ws0, wl0, ws1, wl1, ws2, wl2, sl) = \phi;$ \\
$\quad\quad a\gets\textbf{sample}~\textbf{normal}_\text{REPARAM}(al,as)~\text{"a"};$ \\
$\quad\quad ba\gets\textbf{sample}~\textbf{normal}_\text{REPARAM}(wl0,ws0)~\text{"bA"};$ \\
$\quad\quad br\gets\textbf{sample}~\textbf{normal}_\text{REPARAM}(wl1,ws1)~\text{"bR"};$ \\
$\quad\quad bar\gets\textbf{sample}~\textbf{normal}_\text{REPARAM}(wl2,ws2)~\text{"bAR"};$ \\
$\quad\quad s\gets\textbf{sample}~\textbf{normal}(sl,0.05)~\text{"sigma"};$ \\
$\})$
\end{tabular}}
\end{minipage}
\end{mdframed}
\end{figure}

\noindent We train the guide using $\textbf{ELBO}(\textbf{model}, \textbf{guide})$. In Fig.~\ref{fig:lin_reg_table}, we show the wall clock epoch times and final ELBO values for comparison across training using \texttt{genjax.vi}, Pyro, and NumPyro.

\begin{figure}[htb!]
\centering
\vspace{0.02in}
\begin{tabular}{|l|l|l|}
\hline
System & \makecell{Average wall time per epoch} & \makecell{Average of ELBO estimates (last $100$ steps)} \\
\hline
\texttt{genjax.vi} & 0.18 (ms) & -1.47 \\
NumPyro & 0.09 (ms) & -1.47 \\
Pyro & 3.73 (ms) & -1.46 \\
\hline
\end{tabular}
\caption{Wall clock times and average ELBO over last 100 training steps for Bayesian linear regression.}
\label{fig:lin_reg_table}
\end{figure}

\noindent The estimated mean and 90\% credible interval (CI) for pushforward of the learned guide by the regression function (the function used to compute $m$) in \texttt{genjax.vi} and Pyro is shown in Fig.~\ref{fig:bayes_lin_reg}. 

\begin{figure}[htb]
\centering
\hfill
\begin{minipage}{0.49\textwidth}
\includegraphics[width=1.0\linewidth]{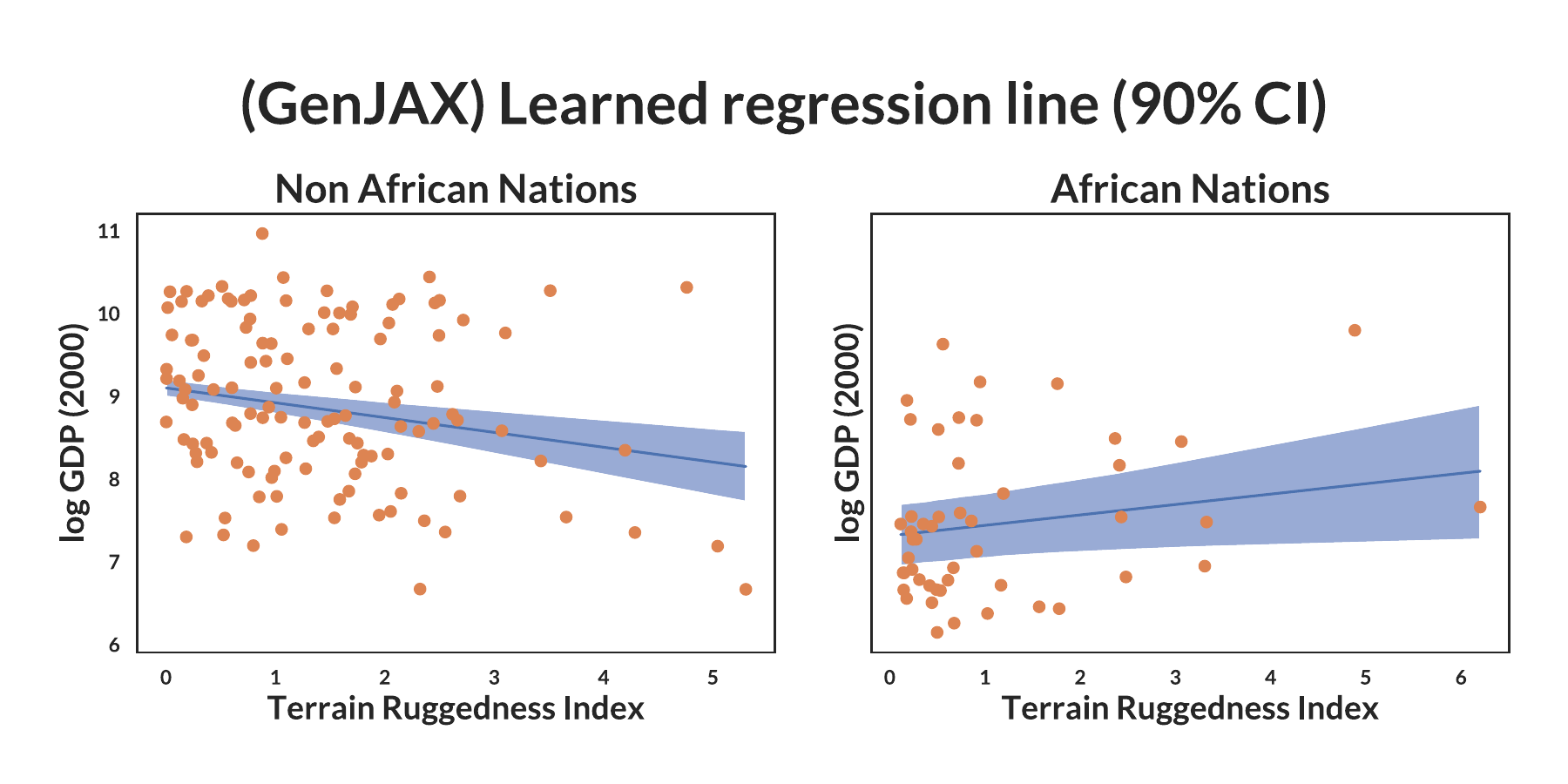}
\end{minipage}
\begin{minipage}{0.49\textwidth}
\includegraphics[width=1.0\linewidth]{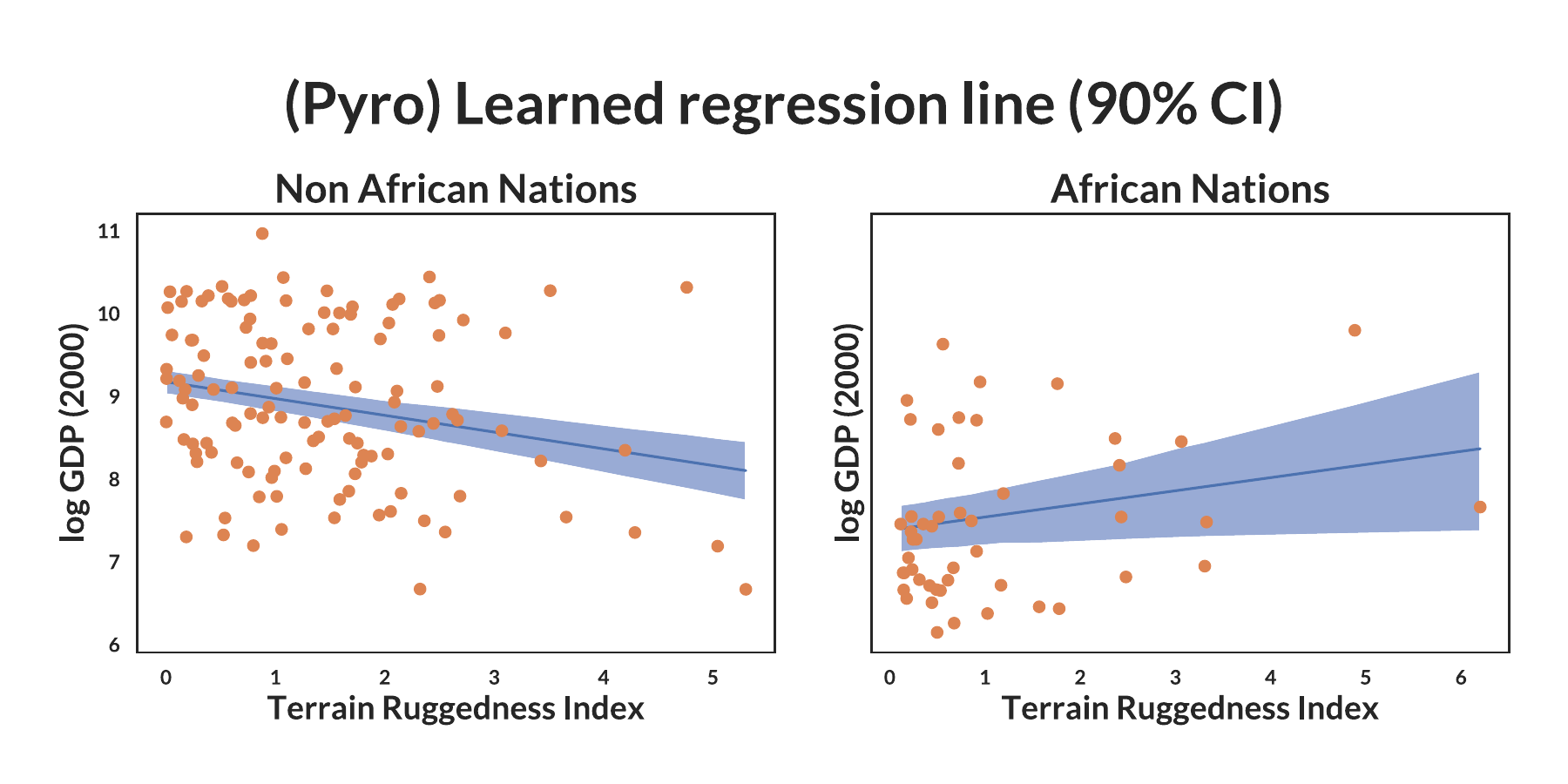}
\end{minipage}
\caption{Final regression fits from variational inference. The variable \texttt{cafrica} denotes whether a country is in Africa. The variable \texttt{ruggedness} denotes the value of the country on the Terrain Ruggedness Index. Shaded blue is the 90\% credible interval (CI) approximated by sampling (3200 samples).}
\label{fig:bayes_lin_reg}
\end{figure}

\subsection{Semi-supervised variational autoencoder}
We consider \href{https://pyro.ai/examples/ss-vae.html}{the semi-supervised variational autoencoder}~\cite{kingma_semi-supervised_2014}: a variational algorithm for settings where both unlabeled and labeled data are available. There are two models for this algorithm in $\lambda_\text{Gen}$, one for unsupervised training and one for supervised training:

\begin{figure}[htb]
\begin{mdframed}
{\small\begin{tabular}{@{}l@{}}
$\textbf{unsup\_model} \coloneqq \lambda(img,\theta).(\textbf{do}\{$ \\
$\quad\quad label\gets\textbf{sample}~\textbf{categorical}(\textbf{map}~(\lambda v.\frac{v}{10})~(\textbf{ones}~10))~\text{"label"};$ \\
$\quad\quad label_\text{oh} = \textbf{one\_hot}~label~10;$ \\
$\quad\quad latent\gets\textbf{sample}~\textbf{mv\_normal}((\textbf{zeros}~50),(\textbf{ones}~50))~\text{"latent"};$ \\
$\quad\quad probs = \textbf{decoder}~\theta~latent~label_\text{oh};$ \\
$\quad\quad \textbf{mapM}~(\lambda (p, px).~\textbf{observe}~(\textbf{bernoulli}(p)~px))~\mathbf{zip}(probs,img)$ \\
$\})$
\end{tabular}}
\end{mdframed}
\vspace{0.1in}
\begin{mdframed}
{\small\begin{tabular}{@{}l@{}}
$\textbf{sup\_model} \coloneqq \lambda(label,img,\theta).(\textbf{do}\{$ \\
$\quad\quad \textbf{observe}~\textbf{categorical}(\textbf{map}~(\lambda v.\frac{v}{10})~(\textbf{ones}~10))~label;$ \\
$\quad\quad label_\text{oh} = \textbf{one\_hot}~label~10;$ \\
$\quad\quad latent\gets\textbf{sample}~\textbf{mv\_normal}((\textbf{zeros}~50),(\textbf{ones}~50))~\text{"latent"};$ \\
$\quad\quad probs = \textbf{decoder}~\theta~latent~label_\text{oh};$ \\
$\quad\quad \textbf{mapM}~(\lambda (p, px).~\textbf{observe}~(\textbf{bernoulli}(p)~px))~\mathbf{zip}(probs,img)$ \\
$\})$
\end{tabular}}
\end{mdframed}
\label{fig:ss_models}
\caption{Models for semi-supervised learning. We make use of several helper functions: $\lambda v.\textbf{one\_hot}~v~N$ converts $\text{label}\in [0, N] \subset\mathbb{N}$ into a one-hot vector in $\mathbb{R}^N$. \textbf{decoder} represents a neural network: it accepts parameters $\theta$ and applies a network to the input $\texttt{latent}$ and $\texttt{label\_oh}$ concatenated into a single vector.}
\end{figure}

\noindent Corresponding to each model, there is also a guide:

\begin{figure}[htb]
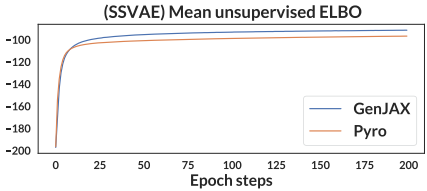

\begin{mdframed}
{\small\begin{tabular}{@{}l@{}}
$\textbf{unsup\_guide} \coloneqq \lambda(img, \phi).(\textbf{do}\{$ \\
$\quad\quad probs = \textbf{encoder}_1~(\pi_1~\phi)~img$ \\
$\quad\quad label\gets\textbf{sample}~\textbf{categorical}_\text{ENUM}(probs)~\text{"label"};$ \\
$\quad\quad label_\text{oh} = \textbf{one\_hot}~label~10;$ \\
$\quad\quad \phi_\text{latent} = \textbf{encoder}_2~(\pi_2~\phi)~label_\text{oh}~img;$ \\
$\quad\quad latent\gets\textbf{sample}~\textbf{mv\_normal}_\text{REPARAM}(\pi_1~\phi_\text{latent},\pi_2~\phi_\text{latent})~\text{"latent"};$ \\
$\})$
\end{tabular}}
\end{mdframed}
\vspace{0.1in}
\begin{mdframed}
{\small\begin{tabular}{@{}l@{}}
$\textbf{sup\_guide} \coloneqq \lambda(label, img, \phi).(\textbf{do}\{$ \\
$\quad\quad label_\text{oh} = \textbf{one\_hot}~label~10;$ \\
$\quad\quad \phi_\text{latent} = \textbf{encoder}_2~(\pi_2~\phi)~label_\text{oh}~img;$ \\
$\quad\quad latent\gets\textbf{sample}~\textbf{mv\_normal}_\text{REPARAM}(\pi_1~\phi_\text{latent},\pi_2~\phi_\text{latent})~\text{"latent"};$ \\
$\})$
\end{tabular}}
\end{mdframed}
\label{fig:ss_guides}
\caption{Guides for semi-supervised learning.}
\end{figure}

\noindent The objective for supervised training is $\textbf{ELBO}(\textbf{sup\_model}, \textbf{sup\_guide})$ and the objective for unsupervised training is
$\textbf{ELBO}(\textbf{unsup\_model},\textbf{unsup\_guide})$. In the semi-supervised setting, the assumption is that a small amount of labeled data is available: the idea is that mixing unsupervised learning with supervised learning will allow the learning process to better use the labeled data.

\vspace{0.05in}
\noindent In Fig.~\ref{fig:ssvae_loss}, we show the unsupervised loss (ELBO) curves for \texttt{genjax.vi} and Pyro. Both systems use the same hyperparameters, similar network architectures (\texttt{genjax.vi} uses SoftPlus activation instead of ReLU), and the same unsupervised / supervised training schedule (3000 supervised examples, with a round of supervised training every 16th batch). In Fig.~\ref{fig:ssvae}, we show conditional generation from the learned SSVAE networks.

\begin{figure}[htb!]
\begin{minipage}{0.43\textwidth}
\includegraphics[width=1.0\linewidth]{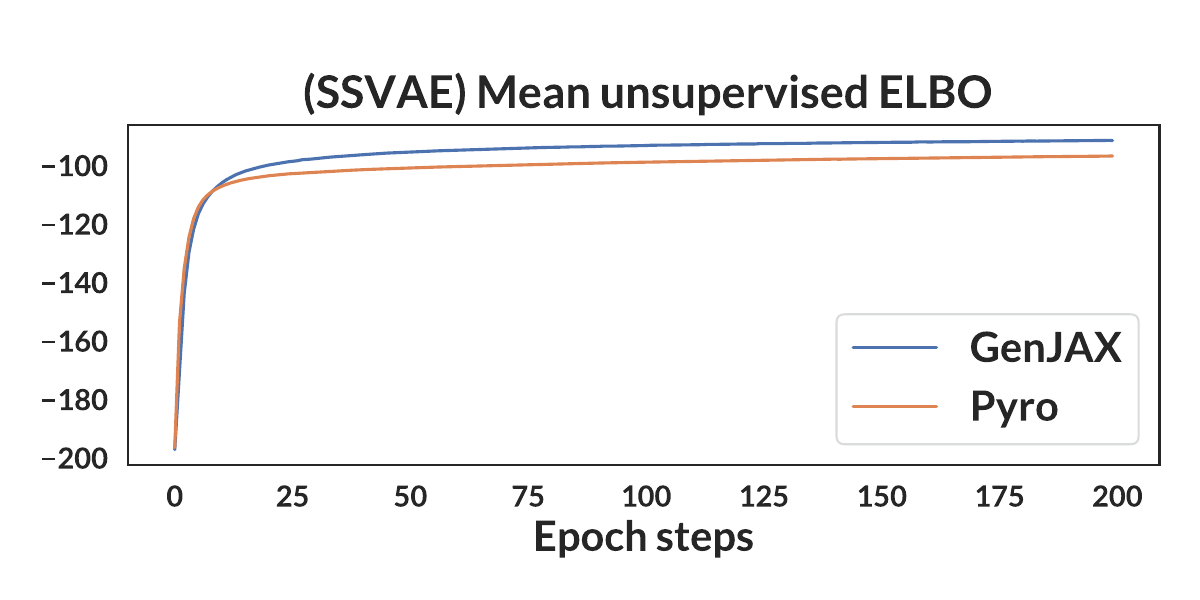}
\end{minipage}
\begin{minipage}{0.52\textwidth}
\centering
{\small\begin{tabular}{|l|l|l|}
\hline
System & \makecell{Average wall time \\ per epoch} & \makecell{Average ELBO \\ (last $100$ steps)} \\
\hline
\texttt{genjax.vi} & 0.41 (s) & -92.48 \\
Pyro & 3.98 (s) & -97.74 \\
\hline
\end{tabular}}
\end{minipage}
\caption{The mean (over batches) unsupervised loss (an ELBO objective) during training. \texttt{genjax.vi} and Pyro exhibit similar loss curves using the same gradient estimators for the ELBO.}
\label{fig:ssvae_loss}
\end{figure}

\begin{figure}[htb!]
\centering
\hfill
\begin{minipage}{0.48\textwidth}
\textbf{\texttt{genjax.vi}}
\vspace{0.02in}
\begin{mdframed}
\includegraphics[width=1.0\linewidth]{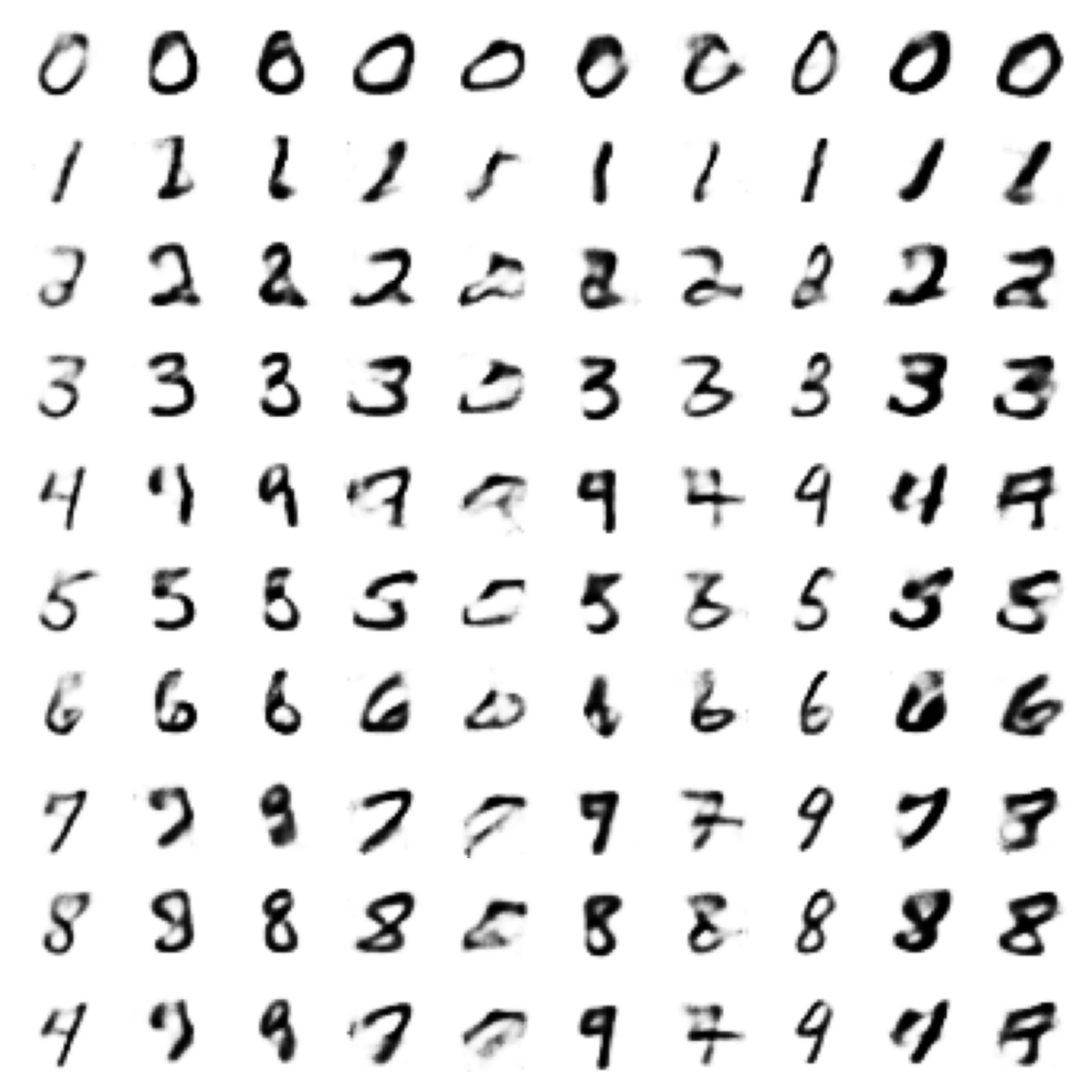}
\end{mdframed}
\end{minipage}
\begin{minipage}{0.48\textwidth}
\textbf{Pyro}
\vspace{0.02in}
\begin{mdframed}
\includegraphics[width=1.0\linewidth]{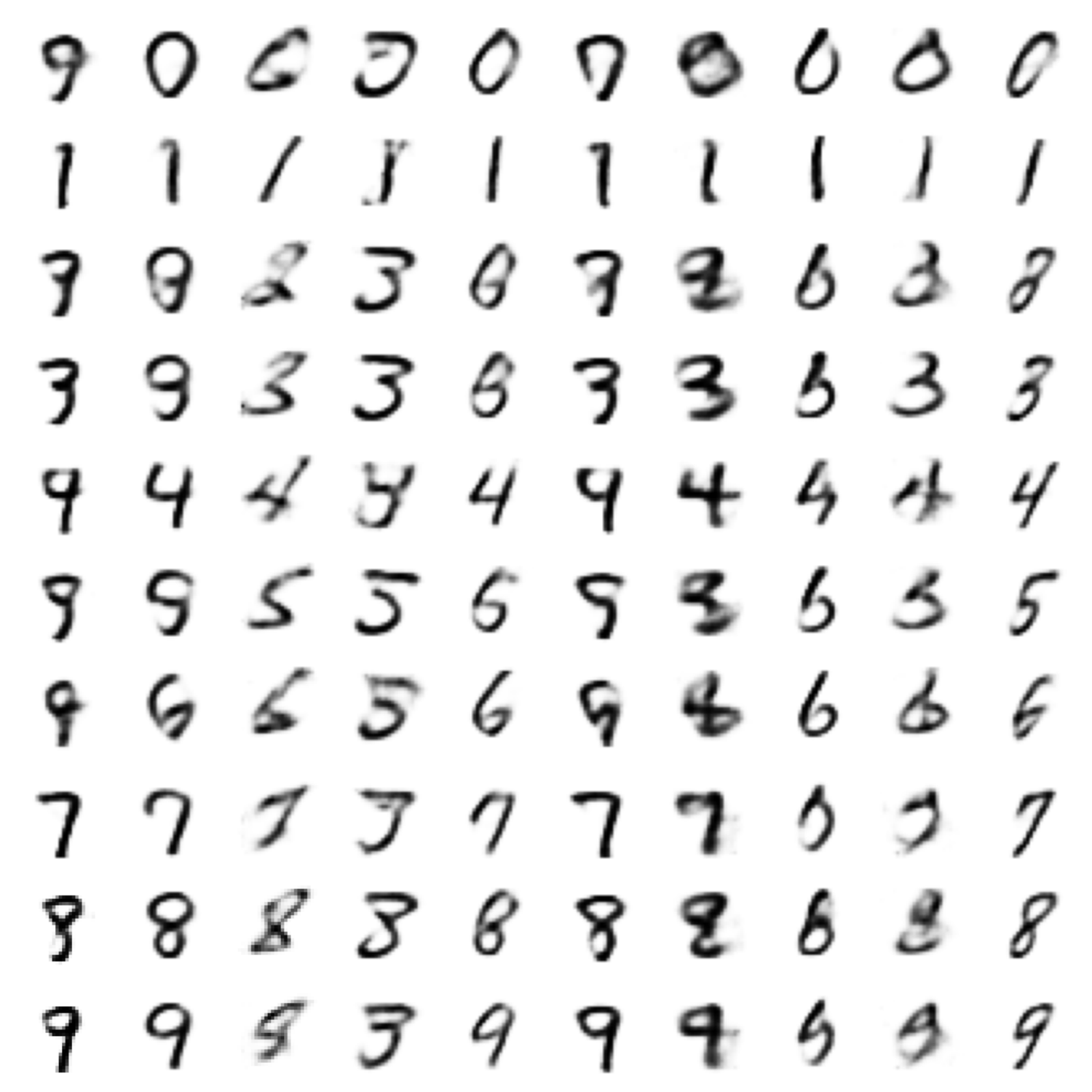}
\end{mdframed}
\end{minipage}
\caption{Conditional generation using trained SSVAE from (left) \texttt{genjax.vi}, and (right) Pyro. The resultant decoder models have similar visual fidelity. Samples were generated by fixing a random seed \textit{per column} - each sampled image in a column has the same \textit{style} (latent sample from multivariate Gaussian).}
\label{fig:ssvae}
\end{figure}

\subsection{Conditional variational autoencoder}

We consider \href{https://pyro.ai/examples/cvae.html}{the conditional variational autoencoder}~\cite{sohn_learning_2015}. The learning task: given an observation of a quadrant of an MNIST image, fill in the other quadrants. The CVAE architecture consists of two components:
\begin{itemize}
    \item A deterministic baseline network which looks at the observed quadrant and attempts to fill in the other quadrants. This network is trained using binary cross entropy loss applied pixel wise across the quadrants which are not conditioned on.
    \item A variational autoencoder, which is trained on the entire image. The architecture uses the trained baseline network as part of the model, to provide an initial guess for the rest of the image, which is then processed by the decoder networks.
\end{itemize}

\begin{figure}[htb!]
\centering
\begin{minipage}{0.48\textwidth}
\textbf{\texttt{genjax.vi}}
\begin{mdframed}
\includegraphics[width=1.0\linewidth]{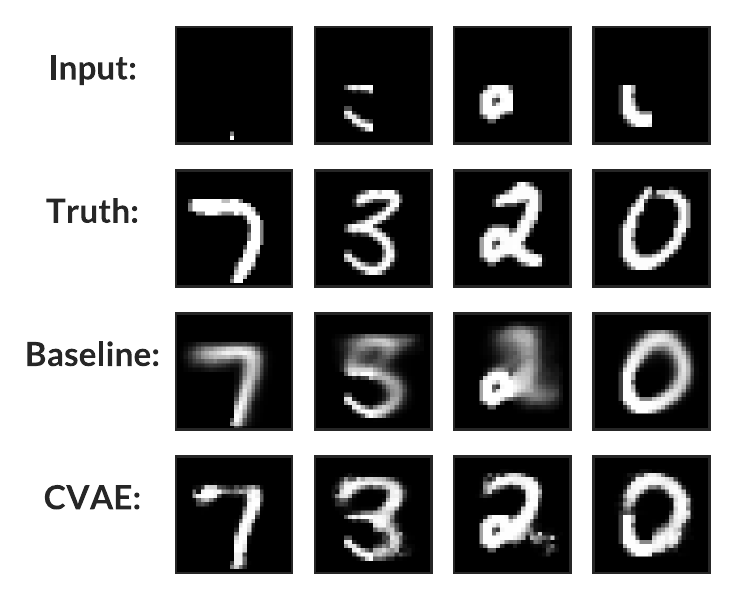}
\end{mdframed}
\end{minipage}
\begin{minipage}{0.48\textwidth}
\textbf{NumPyro}
\begin{mdframed}
\includegraphics[width=1.0\linewidth]{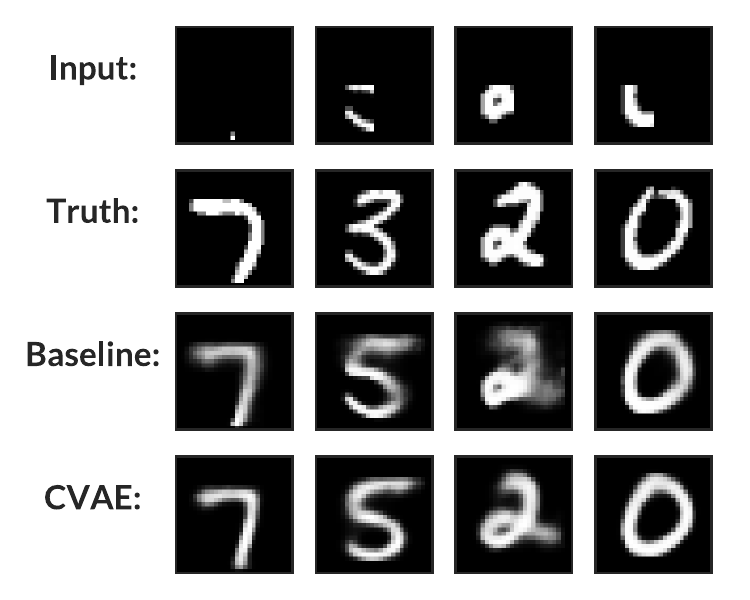}
\end{mdframed}
\end{minipage}
\caption{Conditional generation using the CVAE networks learned by \texttt{genjax.vi} and NumPyro. The \textit{input} is the lower left quadrant of an MNIST digit image: the CVAE attempts to fill in the other quadrants. Both \texttt{genjax.vi} and NumPyro learn approximations with similar reconstruction accuracy.}
\label{fig:cvae}
\end{figure}
\begin{figure}[htb!]
{\small\begin{tabular}{|l|l|}%l|}
\hline
System & \makecell{Average wall time per epoch} \\%& \makecell{Last ELBO estimate} \\
\hline
\texttt{genjax.vi} & 0.36 (s) \\%& -77.13 \\
Pyro & 5.73 (s) \\%& -72.76 \\
NumPyro & 1.36 (s) \\%& \textcolor{purple}{\textbf{$\times$}} \\
\hline
\end{tabular}}
\caption{Wall clock time comparisons for the conditional VAE.} %and final ELBO estimate comparisons.} %We omit NumPyro's final ELBO result: the value was not comparable (we were unable to understand the scaling applied by e.g. the batch dimension to report values similar to the ones reported by Pyro and \texttt{genjax.vi}).}
\label{fig:cvae_table}
\end{figure}

\newpage
\section{Proof Details}
\label{appx:proof-details}

This section provides details on the proofs of several results stated in the main paper.

\subsection{Definition of helper functions \textit{pop} and \textit{split}}

Our semantics is defined in part using helper functions, $\textit{pop}$ and $\textit{split}_\mu$, that we define here.\\
\vspace{-1mm}

\noindent\textbf{Definition of $\textit{pop}$.} For each base type $\sigma$, we define $\textit{pop}_\sigma : \mathbb{T} \times \str \to \sigma \times \RR \times \mathbb{T}$ to return, on input $(u, k)$, the tuple $(x, 1, u')$ when $u = u' \mdoubleplus \{k \mapsto x\}$ for some $x \in \sem{\sigma}$, and the tuple $(\text{default}_\sigma, 0, \{\})$ when $u$ does not map the key $k$ to a value of type $\sigma$. When clear from context, we omit the type $\sigma$.\\\vspace{-1mm}

\noindent\textbf{Definition of $\textit{split}$.} Suppose $\mu$ is a discrete-structured measure on $\mathbb{T}$, absolutely continuous with respect to $\mathcal{B}_\mathbb{T}$. Define $\mathbb{T}_\mu \subseteq \mathbb{T}$ to be the support of $\mu$. For a given trace $u \in \mathbb{T}$, define $\text{subtraces}(u) := \{u' \in \mathbb{T} \mid \exists u'' \in \mathbb{T}. u' \mdoubleplus u'' = u\}$ to be the set of $u$'s subtraces. Note that for all $u \in \mathbb{T}$, the intersection $\mathbb{T}_\mu \cap \text{subtraces}(u)$ is either empty, or a singleton set $\{u^*\}$. This is because, if there were two distinct traces $u_1^*$ and $u_2^*$ in the intersection, then by the discrete-structured condition, they would need to disagree on the value of some choice $k$ present in both $u_1^*$ and $u_2^*$. But this is not possible, since both traces are subtraces of $u$, and so agree with $u$ (and therefore with each other) for every choice $k$ that they share in common.

The function $\textit{split}_\mu(u)$ is defined to equal $(u, \{\})$ when $\mathbb{T}_\mu \cap \text{subtraces}(u)$ is empty, and $(u^*, u - u^*)$ when $\mathbb{T}_\mu \cap \text{subtraces}(u)$ is non-empty. (Here, $u - u^*$ is the subtrace of $u$ that has only keys $k$ which do not appear in $u^*$.)

\subsection{Well-definedness of the semantics}

Our semantics of $\lambda_{\text{Gen}}$ (Fig.~\ref{fig:syntax}) interprets programs of type $G~\tau$ as elements of $\text{Meas}_{\ll \mathcal{B}_\mathbb{T}}^{DS}~\mathbb{T}$, a subspace of particularly well-behaved measures on $\mathbb{T}$. We must therefore prove that our term semantics actually does assign to every term of type $G~\tau$ a well-behaved measure over traces.

\begin{lemma}
Let $\Gamma \vdash t : \tau$ be a $\lambda_\text{Gen}$ term, and let $\gamma \in \sem{\Gamma}$. Then $\sem{t}(\gamma)$, as defined inductively in Figure~\ref{fig:syntax}, lies in $\sem{\tau}$. In particular, if $\tau = G~\tau'$ for some type $\tau'$, then $\sem{t}_1(\gamma)$ is discrete-structured and absolutely continuous with respect to $\mathcal{B}_\mathbb{T}$.
\end{lemma}

\begin{proof}
    The proof is by induction on the derivation of $\Gamma \vdash t : \tau$. The only non-trivial cases are when $t$ builds a new probabilistic program of type $\tau = G~\tau'$, and the only non-trivial conditions to prove are discrete structure and absolute continuity of trace distribution. 
    
    We consider these non-trivial cases in turn:

    \begin{itemize}[leftmargin=*]
    \item $t = \mathbf{return}~t'$: in this case, $\sem{t}_1(\gamma)$ is the Dirac measure on the empty trace, $\delta_{\{\}}$. This measure is clearly absolutely continuous with respect to $\mathcal{B}_\mathbb{T}$ (it has density $1$ at the point $\{\}$, and $0$ elsewhere), and is also clearly discrete-structured (vacuously, as there are no two distinct traces to which it assigns positive density).

    \item $t = \textbf{observe}~t_1~t_2$: in this case, $\sem{t}_1(\gamma)$ is a rescaled Dirac measure, $w \odot \delta_{\{\}}$. By the same argument as above, it is absolutely continuous with respect to $\mathcal{B}_\mathbb{T}$ (its density is $w$ at $\{\}$ and $0$ elsewhere), and discrete-structured (again, there are no two distinct traces to which it assigns positive density).

    \item $t = \textbf{sample}~t_1~t_2$: in this case, $\sem{t}_1(\gamma)$ is the pushforward of a measure $\mu = \sem{t_1}(\gamma)$ on $\tau'$ by the map $\lambda x. \{\sem{t_2}(\gamma) \mapsto x\}$. Because $\mu$ is (by the inductive hypothesis) absolutely continuous with respect to the base measure on $\tau'$, our new measure $\sem{t}_1(\gamma)$ is absolutely continuous with respect to $\mathcal{B}_\mathbb{T}$, with density $\frac{d\mu}{d\mathcal{B}_{\tau'}}(x)$ for traces of the form $\{\sem{t_2}(\gamma) \mapsto x\}$, and density 0 for all other traces. It is also discrete-structured, because any two traces with positive density must have the form $\{\sem{t_2}(\gamma) \mapsto x\}$ for distinct values of $x$\textemdash so $\sem{t_2}(\gamma)$ is an address on which they are guaranteed to disagree.

    \item $t = \haskdo \{x \gets t_1; m\}$: by induction, $\sem{t_1}_1(\gamma)$ is absolutely continuous with respect to $\mathcal{B}_\mathbb{T}$, and so is $\sem{\haskdo \{m\}}_1(\gamma[x \mapsto v])$, for all $v \in \sem{\tau_1}$ (where $G~\tau_1$ is the type of $t_1$). Thus, the joint distribution is also absolutely continuous with respect to $\mathcal{B}_\mathbb{T}$, with density equal to the product of the component densities. For discrete structure, let $u$ and $u'$ be distinct traces in the support of $\sem{t}_1(\gamma)$. Let $(u_1, u_2) = \text{split}_{\sem{t_1}_1(\gamma)}(u)$, and likewise let $(u_1', u_2') = \text{split}_{\sem{t_1}_1(\gamma)}(u')$. If $u_1 \neq u_1'$, then since they are both in the support of $\sem{t_1}_1(\gamma)$, by induction they must share an address to which they assign different values, an address on which $u$ and $u'$ would therefore also disagree. If $u_1 = u_1'$ and $u_2 \neq u_2'$, then since both $u_2$ and $u_2'$ are in the support of $\sem{\haskdo \{m\}}(\gamma[x \mapsto \sem{t_1}_2(\gamma)(u_1)])$ (which is discrete-structured by induction), there must be a shared address on which they disagree. This would again constitute a shared address on which $u$ and $u'$ disagree.
    \end{itemize}
\end{proof}

\subsection{Proof details for Lemma~\ref{lem:fundamental-xi}}
\begin{proof}
Lemma~\ref{lem:fundamental-xi} is an instance of the fundamental lemma of logical relations, for the particular logical relations $\mathcal{R}^\xi$. 

We consider a term $\Gamma \vdash t : \tau^*$ and an arbitrary pair of environments $(\gamma, \gamma') \in \mathcal{R}_\Gamma^\xi$. Note that when $t = x$, $t = (t_1, t_2)$, $t = \pi_i~t'$, $t = \lambda x. t'$, $t = t_1~t_2$, or $t = \mathbf{if}~t_1~\mathbf{then}~t_1~\mathbf{else}~t_2$, we can apply the standard logical relations arguments, because our semantics and logical relations for product types, function types, Booleans, and variables are completely standard, and the macro $\xi$ treats these constructs functorially (e.g. $\xi\{(t_1, t_2)\}=(\xi\{t_1\}, \xi\{t_2\})$). We establish the remaining cases in turn.

\begin{itemize}[leftmargin=*]
\item If $\tau^*$ is some base type $\sigma$ and $t = c$ (a constant of type $\sigma$), then $\xi\{t\} = t$ and it follows immediately that $(\sem{t}(\gamma), \sem{\xi\{t\}}(\gamma')) \in \mathcal{R}_\sigma^\xi$ (which is just the equality relation).

\item If $t = \mathbf{normal}_\text{REPARAM}$, $t = \mathbf{normal}_\text{REINFORCE}$, or $t = \mathbf{flip}_\text{MVD}$, then no matter what $\gamma$ and $\gamma'$ are, $\sem{t}(\gamma)$ is a primitive probability kernel (e.g., the normal distribution or the Bernoulli distribution), and $\sem{\xi\{t\}}(\gamma')$ is a primitive density function. We discharge this case by verifying that the \textit{primitive} density functions are correctly implemented.

\item The key cases of interest are the terms of type $\tau^* = G~\tau$. For each, we need to establish several key properties of $\xi\{t\}$, outlined in the definition of $\mathcal{R}_{G~\tau}^\xi$. In each of the cases below, let $(\mu, f) = \sem{t}(\gamma)$, let $u$ be a trace, and let $(u_1, u_2) = \textit{split}_{\mu}(u)$. Further, let $g = \sem{\xi\{t\}}(\gamma')$. We write $g_i$ (for $i$ in $\{1, 2, 3\}$) as shorthand for $\pi_i \circ g$. Our job is to show that in each case below, $(f(u), g_1(u)) \in R_{\tau}^\xi$, that $g_2(u) = \frac{d\mu}{d\mathcal{B}_\mathbb{T}}(u_1)$, and that $g_3(u) = u_2$.

\begin{itemize}[leftmargin=*]
\item $t = \mathbf{return}~t'$: From the semantics of $\lambda_\text{Gen}$, we have $\mu = \delta_{\{\}}$ and $f = \lambda \_. \sem{t'}(\gamma)$. Then for any $u$, we have that $(u_1, u_2) = (\{\}, u)$. We now establish each requirement:
\begin{enumerate}[leftmargin=*]
    \item From the definition of $\xi\{\mathbf{return}~t'\}$, we have $g_1(u) = \sem{\xi\{t'\}}(\gamma')$, so by the inductive hypothesis for $\Gamma \vdash t' : \tau'$, $(f(u), g_1(u)) \in R_{\tau}^\xi$. 
    \item Also from the definition of $\xi\{\mathbf{return}~t'\}$, we have $g_2(u) = 1$ for all $u$, which, as required, is equal to $\frac{d\mu}{d\mathcal{B}_\mathbb{T}}(u_1)$ (since $u_1 = \{\}$ and $\mu = \delta_{\{\}}$). 
    \item Finally, we can read off the definition of $\xi\{\mathbf{return}~t'\}$ that $g_3(u) = u$ for all $u$, which is equal to $u_2$, as required.
\end{enumerate}
\item $t = \mathbf{observe}~t_1~t_2$: Write $\sigma$ for the type of $t_2$, and let $m = \sem{t_1}(\gamma)$ and $v = \sem{t_2}(\gamma)$. By induction, $\sem{\xi\{t_1\}\}}(\gamma') = \frac{dm}{d\mathcal{B}_\sigma}$, and $\sem{\xi\{t_2\}}(\gamma')=v$. From the semantics, we see that $\mu = \frac{dm}{d\mathcal{B}_\sigma}(v) \odot \delta_{\{\}}$, and that $f = \lambda \_. ()$. As in the previous case, $(u_1, u_2) = (\{\}, u)$.
\begin{enumerate}[leftmargin=*]
    \item We have $g_1(u) = ()$, and so $(f(u), g_1(u)) = ((), ()) \in R_1^\xi$.
    \item We have $g_2(u) = \sem{\xi\{t_1\}}(\gamma')(\sem{\xi\{t_2\}}(\gamma')) = \frac{dm}{d\mathcal{B}_\sigma}(v)$. But this is precisely the density of $\mu$ with respect to $\mathcal{B}_\mathbb{T}$ at the point $u_1 = \{\}$, because $\mu = \frac{dm}{d\mathcal{B}_\sigma}(v) \odot \delta_{\{\}}$.
    \item Finally, we have $g_3(u) = u$, which is equal to $u_2$, as required.
\end{enumerate}

    \item $t = \mathbf{sample}~t_1~t_2$: In this case $\tau$ is some base type $\sigma$. Let $m = \sem{t_1}(\gamma)$ and $k = \sem{t_2}(\gamma)$, and note that by induction, $\sem{\xi\{t_1\}}(\gamma') = \frac{dm}{d\mathcal{B}_\sigma}$ and $\sem{\xi\{t_2\}}(\gamma') = k$. From the semantics, we have that $\mu$ is the pushforward of $m$ by $v \mapsto \{k \mapsto v\}$. We consider two cases. First, suppose the trace $u$ does not have key $k$. Then $f(u)$ is equal to a default value of type $\sigma$, and $(u_1, u_2) = (u, \{\})$. In this case, $\textit{pop}~u~k$ returns the default value of type $\sigma$, the weight $0$, and the remainder $\{\}$, and we reason as follows:
    \begin{enumerate}
        \item $g_1(u)$ is the default value of type $\sigma$, which is equal to $f(u)$, so $(f(u), g_1(u)) \in R_\sigma^\xi$. 
        \item $g_2(u)$ is 0, which is the density under $\mu$ of $u_1 = u$, because $u_1$ does not have key $k$, but $\mu$ assigns positive density only to traces that do have key $k$.
        \item $g_3(u) = \{\} = u_2$, as required.
    \end{enumerate}
    Alternatively, suppose the trace $u$ does have key $k$, assigned to value $v$. In this case, $f(u) = v$, and $(u_1, u_2) = (\{k \mapsto v\}, u - \{k \mapsto v\})$, and the result of $\textit{pop}~u~k$ is $(v, 1, u_2)$. In this case:
    \begin{enumerate}
        \item $g_1(u) = v = f(u)$, and so $(f(u), g_1(u)) \in R_\sigma^\xi$.
        \item $g_2(u) = \sem{\xi\{t_1\}}(\gamma')(v) =  \frac{dm}{d\mathcal{B}_\sigma}(v) =  \frac{d\mu}{d\mathcal{B}_\mathbb{T}}(\{k \mapsto v\}) = \frac{d\mu}{d\mathcal{B}_\mathbb{T}}(u_1)$, as required.
        \item $g_3(u) = u - \{k \mapsto v\} = u_2$, as required.
    \end{enumerate}

    \item $t = \haskdo \{x \gets t'; m\}$: Let $\tau'$ be the type of the intermediate result, i.e., the type such that $t'$ has type $G~\tau'$. Let $(\mu', f') = \sem{t'}(\gamma)$, let $k = \lambda x. \sem{\haskdo\{m\}}_1(\gamma)$, and let $f_k(v) = \sem{\haskdo\{m\}}_2(\gamma[x \mapsto v])$. Then $\mu(U) = \iint \textit{disj}(r_1, r_2) \cdot \delta_{r_1 \mdoubleplus r_2}(U) \mu'(dr_1) k(f'(r_1), dr_2)$. %, and $f(u) = f_k(f'(u))(u')$.

    We now consider several possible scenarios. 
    
    \begin{itemize}[leftmargin=*]
        \item First, suppose no subtrace of $u$ is in the support of $\mu'$, so that $(u_1, u_2) = (u, \{\})$. Then by induction, $\sem{\xi\{t'\}}(\gamma')(u)$ returns $(v, 0, \{\})$ for some value $v$ such that $(f'(u), v) \in R_{\tau'}^\xi$. Because $(f'(u), v) \in R_{\tau'}^\xi$, we have that $(\gamma[x \mapsto f'(u)], \gamma'[x \mapsto v]) \in R_{\Gamma, x : \tau'}^\xi$. We can then apply the inductive hypothesis for $\Gamma, x : \tau' \vdash \haskdo \{m\} : G~\tau$, with these two environments, to obtain that $\sem{\xi\{\haskdo\{m\}\}}(\gamma'[x \mapsto v])(\{\}) = (y, w, \{\})$, for some weight $w$ and some value $y$ such that $(f_k(f'(u))(\{\}), y) \in R_\tau^\xi$.  We can then read off the definition of $\xi\{\haskdo\{x \gets t'; m\}\}$ that $g_1(u) = y$, $g_2(u) = 0 \cdot w = 0$, and $g_3(u) = \{\}$. Now consider each requirement:
        \begin{enumerate}[leftmargin=*]
            \item From the semantics of $\haskdo$, we have that $f(u) = f_k(f'(u))(\{\})$. We also have from above that $g_1(u) = y$, which is related at $R_\tau^\xi$ to $f_k(f'(u))(\{\})$, as required.
            \item We have $g_2(u) = 0$, which is (as required) the density of $\mu$ at $u$ (since $u$ is not in $\mu$'s support).
            \item We have $g_3(u) = \{\} = u_2$, as required.
        \end{enumerate}

        \item Second, suppose there is a subtrace in the support of $\mu'$, but that, writing $(u'_1, u'_2) = \textit{split}_{\mu'}(u)$, there is no subtrace of $u'_2$ in the support of $k(f'(u))$. Then $(u_1, u_2) = (u, \{\})$ once again. By induction, $\sem{\xi\{t'\}}(\gamma')(u) = (v, w, u'_2)$, for some $v$ such that $(f'(u), v) \in R_{\tau'}^\xi$ and some $w > 0$. Now, since no subtrace of $u'_2$ is in the support of $k(f'(u))$, $\textit{split}_{k(f(u))}(u'_2) = (u'_2, \{\})$. Also, since $(f'(u), v) \in R_{\tau'}^\xi$, we can construct extended environments $(\gamma[x \mapsto f'(u)], \gamma'[x \mapsto v]) \in R_{\Gamma, x : \tau'}^\xi$. Then, by the inductive hypothesis for $\Gamma, x : \tau' \vdash \haskdo \{m\} : G~\tau$, we have that $\sem{\xi\{\haskdo\{m\}\}}(\gamma'[x \mapsto v])(u'_2) = (y, 0, \{\})$ such that $(f_k(f'(u))(u'_2), y) \in R_\tau^\xi$. Then $g_1(u) = y$, $g_2(u) = w \cdot 0 = 0$, and $g_3(u) = \{\}$. We can again consider each requirement, using the same arguments as in the first case:
        \begin{enumerate}[leftmargin=*]
            \item From the semantics of $\haskdo$, we have that $f(u) = f_k(f'(u))(u'_2)$. We also have from above that $g_1(u) = y$, which is related at $R_\tau^\xi$ to $f_k(f'(u))(\{\})$, as required.
            \item We have $g_2(u) = 0$, which is (as required) the density of $\mu$ at $u$ (since $u$ is not in $\mu$'s support).
            \item We have $g_3(u) = \{\} = u_2$, as required.
        \end{enumerate}

        \item Finally, suppose that $u$ has a subtrace in the support of $\mu$ (in which case that subtrace must be the one computed by $\textit{split}_\mu$, namely $u_1$). For this to be the case, $u_1$ must decompose as $u'_1 \mdoubleplus u'_2$, where $(u'_1, u'_2) = \textit{split}_{\mu'}(u_1)$. Furthermore, we have that $\textit{split}_{\mu'}(u) = (u'_1, u'_2 \mdoubleplus u_2)$. By induction on $\Gamma \vdash t' : G~\tau'$, we have that $\sem{\xi\{t'\}}(\gamma')(u) = (v, \frac{d\mu'}{d\mathcal{B}_\mathbb{T}}(u'_1), u'_2 \mdoubleplus u_2)$ for some $v$ satisfying $(f'(u), v) \in R_{\tau'}^\xi$.  We can construct extended environments $(\gamma[x \mapsto f'(u)], \gamma'[x \mapsto v]) \in R_{\Gamma, x : \tau'}^\xi$, and apply the inductive hypothesis for $\Gamma, x : \tau' \vdash \haskdo \{m\} : G~\tau$, to obtain that $\sem{\xi\{\haskdo\{m\}\}}(\gamma'[x \mapsto v])(u'_2 \mdoubleplus u_2) = (y, \frac{dk(f'(u))}{d\mathcal{B}_\mathbb{T}}(u'_2), u_2)$ for some $y$ such that $(f_k(f'(u))(u'_2 \mdoubleplus u_2), y) \in R_\tau^\xi$. Then $g_1(u) = y$, $g_2(u) = \frac{d\mu'}{d\mathcal{B}_\mathbb{T}}(u'_1) \cdot \frac{dk(f'(u))}{d\mathcal{B}_\mathbb{T}}(u'_2)$, and $g_3(u) = u_2$. We now verify the requirements:
        \begin{enumerate}[leftmargin=*]
            \item From the semantics of $\haskdo$, we have that $f(u) = f_k(f'(u))(u'_2 \mdoubleplus u_2)$. We also have from above that $g_1(u) = y$, which is related at $R_\tau^\xi$ to $f_k(f'(u))(u'_2 \mdoubleplus u_2)$, as required.
            \item We have $g_2(u) = \frac{d\mu'}{d\mathcal{B}_\mathbb{T}}(u'_1) \cdot \frac{dk(f'(u))}{d\mathcal{B}_\mathbb{T}}(u'_2) = \frac{d\mu}{d\mathcal{B}_\mathbb{T}}(u'_1 \mdoubleplus u'_2)$, which is (as required) the density of $\mu$ at $u_1 = u'_1 \mdoubleplus u'_2$.
            \item We have $g_3(u) = u_2$, as required.
        \end{enumerate}
    \end{itemize} 
\end{itemize}

\end{itemize}
\end{proof}

\subsection{Proof details for Lemma~\ref{lem:fundamental-chi}}
\begin{proof}
Lemma~\ref{lem:fundamental-chi} is an instance of the fundamental lemma of logical relations, for the particular logical relations $\mathcal{R}^\chi$. 

We consider a term $\Gamma \vdash t : \tau^*$ and an arbitrary pair of environments $(\gamma, \gamma') \in \mathcal{R}_\Gamma^\chi$. Note that when $t = x$, $t = (t_1, t_2)$, $t = \pi_i~t'$, $t = \lambda x. t'$, $t = t_1~t_2$, or $t = \mathbf{if}~t_1~\mathbf{then}~t_1~\mathbf{else}~t_2$, we can again apply the standard logical relations arguments, because our semantics and logical relations for product types, function types, Booleans, and variables are completely standard, and the macro $\chi$ treats these constructs functorially (e.g. $\chi\{(t_1, t_2)\}=(\chi\{t_1\}, \chi\{t_2\})$). We establish the remaining cases in turn.

\begin{itemize}[leftmargin=*]
\item If $\tau^*$ is some base type $\sigma$ and $t = c$ (a constant of type $\sigma$), then $\xi\{t\} = t$ and it follows immediately that $(\sem{t}(\gamma), \sem{\chi\{t\}}(\gamma')) \in \mathcal{R}_\sigma^\chi$ (which is just the equality relation).

\item If $t = \mathbf{normal}_\text{REPARAM}$, $t = \mathbf{normal}_\text{REINFORCE}$, or $t = \mathbf{flip}_\text{MVD}$, then no matter what $\gamma$ and $\gamma'$ are, $\sem{t}(\gamma)$ is a primitive probability kernel (e.g., the normal distribution or the Bernoulli distribution), and $\sem{\chi\{t\}}(\gamma')$ has two components: a primitive simulator for the distribution, simulating an outcome and evaluating its density, and a deterministic density function. We discharge this case by verifying that the \textit{primitive} simulators generate from the correct primitive distributions and evaluate the corresponding primitive densities.

\item The key cases of interest are the terms of type $\tau^* = G~\tau$. For each, we need to establish several key properties of $\chi\{t\}$, outlined in the definition of $\mathcal{R}_{G~\tau}^\chi$. In each of the cases below, let $(\mu, f) = \sem{t}(\gamma)$, let $\nu = \sem{\chi\{t\}}(\gamma')$, and let $u$ be a trace in the support of $\mu$. In each case below, we first establish that $\nu$ is the pushforward of $\mu$ by a map $h$; writing $h_i$ for $\pi_i \circ h$, we then show that $(f(u), h_1(u)) \in R_\tau^\chi$, $h_2(u) = \frac{d\mu}{d\mathcal{B}_\mathbb{T}}(u)$, and $h_3(u) = u$.

\begin{itemize}[leftmargin=*]
    \item $t = \mathbf{return}~t'$: We have $\mu = \delta_{\{\}}$ and $f = \lambda \_. \sem{t'}(\gamma)$. Examining the rule for $\chi$ on this expression, we see that $\nu$ is a pushforward of $\mu$ by $h(u) = (\sem{\chi\{t'\}}(\gamma'), 1, \{\})$. Note also that the only trace in the support of $\mu$ is $\{\}$, so $u$ must be $\{\}$. We establish the three required properties:
    \begin{enumerate}[leftmargin=*]
        \item We have $h_1(u) = \sem{\chi\{t'\}}(\gamma')$, which is related to $f(u) = \sem{t'}(\gamma)$ by the inductive hypothesis for $\Gamma \vdash t : \tau$.
        \item We have $h_2(u) = 1 = \frac{d\mu}{d\mathcal{B}_\mathbb{T}}(u)$ (since $u = \{\}$), as required.
        \item We have $h_3(u) = \{\} = u$, as required.
    \end{enumerate}
    \item $t = \mathbf{observe}~t_1~t_2$: Write $\sigma$ for the type of $t_2$, and let $m = \sem{t_1}(\gamma)$ and $v = \sem{t_2}(\gamma)$. By induction, $\pi_2(\sem{\chi\{t_1\}\}}(\gamma'))$ is $\frac{dm}{d\mathcal{B}_\sigma}$, and $\sem{\chi\{t_2\}}(\gamma')=v$. From the semantics, we see that $\mu = \frac{dm}{d\mathcal{B}_\sigma}(v) \odot \delta_{\{\}}$, and that $f = \lambda \_. ()$. From the definition of $\chi$, we see that the generated code first computes the deterministic constant $w = \pi_1(\sem{\chi\{t_1\}}(\gamma'))(\sem{\chi\{t_2\}}(\gamma')) = \frac{dm}{d\mathcal{B}_\sigma}(v)$. Then we can see that the code implements the pushforward of $\mu = w \odot \delta_{\{\}}$ by $h(u) = ((), w, \{\})$. Again, we know that $u$ must be $\{\}$, the only trace in the support of $\mu$. We then verify:
    \begin{enumerate}[leftmargin=*]
        \item We have $h_1(u) = ()$, which is related to $f(u) = ()$ in $R_1^\chi$.
        \item We have $h_2(u) = w = \frac{d\mu}{d\mathcal{B}_\mathbb{T}}(u)$ (since $u = \{\}$), as required.
        \item We have $h_3(u) = \{\} = u$, as required.
    \end{enumerate}
    \item $t = \mathbf{sample}~t_1~t_2$: Let $m = \sem{t_1}(\gamma)$ and $k = \sem{t_2}(\gamma) = \sem{\chi\{t_2\}}(\gamma')$ (where the last equality holds by the inductive hypothesis for $\Gamma \vdash t_2 : \str$). We have that $\mu$ is the pushforward of $m$ by the map $\lambda x. \{k \mapsto x\}$, and that $f(u) = x$ when $u$ is of the form $\{k \mapsto x\}$. We also have, by the inductive hypothesis for $\Gamma \vdash t_1 : D~\tau$, that $\pi_1(\sem{\chi\{t_1\}}(\gamma'))$ is the pushforward of $m$ by $\lambda x. (x, \frac{dm}{d\mathcal{B}_\tau}(x))$. Then, from the definition of $\chi$, we have that $\sem{\chi\{t\}}(\gamma')$ is the pushforward of $m$ by $\lambda x. (x, \frac{dm}{d\mathcal{B}_\tau}(x), \{k \mapsto x\})$, which is equivalently the pushforward of $\mu$ by $h = \lambda u. (f(u), \frac{dm}{d\mathcal{B}_\tau}(f(u)), u)$. Since $u$ is in the support of $\mu$, $u$ must have the form $\{k \mapsto x\}$ for some $x$. We verify:
    \begin{enumerate}
        \item We have $h_1(u) = f(u) = x$, and $(f(u), x) = (x, x) \in R_\tau^\chi$, because $\tau$ is a base type so $R_\tau^\chi$ is the equality relation.
        \item We have $h_2(u) = \frac{dm}{d\mathcal{B}_\tau}(x) = \frac{d\mu}{d\mathcal{B}_\mathbb{T}}(\{k \mapsto x\})$, as required.
        \item We have $h_3(u) = u$, as required.
    \end{enumerate}
    \item $t = \haskdo \{x \gets t'; m\}$: Let $(\mu', f') = \sem{t'}(\gamma)$, $k = \lambda v. \pi_1(\sem{\haskdo\{m\}}(\gamma[x \mapsto v]))$, and $f_k = \lambda v. \pi_2(\sem{\haskdo\{m\}}(\gamma[x \mapsto v]))$.
    By the inductive hypotheses for $\Gamma \vdash t' : G~\tau'$, we have that $\sem{\chi\{t'\}}(\gamma')$ is the pushforward of $\mu'$ by a map $h'(u')$, such that $(f'(u'), h'_1(u')) \in R_{\tau'}^\chi$, $h'_2(u') = \frac{d\mu'}{d\mathcal{B}_\mathbb{T}}(u')$, and $h'_3(u') = u'$. Furthermore, we have that for any $u'$ in the support of $\mu'$, $\sem{\chi\{\haskdo\{m\}\}}(\gamma'[x \mapsto h_1(u')])$ is the pushforward of $k(f'(u'))$ by a function $h^{u'}$, such that for all $u_k$ in the support of $k(f'(u'))$, $(f_k(f'(u'))(u_k), h^{u'}_1(u_k)) \in R_\tau^\chi$, $h^{u'}_2(u_k) = \frac{dk(f'(u'))}{d\mathcal{B}_\mathbb{T}}(u_k)$, and $h_3^{u'}(u_k) = u_k$. Examining $\chi$'s definition for $\haskdo$, we see that $\nu$ is the pushforward of $\mu$ by $h(u) = (h_1^{u'}(u_k), h_2^{u'}(u_k), h_3^{u'}(u_k))$, where $(u', u_k) = \textit{split}_{\mu'}(u)$. We verify:
    \begin{enumerate}
        \item We have $h_1(u) = h_1^{u'}(u_k)$, which we saw above is related by $R_\tau^\chi$ to $f(u) = f_k(f'(u'))(u_k)$.
        \item We have $h_2(u) = h_2^{u'}(u_k)$, which, from above, is equal to $\frac{d\mu'}{d\mathcal{B}_\mathbb{T}}(u') \cdot \frac{dk(f'(u'))}{d\mathcal{B}_\mathbb{T}}(u_k)$. This in turn is equal to $\frac{d\mu}{d\mathcal{B}_\mathbb{T}}(u)$, for $u$ in the support of $\mu$.
        \item We have $h_3(u) = h'_3(u') \mdoubleplus h^{u'}_3(u_k) = u' \mdoubleplus u_k = u$, as required.
    \end{enumerate}
\end{itemize}

\end{itemize}
\end{proof}

\section{Adding New Primitives with Custom Gradient Estimators}
\label{appdx:extensibility}
In this appendix, we explain the process of extending our system with a new stochastic primitive, with its own custom gradient estimation strategy. In particular, we discuss the local proof obligations that the implementer incurs, to ensure that their new primitives do not compromise the unbiasedness results in our system.

As our illustrative example, we consider $\mathbf{poisson}_{\text{REINFORCE}} : \RR_{>0} \to D~\mathbb{N}$, a version of the Poisson distribution that uses REINFORCE to estimate gradients. To add the new primitive, the user must write three pieces of code, each with its own proof obligation:
\begin{itemize}[leftmargin=*]
    \item The \textit{density function}: we must define $\xi\{\mathbf{poisson}_\text{REINFORCE}\}$, an implementation (in $\lambda_\text{ADEV}$) of the density function for the new primitive. The proof of Lemma~\ref{lem:fundamental-xi} must be extended with a new case, showing that for any argument $\textit{rate} \in \RR_{>0}$, $\sem{\mathbf{poisson}_\text{REINFORCE}}(\textit{rate})$ has density $\sem{\xi\{\mathbf{poisson}_\text{REINFORCE}\}}(\textit{rate})$ with respect to the base measure $\mathcal{B}_\mathbb{N}$. 

    \item The \textit{simulator}: we must define $\chi\{\mathbf{poisson}_\text{REINFORCE}\}$, an implementation (in $\lambda_\text{ADEV}$) of the simulator for the new primitive. The proof of Lemma~\ref{lem:fundamental-chi} must be extended with a new case, showing that for any argument $\textit{rate} \in \RR_{>0}$, $\pi_1(\sem{\chi\{\mathbf{poisson}_\text{REINFORCE}\}}(\textit{rate}))$ is the pushforward of $\sem{\mathbf{poisson}_\text{REINFORCE}}(\textit{rate})$ by $\lambda n. (n, \rho(n))$, where $\rho$ is the density of $\sem{\mathbf{poisson}_\text{REINFORCE}}(\textit{rate})$ with respect to $\mathcal{B}_\mathbb{N}$. Note that the simulator may be implemented in terms of existing $\lambda_\text{ADEV}$ primitives, or may itself use $\mathbf{poisson}_\text{REINFORCE}$. (In our actual system, every $\lambda_\text{ADEV}$ primitive distribution must also provide a more basic piece of code that actually uses a PRNG to simulate from the distribution in question, but this is not part of our formal development.)

    \item The \textit{gradient estimator}: finally, we must define $D\{\mathbf{poisson}_\text{REINFORCE}\}$, an implementation (in $\lambda_\text{ADEV}$) of a gradient estimator for the new primitive. The proof of Theorem~\ref{thm:unbiased-adev} must then be extended with a new case. In particular, we must prove the following:

    \textit{
    Let $h : \RR \to \RR$ be differentiable, and let $f : \RR \to \mathbb{N} \to \RR$ be such that $f(-, n)$ is differentiable in its first argument for all $n$. Further let $\tilde{f} : (\RR \times \RR) \times \mathbb{N} \to \text{Prob} (\RR \times \RR)$, such that $\mathbb{E}_{(x, y) \sim \tilde{f}((h(\theta), h'(\theta)), n)}[x] = f(h(\theta), n)$ and $\mathbb{E}_{(x, y) \sim \tilde{f}((h(\theta), h'(\theta)), n)}[y] = \frac{d}{d\theta} f((h(\theta)), n)$. Then, writing $\mu(\theta)$ for $\pi_1(\sem{D\{\mathbf{poisson}_\text{REINFORCE}\}}((h(\theta), h'(\theta)), \tilde{f}(h(\theta), h'(\theta))))$, we must show that $\mathbb{E}_{(x,y) \sim \mu(\theta)}[x] = \mathbb{E}_{n \sim \textbf{poisson}(h(\theta))}[f(h(\theta), n)]$, and $\mathbb{E}_{(x,y) \sim \mu(\theta)}[y] = \frac{d}{d\theta}\mathbb{E}_{n \sim \textbf{poisson}(h(\theta))}[f(h(\theta), n)]$.
    }
\end{itemize}

\end{document}